\documentclass[a4,11pt,pcxf,epsf,amssymb]{article}

\usepackage[OT1]{fontenc}
\usepackage[colorlinks,citecolor=blue,urlcolor=blue]{hyperref}
\usepackage{amsmath,amsthm,amssymb}
\usepackage{mathrsfs}
\usepackage{graphicx}
\usepackage{subfig}
\usepackage{multirow}
\usepackage{placeins,afterpage}
\usepackage{lineno}
\modulolinenumbers[5]
\usepackage{chicago}
\usepackage{rotating} % to make sideway tables/figures
\usepackage{pdflscape}

\usepackage{bm}
\theoremstyle{plain}
\newtheorem{thm}{\protect\theoremname}
\theoremstyle{remark}

\theoremstyle{definition}
\newtheorem{defn}{\protect\definitionname}
\theoremstyle{plain}
\newtheorem{lem}{\protect\lemmaname}

\providecommand{\definitionname}{Definition}
\providecommand{\lemmaname}{Lemma}
\providecommand{\remarkname}{Remark}
\providecommand{\theoremname}{Theorem}

\global\long\def\comp#1{#1^{\mathsf{c}}}

\begin{document}

\title{Constructing likelihood functions for interval-valued random variables}

\author{X. Zhang\footnote{School of Mathematics and Statistics, University of New South Wales, Sydney, Australia},\:\,
B. Beranger$^*$\:
and
S. A. Sisson$^*$\footnote{Corresponding Author: Email \tt{Scott.Sisson@unsw.edu.au}}}
\date{}
\maketitle

\baselineskip1.9em

\begin{abstract}
There is a growing need for flexible methods to analyse interval-valued data, which can provide 
efficient data representations for very large datasets.
However, existing descriptive frameworks to achieve this ignore the process by which 
interval-valued data are typically constructed; namely by the aggregation of real-valued 
data generated from some underlying process.
In this article we develop the foundations of likelihood based statistical inference for 
intervals that directly incorporates the underlying data generating procedure into the
analysis. That is, it permits the direct fitting of models for the underlying real-valued 
data given only the interval-valued summaries.
This generative approach overcomes several problems associated with existing methods, 
including the rarely satisfied assumption of within-interval uniformity. The new methods
are illustrated by simulated and real data analyses.
\end{abstract}

\noindent Keyords:  Aggregate data; Interval-valued data; Likelihood theory; Symbolic data analysis.

%%%%%%%%%%%%%%%%%%%%%%%%%%%%%%%%%%%%%%%%%%%%%%%%%
%%%%%%%%%%%%%%%%%%%%%%%%%%%%%%%%%%%%%%%%%%%%%%%%%
\section{Introduction}
%%%%%%%%%%%%%%%%%%%%%%%%%%%%%%%%%%%%%%%%%%%%%%%%%
%%%%%%%%%%%%%%%%%%%%%%%%%%%%%%%%%%%%%%%%%%%%%%%%%

As we move inevitably towards a more data-centric society, there is a growing need for 
the ability to analyse data that are constructed in non-standard forms, rather than 
represented as continuous points in $\mathbb{R}^p$ \cite{Billard2003}. The simplest and
most popular of these is interval-valued data.

Interval-valued observations can arise naturally through the data recording process, and 
essentially result as a way to characterise measurement error or uncertainty of an 
observation. Examples include blood pressure, which is typically reported as an interval 
due to the inherent continual changes within an individual \cite{Billard2006}; data 
quantisation, such as rounding or truncation, which results in observations being known 
to lie within some interval  \cite{mclachlan+j88,vardeman+l05}; and the expression of 
expert-elicited intervals that contain some quantity of interest \shortcite{lin+cs17,fisher+olmkbc15}, 
among others.

The use of intervals as a summary representation of a collection of classical real-valued 
data is also rapidly gaining traction.
Here the aggregation of a large and complex dataset into a smaller collection of 
suitably constructed intervals can enable a statistical analysis that would otherwise be 
computationally unviable \cite{Billard2003}. Where interest in the outcome of an analysis 
exists at the level of a group, rather than at an individual level, interval-valued data 
provide a convenient group-level aggregation device 
\cite{limaneto+c10,Noirhomme-Fraiture2011}. 
Similarly, aggregation of individual observations within an interval structure allows for 
some preservation of privacy for the individual \shortcite{domingues+dc10}.

The earliest systematic study of interval-valued data is in numerical analysis, where
\citeN{Moore1966} used intervals as a description for imprecise data.
Random intervals are also special cases of random sets \cite{Molchanov2005},
the theory of which brings together elements of topology, convex geometry, and 
probability theory to develop a coherent mathematical framework for their analysis. 
\citeN{Matheron1975} gave the first self-contained development of statistical models for 
random sets, including central limit theorems and a law of large numbers, 
and \citeN{Beresteanu2008} 
derived these limit theorems specifically for random intervals.
In this framework, interval-valued random variables 
$[X]=\left[\underline{X},\overline{X}\right]\subset\mathbb{R}$ 
are modelled as a bivariate real-valued random vector 
$\left(\underline{X},\overline{X}\right)$, 
where $\underline{X}\leq\overline{X}$, using standard inferential techniques.
This approach is also used for partially identified models, 
where the object of economic and statistical interest is a set rather than a point 
\shortcite{Beresteanu2012,Molchanov2014}. 
In probabilistic modelling, \citeN{Lyashenko1983} introduced normal random compact convex
sets in Euclidean space,
and showed that a normal random interval is simply a Gaussian displacement of a fixed
closed bounded interval. \citeN{Sun2014} subsequently extended this idea to normal 
hierarchical models for random intervals.

A more popular framework for the analysis of interval-valued data, 
and one which we focus on here, is symbolic data analysis \cite{Billard2006}.
Symbols can be considered as distributions of real-valued data points in $\mathbb{R}^{p}$,
such as intervals and histograms, or more general structures including lists.
They are typically constructed as the aggregation into summary form of real-valued
data within some group, and so the symbol is interpreted as taking values as described by 
the summary distribution.
As a result, symbols have internal variations and structures which do not exist in 
real-valued data, and methods for analysing them must account for within-symbol variation 
in addition to between symbol variation.
In practice, the most common form of symbol is the interval or its $p$-dimensional 
extension, the $p$-hyper-rectangle.
See \citeN{Billard2003}, \citeN{Billard2006} and \citeN{Noirhomme-Fraiture2011} for a review of recent results.

While many exploratory and descriptive data analysis techniques for symbolic data have 
been developed (see e.g.~\citeNP{Billard2006} for an overview), there is a paucity of 
results for developing a robust statistical inferential framework for these data. The most 
significant of these \cite{Le-Rademacher2011} maps the parameterisation of the symbol 
into a real-valued random vector, and then uses the standard likelihood framework to 
specify a suitable model. In the random interval setting, this is equivalent to the 
random set theory approach, which models the interval-valued random variables
$[X]=\left[\underline{X},\overline{X}\right]\subset\mathbb{R}$ 
by the constrained real-valued random vector $\left(\underline{X},\overline{X}\right)\in\mathbb{R}^2$ or, 
more commonly, a reparameterisation  to the unconstrained interval centre and half-range 
$(X_c,X_r)=\left((\underline{X}+\overline{X})/2,(\overline{X}-\underline{X})/2\right)$, 
which is then more easily modelled, e.g.~$(X_c,\log X_r)\sim N_2(\mu,\Sigma)$. This 
likelihood framework has been used for the analysis of variance \cite{Brito2012}, time 
series forecasting \shortcite{Arroyo2010} and interval-based regression models \cite{xu10} 
among others.

While sensible, by nature the above methods for modelling 
real-valued random variables only permit 
descriptive modelling at the level of the real-valued random 
vector $(\underline{X},\overline{X})$ 
(or its equivalent for $p$-hyper-rectangles).  However this descriptive approach completely 
ignores the data generating procedure commonly assumed and implemented for the construction of 
observed intervals; namely that the underlying real-valued data are produced from some 
data generating model $X_1,\ldots,X_m\sim f(x_1,\ldots,x_m|\alpha)$, and the interval is then 
constructed via some aggregation process, e.g.~$\underline{X}=\min\{X_k\}$ and 
$\overline{X}=\max\{X_k\}$.  If interest is then in fitting the underlying data generating model 
$f(x_1,\ldots,x_m|\alpha)$ for inferential or predictive purposes, while only observing 
interval-valued data rather than the underlying real-valued dataset, or in having the 
interpretation of model parameters be independent of the form of the interval 
construction process, then the above descriptive models will be inadequate.
Further, the existing descriptive models for random intervals typically assume that the 
distribution of latent data points within the interval is uniform. Under the above 
data generating procedure, except in specific cases this will almost always be untrue. This assumption 
is generally accepted as a false in practice, but is typically ignored.

In this article we develop the 
% theoretical and 
methodological foundations of statistical 
models for interval-valued data that are directly constructed from an assumed underlying data generating 
model $f(x_1,\ldots,x_m|\alpha)$ and a data aggregation function $\varphi(\cdot)$ that maps 
the space of real-valued data to the space of intervals.

To the best of our knowledge, this represents the first attempt to move beyond the 
restrictive descriptive models which are prevalent in the literature, and provide an 
inferential framework that aligns with the generative interval-construction process that 
is  typical in practice. In addition to providing more directly interpretable parameters, 
it also provides a natural mechanism for departure from the uncomfortable 
uniformity-within-intervals assumption of descriptive models.

% After establishing a measurable space for a random interval $[X]$ in 
% Section~\ref{sec:distr} (of which, Section~\ref{sec:measure} may be omitted on a first 
% reading), 

In Section~\ref{sec:distr}, after establishing the containment distribution function, $F_{[X]}(\cdot)$, 
for random intervals $[X]$ based on the idea of containment functionals \cite{Molchanov2005}, 
we demonstrate the one-to-one mapping between $F_{[X]}(\cdot)$ and $f_{[X]}(\cdot)$, which is 
the density function of the bivariate real-valued random vector $(\underline{X},\overline{X})$, thereby lending some support to the current best practice for modelling random intervals. All proofs are provided in the Supporting Information.
In Section~\ref{sec:generative}, these results naturally lead to the construction of likelihood 
functions for generative models that are directly constructed from likelihood functions 
for the underlying real-valued data. We demonstrate the recovery of 
existing results on the distribution of the order statistics of a random sample under 
certain conditions. We are also 
able to show that a limiting case of the derived generative models results in a valid 
descriptive model in the sense of \citeN{Le-Rademacher2011}, implying that existing descriptive models in fact have a direct 
interpretation in terms of an underlying generative model.

All results are naturally extended from intervals to $p$-hyper-rectangles in 
Section~\ref{sec:hyper}.
In Section~\ref{sec:examples}, we contrast the performance of generative and 
descriptive models for interval-valued random variables on both simulated data, and for 
a reanalysis of a credit card dataset previously examined by \citeN{Brito2012}. Here we 
establish that the use of existing descriptive models to analyse interval-valued data 
constructed under a data generating process (which is typical in practice), can result in 
misinterpretable and biased parameter estimates,
and poorer overall fits to the observed interval-valued data than 
those obtained under generative models.
We also examine the robustness of the generative model to model and aggregation function mis-specification.
Finally, we conclude with a discussion.

%%%%%%%%%%%%%%%%%%%%%%%%%%%%%%%%%%%%%%%%%%%%%%%%%
%%%%%%%%%%%%%%%%%%%%%%%%%%%%%%%%%%%%%%%%%%%%%%%%%
\section{Distributions of Random Intervals\label{sec:distr}}
%%%%%%%%%%%%%%%%%%%%%%%%%%%%%%%%%%%%%%%%%%%%%%%%%
%%%%%%%%%%%%%%%%%%%%%%%%%%%%%%%%%%%%%%%%%%%%%%%%%

We first investigate the distribution for a random (closed) interval 
$[X]=[\underline{X},\overline{X}]$ defined on the space of 
$\mathbb{I}=\{[x,y] \colon -\infty < x \leq y < +\infty\}$. The current practice of 
constructing models for $[X]$ is by constructing models for 
the two real-valued random variables $\underline{X}$ and $\overline{X}$ with 
$\underline{X}\leq\overline{X}$ \cite{Le-Rademacher2011}. We term this approach 
the \emph{descriptive model}.

Throughout this article,
we only consider closed intervals (hyper-rectangles). Results for other types of intervals 
(hyper-rectangles) can be obtained in a similar way. We denote a vector of $m<\infty$ 
real-valued random variables by 
$X_{1:m}=\left(X_{1},\ldots,X_{m}\right)'$, 
where $X_{k}\in\mathbb{R}$ for $k=1,\ldots,m$, 
and $x_k$ is a realisation of $X_k$.
We can then 
define a data aggregation function $\varphi\colon\mathbb{R}^{m}\mapsto\mathbb{I}$ that 
maps a vector $X_{1:m}$ to the space of intervals $\mathbb{I}$
via $\left[X\right]=\varphi(X_{1:m})$, so that $[X]$ is a random (closed) interval.
For example, a useful specification for random intervals might construct the bivariate 
real-valued random variable $(\underline{X},\overline{X})$ from the minimum ($\underline{X}$) 
and maximum ($\overline{X}$) of the components of $X_{1:m}$. 

%%%%%%%%%%%%%%%%%%%%%%%%%%%%%%%%%%%%%%%%%%%%%%%%%
\subsection{Descriptive models}
%%%%%%%%%%%%%%%%%%%%%%%%%%%%%%%%%%%%%%%%%%%%%%%%%
\label{sec:descriptive}

A descriptive model treats $[X]=[\underline{X},\overline{X}]$ as a bivariate real-valued
random variable $(\underline{X},\overline{X})$ with $\underline{X}\leq\overline{X}$. 
We write $f_{[X]}(\underline{x},\overline{x}|\alpha) = f(\underline{x},\overline{x}|\alpha)$
as the likelihood function of $\left(\underline{X},\overline{X}\right)$,
where $f(\underline{x},\overline{x}|\alpha)$ is a valid density function and $\alpha$ denotes
the parameter vector of interest. 
Rather than construct models directly on $(\underline{X},\overline{X})$ with the awkward 
constraint $\underline{X}\leq\overline{X}$, a simpler approach is to remove this 
constraint through reparameterisation. For example, defining the interval 
centre $X_{c}=\frac{\underline{X}+\overline{X}}{2}$ and half-range 
$X_{r}=\frac{\overline{X}-\underline{X}}{2}$, we obtain
$
f_{[X]}(\underline{x},\overline{x}|\alpha)=\frac{1}{2}
g(\frac{\underline{x}+\overline{x}}{2},\frac{\overline{x}-\underline{x}}{2}|\alpha)
$, 
where $g(x_{c},x_{r}|\alpha)$ is a density function for $X_{c}$ and $X_{r}$.

Most existing methods to model random intervals 
(e.g.~\shortciteNP{Arroyo2010,Le-Rademacher2011,Brito2012})
can be classified as descriptive models. 
Their interpretation is simple and they are convenient to use. 
However, by construction they are only models for interval endpoints, and as a 
consequence have limitations in providing information about the distribution of the 
latent data points $X_{1:m}$.

In both symbolic data analysis \cite{Billard2006} and 
theory of random sets \cite{Molchanov2005}, the distribution of $[X]$
can be uniquely identified by a density function for a 
bivariate real-valued random variables, i.e.~$f(\underline{x},\overline{x})$ with
$\underline{x}\leq\overline{x}$.

%%%%%%%%%%%%%%%%%%%%%%%%%%%%%%%%%%%%%%%%%%%%%%%%%
\subsection{Containment distribution functions}
%%%%%%%%%%%%%%%%%%%%%%%%%%%%%%%%%%%%%%%%%%%%%%%%%
\label{sec:cdf}

In the theory of random sets, two types of 
functionals, the capacity functional and the containment functional,
are commonly used to identify a unique distribution for random sets. For 
random intervals, the capacity functional and the containment functional 
are $T_{[X]}([x])=P([X]\cap[x])$ and $C_{[X]}^{\star}([x])=P([X]\subset[x])$,
respectively. 

In the present setting, we consider a variant of the containment functional,
$C_{[X]}([x])=P([X]\subseteq[x])$,
which is more convenient for model construction. 
Due to its similarity to $C^\star_{[X]}(\cdot)$ in both functionality and interpretation, 
we still refer to $C_{[X]}(\cdot)$ as the 
containment functional throughout this article.

Similar to $C_{[X]}^{\star}(\cdot)$, a \emph{containment functional}
of a random interval $[X]$ is a functional $C_{[X]}\colon\mathbb{I}\mapsto[0,1]$
having the following properties:
\begin{enumerate}
\item[i)] $C_{[X]}([\underline{x},\overline{x}]) \to 1$,
when $\underline{x} \to -\infty$ and $\overline{x} \to +\infty$;
\item[ii)] If $[x_{1}] \supseteq [x_{2}] \supseteq \cdots \supseteq [x_{n}] \supseteq \cdots$
and $\cap_{n=1}^{\infty}[x_{n}]\in\mathbb{I}$, 
then 
\[
\lim_{n \to \infty}C_{[X]}([x_{n}])=C_{[X]}(\cap_{n=1}^{\infty}[x_{n}]);
\]
\item[iii)] For any $[x] \subseteq [y]$, $C_{[X]}([x]) \leq C_{[X]}([y])$ and
\[
C_{[X]}([y]) - C_{[X]}([\underline{y},\overline{x}]) - 
C_{[X]}([\underline{x},\overline{y}]) + C_{[X]}([x]) \geq 0.
\]
\end{enumerate}
However, it is more convenient to work
with functions defined on the real plane, 
so we equivalently define the \emph{containment distribution function} as 
$F_{[X]}(\underline{x},\overline{x})=C_{[X]}([x])$.
\begin{defn}
\label{def:cdf-intvl}The containment distribution function 
$F_{\mathrm{[X]}}\colon\mathbb{R}^{2}\mapsto[0,1]$ of the random interval $[X]$ 
has the following properties: 
\begin{enumerate}
\item[i)] $F_{[X]}(-\infty,+\infty)=1$ and $F_{[X]}(\underline{x},\overline{x})=0$
when $\underline{x}>\overline{x}$;
\item[ii)] $F_{[X]}(\underline{x},\overline{x})$ is left-continuous
in $\underline{x}$ and right-continuous in $\overline{x}$;
\item[iii)] $F_{[X]}(\underline{x},\overline{x})$ is non-increasing in $\underline{x}$, 
and non-decreasing in $\overline{x}$;
\item[iv)] For $\underline{y} \leq \underline{x} \leq \overline{x} \leq \overline{y}$,
$
F_{[X]}(\underline{y},\overline{y}) - F_{[X]}(\underline{y},\overline{x}) - 
F_{[X]}(\underline{x},\overline{y}) + F_{[X]}(\underline{x},\overline{x}) \geq 0.
$
\end{enumerate}
\end{defn}

The containment distribution function of $[X]$ can
be obtained by integration of a valid density function for random intervals.
\begin{thm}
\label{thm:pdf-intvl}Provided that   
$f_{[X]} \colon \mathbb{R}^{2} \mapsto \mathbb{R}$
is the density function of a random interval $[X]$, the containment distribution function
of $[X]$ can be derived as 
$F_{[X]}(\underline{x},\overline{x})=
\int_{\underline{x}}^{\overline{x}}\int_{\underline{x}}^{\overline{x}}\!
f_{[X]}(\underline{x}^\prime,\overline{x}^\prime)\,
\mathrm{d}\underline{x}^\prime\mathrm{d}\overline{x}^\prime$. 
\end{thm}

Conversely, the density function of $[X]$ can
be obtained by differentiation of a containment distribution function.  
\begin{thm}
\label{thm:pdf-derivative}Let $F_{[X]} \colon \mathrm{\mathbb{R}^{2} \mapsto [0,1]}$
be the containment distribution function of  a random interval $[X]$. If $F_{[X]}(\cdot)$
is twice differentiable, then the density function of $[X]$ is 
\begin{equation}
    f_{[X]}(\underline{x},\overline{x}) = 
        -\frac{\partial^{2}}{\partial\underline{x}\partial\overline{x}}
        F_{[X]}(\underline{x},\overline{x}).
\label{eq:pdf-intvl}
\end{equation}
\end{thm}

Given the data generating process, $F_{[X]}(\underline{x},\overline{x})$ can 
be naturally constructed from the generative framework, where $[X]=\varphi(X_{1:m})$, 
by noting that the two events,
$\{\varphi(X_{1:m}) \subseteq [x]\}$
and $\{[X] \subseteq [x]\} $, are equal. 
If $\varphi$ is measurable, we may compute the probability
of $\{ [X]\subseteq[x]\} $ via $P(\varphi(X_{1:m})\subseteq[x])$,
given the distribution of latent data points $X_{1:m}$. Accordingly, the containment 
distribution function of $[X]$ can be constructed as 
\begin{equation}
F_{[X]}(\underline{x},\overline{x})=P(\varphi(X_{1:m})\subseteq[x]).\label{eq:intvl-gen}
\end{equation}

Note that $[X]$ degenerates to a scalar random variable when it only contains a single
point, i.e.~when $\underline{X}=\overline{X}=X$, and so $P([X]\subseteq[x])=P(X\in[x])$ 
identifies the distribution
of a univariate real-valued random variable. In the generative framework, a 
univariate real-valued random variable is produced
when either $m=1$, or when $X_1=\cdots=X_m=X$ for $m>1$.
Accordingly, this
theory for random intervals is consistent with standard statistical theory.
For the following sections
we assume that the data aggregation function $\varphi(\cdot)$
is always measurable.

%%%%%%%%%%%%%%%%%%%%%%%%%%%%%%%%%%%%%%%%%%%%%%%%%
\subsection{Density functions}
%%%%%%%%%%%%%%%%%%%%%%%%%%%%%%%%%%%%%%%%%%%%%%%%%
\label{sec:density}

We can formally establish the distribution of random intervals by 
constructing a measurable space of $\mathbb{I}$. 

\begin{thm}
\label{thm:cdf-intvl}The containment
distribution function $F_{[X]}$ determines a unique distribution of $[X]$, 
such that $P([X]\subseteq[x])=F_{[X]}(\underline{x},\overline{x})$ 
for all $[x]\in\mathbb{I}$.
\end{thm}
From the above, $1-F_{[X]}(\underline{x},+\infty)$ and $F_{[X]}(-\infty,\overline{x})$
are the marginal cumulative distribution functions of the lower bound $\underline{X}$
and the upper bound $\overline{X}$, respectively.

The density function of $[X]$ is formally defined as the
Radon-Nikodym derivative \cite{Durrett:2010} of a probability measure on $\mathbb{I}$ 
over the uniform measure as the reference measure, as described in Theorem~\ref{thm:pdf-derivative}.

Note that a valid density function of $[X]$ is also a density function for a bivariate 
real-valued random variable. Being able to express the density function 
$f_{[X]}(\underline{x},\overline{x})$ 
of the random interval $[X]$ as the joint density of two real-valued random variables, 
$\underline{X}$ and $\overline{X}$, justifies those
existing (descriptive) methods for modelling random intervals 
(e.g.~\shortciteNP{Arroyo2010,Le-Rademacher2011,Brito2012} -- see Section~\ref{sec:descriptive})
that directly specify a joint distribution for 
$\underline{X},\overline{X} | \underline{X} \leq \overline{X}$, 
or some reparameterisation that circumvents bounding the parameter space.

%%%%%%%%%%%%%%%%%%%%%%%%%%%%%%%%%%%%%%%%%%%%%%%%%
\section{Generative models}
%%%%%%%%%%%%%%%%%%%%%%%%%%%%%%%%%%%%%%%%%%%%%%%%%
\label{sec:generative}

One approach for constructing models for $[X]$
is by constructing models for the two real-valued random variables $\underline{X}$
and $\overline{X}$ with $\underline{X}\leq\overline{X}$, i.e.~descriptive models. 
While it can describe the structure and variation between 
intervals, it is unable to model the distribution of latent data points within an 
interval, as it is simply a model for the interval endpoints. This approach is almost 
universal in the symbolic data analysis literature.
As an alternative we develop the \emph{generative model}, which is constructed directly at 
the level of the latent data points $X_{1:m}$ through the data aggregation 
function $\varphi(\cdot)$. In the following, we use $F_{[X]}(\cdot)$ and $f_{[X]}(\cdot)$ 
for interval-valued random variables, and  $F(\cdot)$ and $f(\cdot)$ for real-valued 
random variables.

A generative model of the random interval may be constructed bottom up from the
distribution of latent data points $X_{1:m}$ and the data aggregation function 
$\varphi(\cdot)$, based on (\ref{eq:intvl-gen}). 
Here, the random interval $[X]$ is constructed from $X_{1:m}$ and $\varphi(\cdot)$
via $[X]=\varphi(X_{1:m})$. If $f(x_{1:m}|\alpha)$ is
the likelihood function of the $m$ data points, then from (\ref{eq:intvl-gen}) we may 
form the containment distribution function of $[X]$  as
\begin{equation}
    F_{[X]}(\underline{x},\overline{x} | \alpha,m)=
        \int_A\!f(x_{1:m} | \alpha)\,\mathrm{d}x_{1:m},
        \label{eq:cdf-g}
\end{equation}
where $A=\{\varphi(x_{1:m}) \subseteq [x]\}$
denotes the collection of $x_{1:m}$, for which the corresponding interval
is a subset of or equal to $[x]$. If $\varphi(\cdot)$ is continuous,
the containment distribution function  (\ref{eq:cdf-g}) is twice differentiable, and so from
(\ref{eq:pdf-intvl}) its  contribution to the likelihood function would be
\begin{equation}
    f_{[X]}(\underline{x},\overline{x} | \alpha,m) = -
    \frac{\partial^{2}}{\partial\underline{x}\partial\overline{x}}\int_A\!
    f(x_{1:m} | \alpha)\,\mathrm{d}x_{1:m}.
\label{eq:lik-g}
\end{equation}
Note that containment distribution functions (\ref{eq:cdf-g}) and density 
functions (\ref{eq:lik-g}) of generative models contain a parameter 
$m$ specifying the number of latent data points 
within $[X]$.

When $m$ is large, the evaluation of (\ref{eq:lik-g}) can be challenging as it involves a
high dimensional integration. This integration can be simplified in the case where 
$X_{1:m}$ is a sequence of $\mathrm{i.i.d.}$ random variables with $X_{k}\sim f(x|\theta)$
for $k=1,\ldots,m$. We denote the likelihood function of $[X]$ with the $\mathrm{i.i.d.}$ 
latent data points by 
\begin{equation}
f_{[X]}^\star(\underline{x},\overline{x} | \theta,m)=
-\frac{\partial^{2}}{\partial\underline{x}\partial\overline{x}}
\int_A\prod_{k=1}^{m}f(x_{k}|\theta)\,\mathrm{d}x_{1:m},
\label{eq:lik-g-iid}
\end{equation}
and term it the \emph{i.i.d. generative model}.

In practice, the data aggregation function $\varphi(\cdot)$ will typically depend 
on the order statistics of the latent data points, so that 
$\varphi_{l,u}(x_{1:m})=[x_{(l)},x_{(u)}]$,
where $x_{(l)}$ and $x_{(u)}$ are respectively the $l$-th (lower)
and $u$-th (upper) order statistics of $x_{1:m}$.
The region for integration in (\ref{eq:cdf-g}) and (\ref{eq:lik-g}) then becomes
$A=\{ x_{1:m}\colon\underline{x}\leq x_{(l)}\leq x_{(u)}\leq\overline{x}\} $ --
the collection of $x_{1:m}$ for which the $l$-th order
statistic is no less than $\underline{x}$ and the $u$-th order statistic
is no greater than $\overline{x}$. 
In this case, and for $\mathrm{i.i.d.}$ random variables $X_k\sim f(x|\theta)$ for
$k=1,\ldots,m$, the likelihood
function (\ref{eq:lik-g-iid}) becomes
\begin{multline}
f_{[X]}^\star(\underline{x},\overline{x} | \theta,m,l,u)=
\frac{m!}{(l-1)!(u-l-1)!(m-u)!}[F(\underline{x}|\theta)]^{l-1}\\
\times[F(\overline{x}|\theta)-F(\underline{x}|\theta)]^{u-l-1}
[1-F(\overline{x}|\theta)]^{m-u}f(\underline{x}|\theta)f(\overline{x}|\theta),
\label{eq:lik-g-order-iid}
\end{multline}
where $F(x|\theta)=\int_{-\infty}^{x}\!f(z|\theta)\,\mathrm{d}z$
is the cumulative distribution function of $X_{k}$.
That is, (\ref{eq:lik-g-iid}) reduces to (\ref{eq:lik-g-order-iid}),
which is the joint likelihood
function of the $l$-th and $u$-th order statistics 
of $m$ $\mathrm{i.i.d.}$ samples. 
Consequently, if $l/(m+1) \to \underline{p}$ and $u/(m+1) \to \overline{p}$ as $m\to\infty$, 
the distribution of $[X]$ converges to a point mass at 
$[Q(\underline{p};\theta),Q(\overline{p};\theta)]$,
where $Q(\cdot;\theta)$ is the quantile function of $f(x|\theta)$.

Further simplification is possible when $[X]$ is constructed from the minimum and maximum 
values of $X_{1:m}$ (so that $l=1$ and $u=m$).
Here $A=\{ x_{1:m}\colon\underline{x}\leq x_{k}\leq\overline{x},\,k=1,\ldots,m\} $ is
a hyper-rectangle in $\mathbb{R}^{m}$ with identical length in
each dimension, and the likelihood function (\ref{eq:lik-g-order-iid}) becomes
\begin{equation}
f_{[X]}^{\star\star}(\underline{x},\overline{x} | \theta,m)=
m(m-1)[F(\overline{x}|\theta)-F(\underline{x}|\theta)]^{m-2}
f(\underline{x}|\theta)f(\overline{x}|\theta).
\label{eq:lik-g-mm-iid}
\end{equation}
In this case, if the support of $f(x|\theta)$ is bounded
on $[a,b]$, then as $m\to\infty$, the distribution of $[X]$
converges to a point mass at $[a,b]$.
However, if $f(x|\theta)$ has unbounded support, the distribution of $[X]$ will diverge
to $(-\infty,+\infty)$. 
 
From the above we may conclude that for  $\mathrm{i.i.d.}$ generative models, 
when $m$ is large, all interval-valued observations will be similar. As in practice we 
may expect significant variation in interval-valued observations, even for large $m$, 
this indicates that the usefulness of an $\mathrm{i.i.d.}$ model may be restricted to 
specific settings.

%%%%%%%%%%%%%%%%%%%%%%%%%%%%%%%%%%%%%%%%%%%%%%%%%
\subsection{Hierarchical generative models}
%%%%%%%%%%%%%%%%%%%%%%%%%%%%%%%%%%%%%%%%%%%%%%%%%

Evaluating the likelihood function (\ref{eq:lik-g}) of the generative model for general 
latent distributions $f(x_{1:m}|\alpha)$ of latent data points is challenging, except in 
simplified settings. 
Here we consider a special class of the generative model for which the latent data points 
$X_{1:m}$ are exchangeable. This exchangeability leads to a hierarchical generative 
model, which can capture both inter- and intra- interval structure and variability. 

Suppose that $X_{1:m}$ are exchangeable, i.e.~their joint distribution is 
invariant to any permutations of $X_{1:m}$. From de Finetti's Theorem 
\cite{Aldous:1985}, the distribution of $X_{1:m}$ may be represented as
a mixture, i.e.
\begin{equation}
    P(X_{1:m} \in A)=
    \int\!P_{\star}^{(m)}(X_{1:m} \in A) \, \mu_{P_{\star}}(\mathrm{d}P_{\star}),
\label{eq:definetti}
\end{equation}
where $\mu_{P_{\star}}$ is the distribution on the space of all probability
measures of $\mathbb{R}$, and $P_{\star}^{(m)}=\prod_{m}P_{\star}$
is the product measure on $\mathbb{R}^{m}$. In other words, all $X_{k}$ for $k=1,\ldots,m$, 
are $\mathrm{i.i.d.}$ from $P_{\star}$ with $P_{\star}\sim\mu_{P_{\star}}$.
By recalling from (\ref{eq:cdf-g}) and (\ref{eq:lik-g}) that 
$A = \{\varphi(x_{1:m}) \subseteq [x]\}$, 
then the mixture component $P_{\star}^{(m)}(X_{1:m} \in A)$
equals $P_{\star}^{(m)}([X] \subseteq [x])$, which is the containment distribution 
function for an $\mathrm{i.i.d.}$ generative model of $[X]$, with $X_k \sim P_{\star}$ 
for $k=1,\ldots,m$ and the same data aggregation function $\varphi(\cdot)$. 
This means that $P([X]\subseteq[x])$, which equals $P(X_{1:m}\in A)$,
may be represented as the mixture of $P_{\star}^{(m)}([X]\subseteq[x])$
with $P_{\star} \sim \mu_{P_{\star}}$, i.e.~as a mixture of $\mathrm{i.i.d.}$ 
generative models.

In the following we consider the case when $P_{\star}$ belongs to some parametric family, 
so that $\mathrm{d}P_{\star}=f(x|\theta)\,\mathrm{d}x$.
From (\ref{eq:definetti}), the joint density function of $X_{1:m}$ is then given 
by the mixture representation 
$\int\!\prod_{k=1}^{m}f(x_{k}|\theta)\pi(\theta)\,\mathrm{d}\theta$, 
where the mixing distribution $\pi(\theta)$ may be non-parametric or 
parametric $\pi(\theta|\alpha)$ 
with parameter $\alpha$.
The resulting containment distribution function of $[X]$ is then the mixture of 
$F_{[X]}(\underline{x},\overline{x}|\theta,m)$
given in (\ref{eq:cdf-g}), with $f(x_{1:m} | \theta) = \prod_{k=1}^m f(x_k | \theta)$,
w.r.t. $\pi(\theta|\alpha)$.
If $\varphi(\cdot)$ is continuous, we obtain the likelihood
function of such a generative model as
\begin{equation}
f_{[X]}(\underline{x},\overline{x} | \alpha,m)=
\int\!f_{[X]}^\star(\underline{x},\overline{x} | \theta,m)
\pi(\theta|\alpha)\,\mathrm{d}\theta,
\label{eq:lik-g-mix}
\end{equation}
where $f_{[X]}^\star(\underline{x},\overline{x} | \theta,m)$
is the likelihood function of $\mathrm{i.i.d.}$ generative model (\ref{eq:lik-g-iid}). 

In practice, the latent data points $X_{1:m}$ may not be exchangeable. However the data 
aggregation function $\varphi(\cdot)$ may be symmetric. 
Let $\Gamma$ be the set of all permutations of the indices from $1$ to $m$, 
and $X_\gamma$ be the 
latent data points $X_{1:m}$ permuted according to $\gamma\in\Gamma$ with density 
function $f(x_\gamma)$. As $\varphi(\cdot)$ is symmetric, 
$\varphi(x_{\gamma})=\varphi(x_{1:m})$
and thus, $[X_\gamma]=\varphi(X_\gamma)$ has the same containment distribution function as 
$[X]$. As a result, for 
the exchangeable random variables defined as
$Y_{1:m} \sim \frac{1}{m!} \sum_{\gamma\in\Gamma}f(X_{\gamma})$,
$[Y]=\varphi(Y_{1:m})$ has the same containment distribution function as $[X]$.

The existence of such $Y_{1:m}$ implies that when the latent data points $X_{1:m}$ are 
aggregated into intervals $[X]$ by symmetric data aggregation methods, information on 
the order-related dependence structure will vanish. As a 
result, it is unnecessary to model the distribution of $X_{1:m}$ with a more complex 
dependence structure than exchangeability -- modelling the exchangeable $Y_{1:m}$ will be 
sufficient.

Accordingly, for random intervals $[X_1],\ldots,[X_n]$, 
the generative model (\ref{eq:lik-g-mix}) can be directly 
interpreted as the hierarchical model
\begin{eqnarray*}
[X_{i}] & = & \varphi(X_{i,1:m}),\nonumber \\
X_{i,k} & \sim & f(x|\theta_{i}),\,k=1,\ldots,m_{i},\nonumber \\
\theta_{i} & \sim & \pi(\theta|\alpha),\label{eq:hierarchy}
\end{eqnarray*}
with known $m_i$ for $i=1,\ldots,n$. Thus, we term them \emph{hierarchical generative models}.
The contribution to the integrated likelihood (\ref{eq:lik-g-mix}) for the first two 
terms is given by $f_{[X]}^\star(\underline{x}_{i},\overline{x}_{i}|\theta_i,m_{i})$
-- the likelihood function of the $\mathrm{i.i.d.}$ 
generative model (\ref{eq:lik-g-iid}) for the 
interval-valued observation $[x_i]$, with
the density function
of  each (conditionally) $\mathrm{i.i.d.}$ latent data points $X_{i,1:m_i}$ given by 
$f(x_{i,k}|\theta_{i})$ -- and where $\pi(\theta|\alpha)$
is the mixing distribution for $\theta_i$ given the parameter $\alpha$. Given such 
interpretation, $f(x_{i,k}|\theta_{i})$ (or $\theta_{i}$) is the \emph{local} density function 
(or parameter) for $[X_i]$, while $\pi(\theta | \alpha)$ (or $\alpha$) 
is the \emph{global} density 
function (or parameter) among all intervals. Therefore, the intra-interval structure is 
described by the local density function and $m$, while the inter-interval variability is 
modelled by the global density function.

As a result, inference on this model permits direct analysis of the underlying 
distribution of data points $X_{1:m}$ within each interval $[X_i]$ and its model 
parameter $\theta_i$ -- an advantageous property of the generative model over the 
descriptive model. 
For example if the global density $\pi(\theta|\alpha)$ works as the prior distribution, 
in the Bayesian framework, for the local parameter $\theta_i$,
$\pi(\theta_{i}|\alpha,[x_i])\propto 
f_{[X]}^\star(\underline{x}_{i},\overline{x}_{i}|\theta_i,m_{i})\pi(\theta_i|\alpha)$ 
is the posterior distribution of the parameter of the local density 
$f(x | \theta_i)$ underlying $[x_{i}]$.
Similarly,
the posterior predictive distribution of latent data points underlying $[x_i]$ is 
directly available as 
$\pi(x|\alpha,[x_i]) \propto 
\int\!f(x|\theta_{i})\pi(\theta_{i}|\alpha,[x_{i}])\,\mathrm{d}\theta_i$.

%%%%%%%%%%%%%%%%%%%%%%%%%%%%%%%%%%%%%%%%%%%%%%%
%%%%%%%%%%%%%%%%%%%%%%%%%%%%%%%%%%%%%%%%%%%%%%%%
\subsection{Asymptotic properties}%%%
%%%%%%%%%%%%%%%%%%%%%%%%%%%%%%%%%%%%%%%%%%%%%%%%
%%%%%%%%%%%%%%%%%%%%%%%%%%%%%%%%%%%%%%%%%%%%%%%

Although they are constructed quite distinctly, it is possible to directly relate the 
descriptive and generative models under specific circumstances. In particular for 
standard (descriptive) symbolic analysis techniques,
when there is no prior knowledge on the distribution of data within an interval, 
this distribution is commonly assumed to be uniform $U(a,b)$ with 
$a \leq b$ (e.g.~\citeNP{Le-Rademacher2011}).
Let $I(\underline{x}, \overline{x}\colon a \leq \underline{x} \leq \overline{x} \leq b)$ be
an indicator function of $\underline{x}$ and $\overline{x}$, which equals 1 when 
$a \leq \underline{x} \leq \overline{x} \leq b$, and 0 elsewhere.
Defining $f(x|\theta)$ so that $X_k \sim U(a,b)$ for $k=1,\ldots,m$, and 
constructing $[X]=\varphi_{1,m}(X_{1:m})$ from the minimum and maximum values of these latent 
data points, then the density function of $[X]$ given by (\ref{eq:lik-g-mm-iid}) becomes
\[
f_{[X]}^{\star\star}(\underline{x},\overline{x} | a,b,m)=
m(m-1)(\overline{x} - \underline{x})^{m-2}(b-a)^{-m}
I(\underline{x}, \overline{x} \colon a \leq \underline{x} \leq \overline{x} \leq b),
\]
which converges to a point mass at $[a,b]$ as $m\to\infty$ 
(Section~\ref{sec:generative}). 
Then, by substituting 
$f_{[X]}^{\star\star}(\underline{x},\overline{x}| a,b,m)$ into (\ref{eq:lik-g-mix}),
the hierarchical generative model becomes
\begin{equation}
\label{eq:lik-g-unif}
f_{[X]}(\underline{x},\overline{x}|m)=\iint_{\{a \leq \underline{x}, b \geq \overline{x}\}}\!
m(m-1)\frac{(\overline{x}-\underline{x})^{m-2}}{(b-a)^{m}}\pi(a,b)\,\mathrm{d}a\mathrm{d}b.
\end{equation}
where $\pi(a,b)$ describes the inter-interval parameter variability.
When $m$ is large, the following theorem states that this hierarchical generative  model 
converges to
$\pi(\underline{x},\overline{x})$,
which is a valid descriptive model.
\begin{thm}
\label{thm:g2d-unif}
Suppose that $[X]=\varphi_{1,m}(X_{1:m})$ with $X_{k} \sim U(a,b)$ for $k=1,\ldots,m$, and
the global density function $\pi(a,b)$ is bounded, continuous and equal to 0 when $a > b$.
Then as $m\to\infty$, the density function of $[X]$ (\ref{eq:lik-g-unif}) 
converges to $\pi(\underline{x},\overline{x})$ pointwise, i.e.
\[
\lim_{m\to\infty}f_{[X]}(\underline{x},\overline{x}|m)=\pi(\underline{x},\overline{x}).
\]
\end{thm}
This result is interesting in that it reveals that descriptive models for 
$[X] \sim f_{[X]}(\underline{x},\overline{x}|\theta)$ described 
in Section~\ref{sec:descriptive} (e.g.~\shortciteNP{Arroyo2010,Le-Rademacher2011,Brito2012})
actually possess an underlying and implicit generative structure. 
Specifically, the sampling process
of the descriptive model 
$[X] \sim f_{[X]}(\underline{x},\overline{x}) = \pi(\underline{x},\overline{x})$
can be expressed via the generative process 
\begin{eqnarray*}
[X] & = & \lim_{m\to\infty}\varphi_{1,m}(X_{1:m}),\\
X_{1},X_{2}\ldots & \sim & U(\underline{X}_{\star},\overline{X}_{\star}),\\
(\underline{X}_{\star},\overline{X}_{\star}) & \sim & \pi(\underline{x},\overline{x}).
\end{eqnarray*}
That is, to obtain a sample realisation of $[X]$, 
values of lower and upper bound parameters,  
$(\underline{X}_{\star},\overline{X}_{\star})$, of local 
uniform distribution are first drawn from the descriptive model 
$\pi(\underline{x},\overline{x})$, 
which in this case is exactly equivalent to the global density for the 
associated underlying hierarchical generative model.
As the resulting infinite collection of latent data points 
$X_k \sim U(\underline{X}_{\star},\overline{X}_{\star})$ fully 
identifies the local density, and
$\min\{X_k\} = \underline{X}_\star$, $\max\{X_k\}=\overline{X}_\star$ 
are sufficient statistics for uniform distributions, the generated interval $[X]$ is then 
determined as $[X]=[\underline{X}_{\star},\overline{X}_{\star}]$ with 
$(\underline{X}_{\star},\overline{X}_{\star}) 
\sim \pi(\underline{x}_{\star},\overline{x}_{\star})$.
As a result, there is no loss of information from the data aggregation procedure and 
the variation of $[X]$ is completely due
to the variation permitted in the distribution of local parameters, which is the global 
distribution. In this manner, the descriptive
model is a special case of and directly interpretable as a particular generative model.

This idea can be extended to 
a more general class of hierarchical generative models in which
the local distribution is only governed by location ($\mu$) and
scale ($\tau>0$) parameters, so that $X_k \sim f(x|\mu,\tau)$ for $k=1,\ldots,m$. 
Suppose $\underline{x}$ and $\overline{x}$
are the $l$-th and $u$-th order statistics, respectively. The associated values of 
$\mu$ and $\tau$ are available by solving 
\begin{equation}
\begin{cases}
Q({l}/{(m+1)};\mu,\tau) & =\underline{x}\\
Q({u}/{(m+1)};\mu,\tau) & =\overline{x},
\end{cases}\label{eq:intvl-identify}
\end{equation}
where $Q(\cdot;\mu,\tau)$ is the quantile 
function of $f(x|\mu,\tau)$.
If a unique solution exists for (\ref{eq:intvl-identify}), then $f(x | \mu,\tau)$
is an \emph{interval-identifiable} local distribution. 

We previously discussed that under the order statistic based data aggregation function, the
$\mathrm{i.i.d.}$ generative model~(\ref{eq:lik-g-order-iid}) will converge to a point
mass as $m\rightarrow\infty$. Similar to Theorem~\ref{thm:g2d-unif}, 
those hierarchical generative models (\ref{eq:lik-g-mix})
with interval-identifiable local density functions $f(x|\mu,\tau)$  will also
converge to descriptive models.
\begin{thm}
\label{thm:g2d}
Suppose that $[X]=\varphi_{l,u}(X_{1:m})$ with $X_{k} \sim f(x | \mu,\tau)$ 
for $k=1,\ldots,m$, 
where the local density function $f(x | \mu,\tau)$ 
is interval-identifiable with location parameter $\mu$ 
and scale parameter $\tau>0$.
Further suppose that $l/(m+1) \to \underline{p} > 0$ and 
$u/(m+1) \to \overline{p} < 1$ as $m\to\infty$, and
\begin{enumerate}
\item[i)] the global density function $\pi(\mu,\tau)$ is twice differentiable;
\item[ii)] $f(x | \mu,\tau)$ is positive and continuous in neighbourhoods of 
$Q(\underline{p};\mu,\tau)$ and $Q(\overline{p};\mu,\tau)$;
\item[iii)] $\iint\!|f_{[X]}^{\star}(\underline{x},\overline{x}|\mu,\tau,m,l,u)|
\pi(\mu,\tau)\,\mathrm{d}\mu\mathrm{d}\tau<\infty$ for any $0 < l < u < m$.
\end{enumerate}
Then as $m\to\infty$, the density function of $[X]$ for 
the hierarchical generative model (\ref{eq:lik-g-mix})
converges pointwise to 
\[
\pi_{\star}(\underline{x},\overline{x})=
\pi\left(\mu(\underline{x},\overline{x};\underline{p},\overline{p}),
\tau(\underline{x},\overline{x};\underline{p},\overline{p})\right)\times
\left|J(\mu(\underline{x},\overline{x};\underline{p},\overline{p}),
\tau(\underline{x},\overline{x};\underline{p},\overline{p});
\underline{p},\overline{p})\right|^{-1},
\]
where $\mu(\underline{x},\overline{x};\underline{p},\overline{p})$
and $\tau(\underline{x},\overline{x};\underline{p},\overline{p})$
are the solution of (\ref{eq:intvl-identify}) and
\[
J(\mu,\tau;\underline{p},\overline{p})=
\begin{pmatrix}
    \frac{\partial}{\partial \mu}Q(\underline{p}|\mu,\tau) 
    & \frac{\partial}{\partial \tau}Q(\underline{p}| \mu,\tau)\\
    \frac{\partial}{\partial \mu}Q(\overline{p}| \mu,\tau) 
    & \frac{\partial}{\partial \tau}Q(\overline{p}| \mu,\tau)
\end{pmatrix}. 
\]
\end{thm}
In the specific case where $f(x| a,b)$ is a $U[a,b]$ 
local density function, with quantile function $Q(p| a,b)=(1-p)a + pb$,
the hierarchical generative model (\ref{eq:lik-g-mix}) converges to the distribution of 
\[
[(1-\underline{p})\underline{X}_{\star}+\underline{p}\overline{X}_{\star},
(1-\overline{p})\underline{X}_{\star}+\overline{p}\overline{X}_{\star}],
\]
where $(\underline{X}_{\star},\overline{X}_{\star})\sim\pi(\underline{x},\overline{x})$.

%%%%%%%%%%%%%%%%%%%%%%%%%%%%%%%%%%%%%%%%%%%%%%%%%
%%%%%%%%%%%%%%%%%%%%%%%%%%%%%%%%%%%%%%%%%%%%%%%%%
\section{Multivariate Models for hyper-rectangles\label{sec:hyper}}
%%%%%%%%%%%%%%%%%%%%%%%%%%%%%%%%%%%%%%%%%%%%%%%%%
%%%%%%%%%%%%%%%%%%%%%%%%%%%%%%%%%%%%%%%%%%%%%%%%%

The $p$-dimensional equivalent of the univariate interval-valued random variable $[X]$ 
is the random $p$-hyper-rectangle, which corresponds to a $p$-tuple of random intervals. 
In specific, we denote $[\boldsymbol{x}]=([x_{1}],\ldots,[x_{p}])\in\mathbb{I}^p$
as a hyper-rectangle in the space of $p$-hyper-rectangles, 
and $\boldsymbol{x}=(x_1,\ldots,x_p)\in\mathbb{R}^p$
as one $p$-dimensional latent data point.
It is straightforward to extend
the previous theory on containment distribution functions and
likelihood functions for random intervals (Sections~\ref{sec:distr}~and~\ref{sec:generative})
to random hyper-rectangles.

%%%%%%%%%%%%%%%%%%%%%%%%%%%%%%%%%%%%%%%%%%%%%%%%%
\subsection{Containment distribution functions}
%%%%%%%%%%%%%%%%%%%%%%%%%%%%%%%%%%%%%%%%%%%%%%%%%

Similar to Section~\ref{sec:descriptive}, descriptive models for random $p$-hyper-rectangles may be 
constructed through direct specification of the $2p$-dimensional density function 
$f_{[\boldsymbol{X}]}(\underline{x}_{1},\overline{x}_{1},\ldots,
\underline{x}_{p},\overline{x}_{p})$.
These models are easily constructed and simple to use, 
but have the same limitations as the descriptive models for random intervals
discussed in Section \ref{sec:descriptive}. 

The containment distribution function of $[\boldsymbol{X}]$, denoted 
$F_{\mathrm{[\boldsymbol{X}]}}\colon\mathbb{R}^{2p}\mapsto[0,1]$, is
a function on the real hyperplane, having similar properties to those described in
Definition~\ref{def:cdf-intvl} (not stated here for brevity). 
The following theorems show the connection 
between the containment distribution function and the density function for $[\boldsymbol{X}]$.
\begin{thm}
\label{thm:pdf-hyper}
Provided that $f_{[\boldsymbol{X}]}\colon\mathbb{R}^{2p}\mapsto\mathbb{R}$
is the density function of a random $p$-hyper-rectangle $[\boldsymbol{X}]$, 
the containment distribution function can be derived as follows,
\[
F_{[\boldsymbol{X}]}(\underline{x}_{1},\overline{x}_{1},
\ldots,\underline{x}_{p},\overline{x}_{p})=
\int_{\underline{x}_p}^{\overline{x}_p}\cdots\int_{\underline{x}_1}^{\overline{x}_1}\!
f_{[\boldsymbol{X}]}(\underline{x}_1^\prime,\overline{x}_1^\prime,\ldots,
\underline{x}_p^\prime,\overline{x}_p^\prime)\,
\mathrm{d}\underline{x}_1^\prime\mathrm{d}\overline{x}_1^\prime\ldots
\mathrm{d}\underline{x}_p^\prime\mathrm{d}\overline{x}_p^\prime.
\]
\end{thm}
\begin{thm}
\label{thm:pdf-derivative-hyper}
Let $F_{[\boldsymbol{X}]}\colon\mathrm{\mathbb{R}^{2p}\mapsto[0,1]}$ be the containment 
distribution function of a random hyper-rectangle $[\boldsymbol{X}]$. If $F_{[\boldsymbol{X}]}$ 
is $2p$-times differentiable, then the density function of $[\boldsymbol{X}]$ is
\begin{equation}
f_{[\boldsymbol{X}]}(\underline{x}_{1},\overline{x}_{1},\ldots,
\underline{x}_{p},\overline{x}_{p})=
(-1)^{p}\frac{\partial^{2p}}{\partial\underline{x}_{1}\partial\overline{x}_{1}\ldots
\partial\underline{x}_{p}\partial\overline{x}_{p}}
F_{[\boldsymbol{X}]}(\underline{x}_{1},\overline{x}_{1},
\ldots,\underline{x}_{p},\overline{x}_{p}).
\label{eq:pdf-hyper}
\end{equation}
\end{thm}

%%%%%%%%%%%%%%%%%%%%%%%%%%%%%%%%%%%%%%%
\subsection{Generative models}
%%%%%%%%%%%%%%%%%%%%%%%%%%%%%%%%%%%%%%%
\label{sec:pdim-lik}
 
Containment distribution functions and likelihood
functions of generative models may be formulated using the same ideas as in (\ref{eq:cdf-g})
and (\ref{eq:lik-g}). However, due to the necessity of calculating $2p$-th order mixed 
derivatives in (\ref{eq:pdf-hyper}), although intuitive,
the structure of the resulting likelihood functions would be highly complex,
even for $\mathrm{i.i.d.}$ generative models of random rectangles. The full form 
of the likelihood function for an $\mathrm{i.i.d.}$ generative model in 
the bivariate case $[\boldsymbol{X}]=[X_1]\times[X_2]$ 
is presented in the Supplementary Information.

The complex form of the likelihood function of an $\mathrm{i.i.d.}$ 
generative model accordingly 
induces a similarly complex hierarchical generative model.
One option to produce more tractable models is to impose a conditional independent 
structure within each $p$-dimensional latent data point, so that 
$\boldsymbol{x}_k \sim f(\boldsymbol{x}|\theta_{1:p})=\prod_{j=1}^{p}f(x_{j}|\theta_{j})$.
Consequently, each random interval marginal distribution of the 
$p$-hyper-rectangle is conditionally 
independent of the others, i.e. 
\[f^\star_{[\boldsymbol{X}]}(\underline{x}_{1},\overline{x}_{1},\ldots,
\underline{x}_{p},\overline{x}_{p}|
\theta_{1:p})=\prod_{j=1}^{p}f^\star_{[X_j]}(\underline{x}_{j},\overline{x}_{j}|\theta_{j}),\]
where $f^\star_{[X_j]}(\underline{x}_{j},\overline{x}_{j}|\theta_{j})$ 
is the likelihood function 
of the $\mathrm{i.i.d.}$ generative model (\ref{eq:lik-g-iid}) for $[X_j]$. 
Although this choice will result in clear modelling consequences,
the resulting likelihood function for the hierarchical generative model
\begin{equation}
f_{[\boldsymbol{X}]}(\underline{x}_{1},\overline{x}_{1},\ldots,
\underline{x}_{p},\overline{x}_{p}|m,\alpha)=
\int\!\prod_{j=1}^{p}f^\star_{[X_j]}(\underline{x}_{j},\overline{x}_{j}|\theta_{j})
\pi(\theta_{1:p}|\alpha)\,\mathrm{d}\theta_{1:p}
\label{eq:lik-g-mix-hyper}
\end{equation}
will only then require $p$ second-order mixed derivatives. 

In this scenario, dependencies between the random interval marginal distributions of 
$[\boldsymbol{X}]$, such as temporal or spatial dependencies, are controlled only by
the dependence among local parameters $\theta_{1:p}$ as introduced by the 
global distribution $\pi(\theta_{1:p}|\alpha)$.
As a result, beyond any \emph{a priori} information on the joint distribution of the 
$p$-dimensional latent data points underlying construction of the random interval 
$[\boldsymbol{X}]$ 
being incorporated within $\pi(\theta_{1:p}|\alpha)$,
it will be impossible to identify any further dependence based on the observed 
$p$-hyper-rectangles. If this is inadequate for a given analysis, the full multivariate 
likelihood will need to be derived (see e.g.~the Supplementary Information).

%%%%%%%%%%%%%%%%%%%%%%%%%%%%%%%%%%%%%%%%%%%%%%%%%
%%%%%%%%%%%%%%%%%%%%%%%%%%%%%%%%%%%%%%%%%%%%%%%%%
\section{Applications\label{sec:application}}
%%%%%%%%%%%%%%%%%%%%%%%%%%%%%%%%%%%%%%%%%%%%%%%%%
%%%%%%%%%%%%%%%%%%%%%%%%%%%%%%%%%%%%%%%%%%%%%%%%%
\label{sec:examples}

We illustrate our new models by firstly comparing the performance of 
the generative models to the existing descriptive models for simulated 
univariate (random interval) data.
We then provide a generative model reanalysis of a real dataset of 5,000 credit 
card customers, as previously analysed by \cite{Brito2012} using a descriptive 
model.
The size of this dataset does not merit the use of symbolic data methods for its analysis, however it does serve as a useful illustration of the benefits of generative models.
We conclude with an examination of the robustness of the generative method to model mis-specification.

%%%%%%%%%%%%%%%%%%%%%%%%%%%%%%%%%%%%%%%%%%%%%%%%%
\subsection{Simulated data analysis}
%%%%%%%%%%%%%%%%%%%%%%%%%%%%%%%%%%%%%%%%%%%%%%%%%

\begin{figure}[t]
\centering
\includegraphics[scale=0.28]{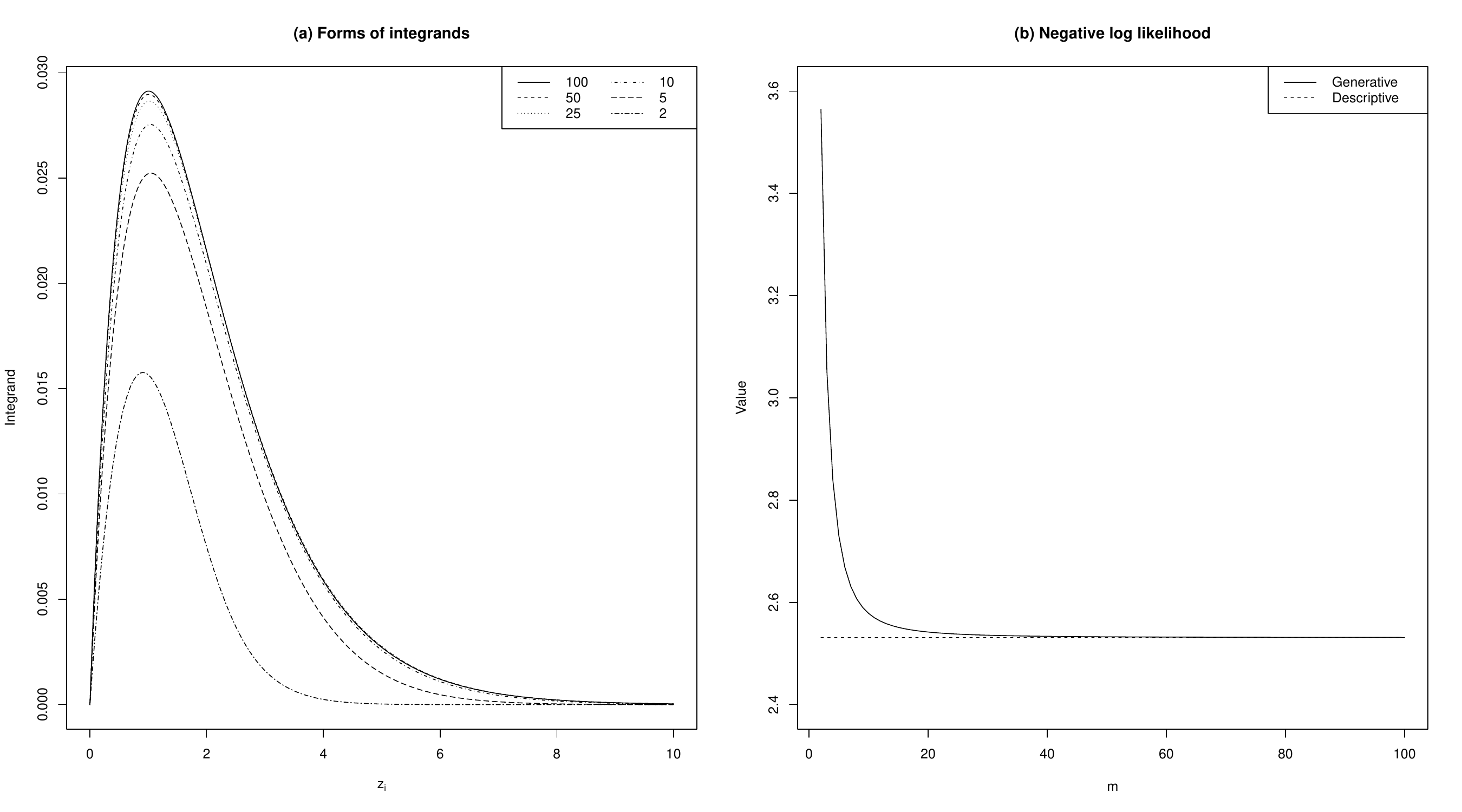}
\protect\caption{\label{fig:sim-lik}\small (a) Forms of the integrand 
in~(\ref{eq:lik-g-mix-mm-unif}), with $m=2,5,10,25,50,100$, as a function of $z_i$, and 
(b) the negative log likelihood function as a function of $m$, when $\underline{x}_{i}=-1$,
$\overline{x}_{i}=1$, $\mu_{c}=\mu_{\tau}=0$ and $\sigma_{c}^{2}=\sigma_{\tau}^{2}=1$.}
\afterpage{\FloatBarrier}
\end{figure}

In order to provide a direct comparison between descriptive and generative models, 
we construct our observed random intervals under the generative model as
$[x_{i}]=[\underline{x}_{i},\overline{x}_{i}]$, where $\underline{x}_{i}$ and 
$\overline{x}_{i}$ are respectively the observed minimum and maximum values 
of $x_{i1},\ldots,x_{im_i}$ under the mixture model
\begin{eqnarray}
x_{i1},\ldots,x_{im_i} & \sim & U(c_{i}-\mathrm{e}^{\tau_{i}},c_{i}+\mathrm{e}^{\tau_{i}}),
\label{eqn:sugar}\\
c_{i}\sim N(\mu_{c},\sigma_{c}^{2}) & \text{and} & \tau_{i}\sim 
N(\mu_{\tau},\sigma_{\tau}^{2}),\nonumber
\end{eqnarray}
for $i=1,\ldots,n$.
From Theorem~\ref{thm:g2d-unif}, this hierarchical model is asymptotically 
equivalent (as $m_i\to\infty$ for each $i$) to a descriptive model with 
$[x_{i}^{\star}]=
[c_{i}^{\star}-\mathrm{e}^{\tau_{i}^{\star}},c_{i}^{\star}+\mathrm{e}^{\tau_{i}^{\star}}]$,
where $(c_{i}^{\star},\tau_{i}^{\star})$ follows the same joint
distribution as $(c_{i},\tau_{i})$. While in practice random intervals will generally be 
constructed from different numbers of random samples, $x_{i1},\ldots,x_{im_i}$ (e.g.~see 
Section~\ref{sec:realDataAnalysis}), here we specify $m_{i}=m$ for all $i=1,\ldots,n$. 
In this analysis we will compare the maximum likelihood estimators (MLEs) of parameters 
for both generative and descriptive models obtained using  data simulated from each model.

\begin{figure}[t]
\centering
\includegraphics[scale=0.52]{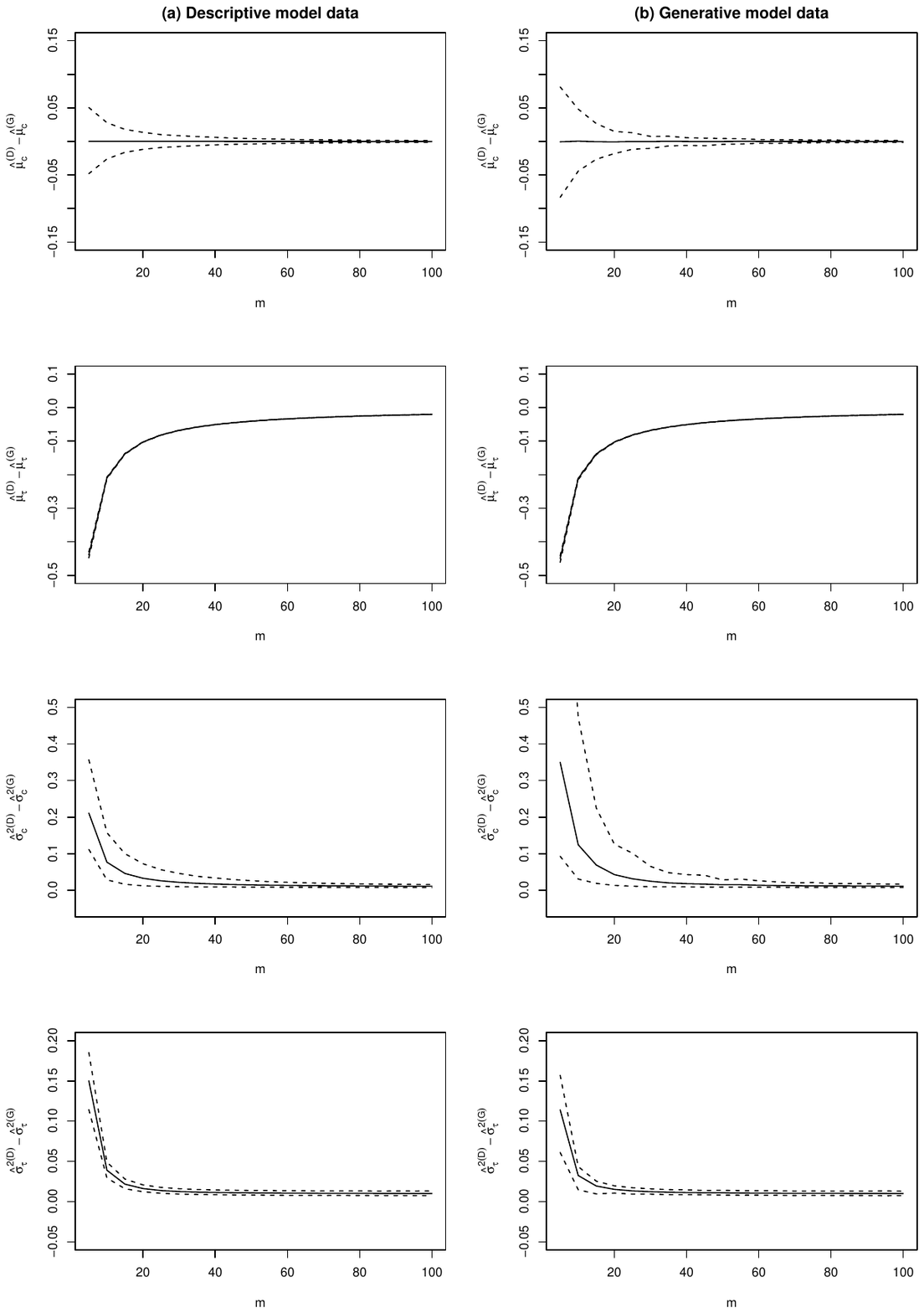}
\protect\caption{\small\label{fig:sim-compare} Differences between MLEs of the 
descriptive model and the generative hierarchical model, based on data generated from each 
model (left column = descriptive model data, right column = generative model data), as a 
function of $m=5, \ldots, 100$, the number of latent data points $x_{i1},\ldots,x_{im}$ in 
the generative model. Lines indicate the MLE means (solid lines) and 2.5\% and 97.5\% 
quantiles (dashed lines) based on 1,000 replicate datasets.}
\afterpage{\FloatBarrier}
\end{figure}

For each random interval $[x_i]$ under the mixture model, the two-dimensional 
integration~(\ref{eq:lik-g-mix}), with $\theta=(c_i,\tau_i)$, can be reduced to a 
one-dimensional integration by first integrating out $c_i$, and then reparameterising to
$z_{i}=m(\tau_{i}-\log\frac{1}{2}(\overline{x}_{i}-\underline{x}_{i}))$. 
This leads to
the likelihood function of a single interval observation $[x_{i}]$ given by
\begin{multline}
\int_{0}^{\infty}\!(\overline{x}_i-\underline{x}_i)^{-2}(m-1)\mathrm{e}^{-z_i}
\phi(m^{-1}z_i+\log\frac{\overline{x}_i-\underline{x}_i}{2};\mu_{\tau},\sigma_{\tau}^2)
\times\\
\{\Phi(\underline{x}_i+\frac{\overline{x}_i-\underline{x}_i}{2}
\mathrm{e}^{m^{-1}z_i};\mu_c,\sigma_c^2)-
\Phi(\overline{x}_i-\frac{\overline{x}_i-\underline{x}_i}{2}
\mathrm{e}^{m^{-1}z_i};\mu_c,\sigma_c^2)\}
\,\mathrm{d}z_{i},\label{eq:lik-g-mix-mm-unif}
\end{multline}
where $\phi$ and $\Phi$ respectively denote the Gaussian density and distribution
function. This form may be quickly and accurately approximated by
Gauss-Laguerre quadrature methods (e.g.~\citeNP{Evans2000}).
The form  of the integrand in (\ref{eq:lik-g-mix-mm-unif}) for varying $m$ and the 
resulting negative log-likelihood function is shown in Figure~\ref{fig:sim-lik} for 
$\underline{x}_{i}=-1$, $\overline{x}_{i}=1$, $\mu_{c}=\mu_{\tau}=0$
and $\sigma_{c}^{2}=\sigma_{\tau}^{2}=1$. 
These plots illustrate the convergence of the generative model to the descriptive model 
as $m$ gets large (Theorem~\ref{thm:g2d-unif}), with only very minor differences observed 
for $m>30$, and also suggest (panel (a)) that quadrature integration methods will be 
accurate with around 20 nodes.

We simulate 1,000 replicate datasets, each comprising $n=100$ intervals, from the 
descriptive model with  $c_{i}^{\star}, \tau_{i}^{\star}\sim N(0,1)$ for $i=1,\ldots,n$ 
(i.e.~$\mu_c=\mu_\tau=0$ and $\sigma^2_c=\sigma^2_\tau=1$).
MLEs of the model parameters ($\mu_c,\mu_\tau,\sigma^2_c,\sigma^2_\tau$) are obtained 
from fitting both descriptive and generative models, with the latter assuming a specified 
number of latent variables, $m$. Note that in practice, the number of latent variables, 
$m$, will typically be known (and finite).
The first column of Figure~\ref{fig:sim-compare} illustrates the differences between the 
resulting descriptive and generative model parameter MLEs 
(e.g.~$\hat{\mu}_c^{(D)}-\hat{\mu}_c^{(G)}$, 
where the superscripts indicate parameters of the descriptive ($D$) and generative ($G$) 
models), with the solid line indicating the mean and the dotted lines the central 95\% 
interval, computed over the 1,000 replicates.

Firstly, we notice that the difference between the estimates is large for small $m$, and 
becomes gradually smaller as $m$ increases. This is not surprising as in this model 
specification, the generative model approaches the descriptive mode as $m\to\infty$.
However, as both models are identically centred, the mean difference between the location 
parameter estimates $\hat{\mu}_c^{(D)}$ and $\hat{\mu}_c^{(G)}$  is zero, regardless of 
the number of latent variables.

An obvious area of difference is that the point estimates of the interval half-range 
(modelled by $\mu_\tau$) are much smaller for the (correct) descriptive model than for 
the generative model. This occurs as, the expected range of $x_{i1},\ldots,x_{im}$ under 
a generative model is lower for small $m$ than it is for large $m$. As a result, the 
generative model will determine that $\mu_\tau$ should be sufficiently larger for small 
$m$ than it would be for large $m$, given the same observed 
$[\underline{x}_i,\overline{x}_i]$. 
That is, if the data are truly generated from the descriptive model, parameters estimated 
from the generative model are effectively biased for any finite $m$, and overestimate the 
true model parameters, with the magnitude of the bias determined by the assumed value of 
$m$. Of course, this bias can be reduced by setting $m$ to be large in this case.

The second area of difference is that the estimated variability of the point estimates
of interval location and scale ($\hat{\sigma}^2_c$ and $\hat{\sigma}^2_\tau$) is 
higher under the descriptive model than under the generative model. This occurs as the 
generative model assumes that the variability of e.g.~$\frac{\underline{x}_i+\overline{x}_i}{2}$ comprises both the 
variability of the latent data $x_{i1},\ldots,x_{m}$ within interval $i$, in addition 
to the variability of interval locations $c_i$ between intervals. Under the descriptive 
model, this first source of variability is zero, and therefore $\hat{\sigma}^{2(D)}_c$ 
will always be greater than $\hat{\sigma}^{2(G)}_c$ for finite $m$. Similar reasoning  
explains why $\hat{\sigma}^{2(D)}_\tau$ is always greater than $\hat{\sigma}^{2(G)}_\tau$.

The second column of Figure~\ref{fig:sim-compare} shows the same output as the first 
column, but based on data simulated from the generative model with the same parameter 
settings as before, and for varying (true) numbers of latent data points 
$m=5,\ldots, 100$. The 
results are similar to before, except critically with the interpretation that the 
generative model with fixed $m$ is now correct. This means that, for example, if 
intervals are constructed using the generative process (which is the most likely 
scenario in practice) but are then analysed with a descriptive model, the point 
estimates of interval range ($\mu_\tau$) can be substantially underestimated by assuming
$m\to\infty$ under the descriptive model, when in fact $m$ is small and finite. 
Similarly, the estimates of $\sigma^2_c$ and $\sigma^2_\tau$ will always be overestimated 
when assuming an incorrect descriptive model.
These scenarios will obviously be problematic for data analysis in practice.

The takeaway message of  this analysis is that it is important to fit the model 
(descriptive or generative) that matches the interval (or $p$-hyper-rectangle) construction 
process. Failure to do so can result in misinterpretation of model parameters, resulting 
in severe biases in parameter estimates, which can then detrimentally impact on an 
analysis. In practice, intervals tend to be constructed from underlying classical data 
(e.g.~see Section \ref{sec:realDataAnalysis}), using a known process and where $m$ is 
also known. This implies that the generative model is a more natural construction than 
the descriptive model, and with parameters that more directly relate to the observed data.

While this analysis has assumed uniformity of the generative process (\ref{eqn:sugar}) in 
order that the descriptive model is obtained as  $m\to\infty$, and hence that the parameter 
estimates between the two models can be directly compared, the same principles of 
interpretation and bias occur regardless of the generative model. The parameters are 
simply less directly comparable with each other.

%%%%%%%%%%%%%%%%%%%%%%%%%%%%%%%%%%%%%%%%%%%%%%%%%
\subsection{Analysis of credit card data}
%%%%%%%%%%%%%%%%%%%%%%%%%%%%%%%%%%%%%%%%%%%%%%%%%
\label{sec:realDataAnalysis}

\begin{figure}[t]
\centering
\includegraphics[scale=0.41]{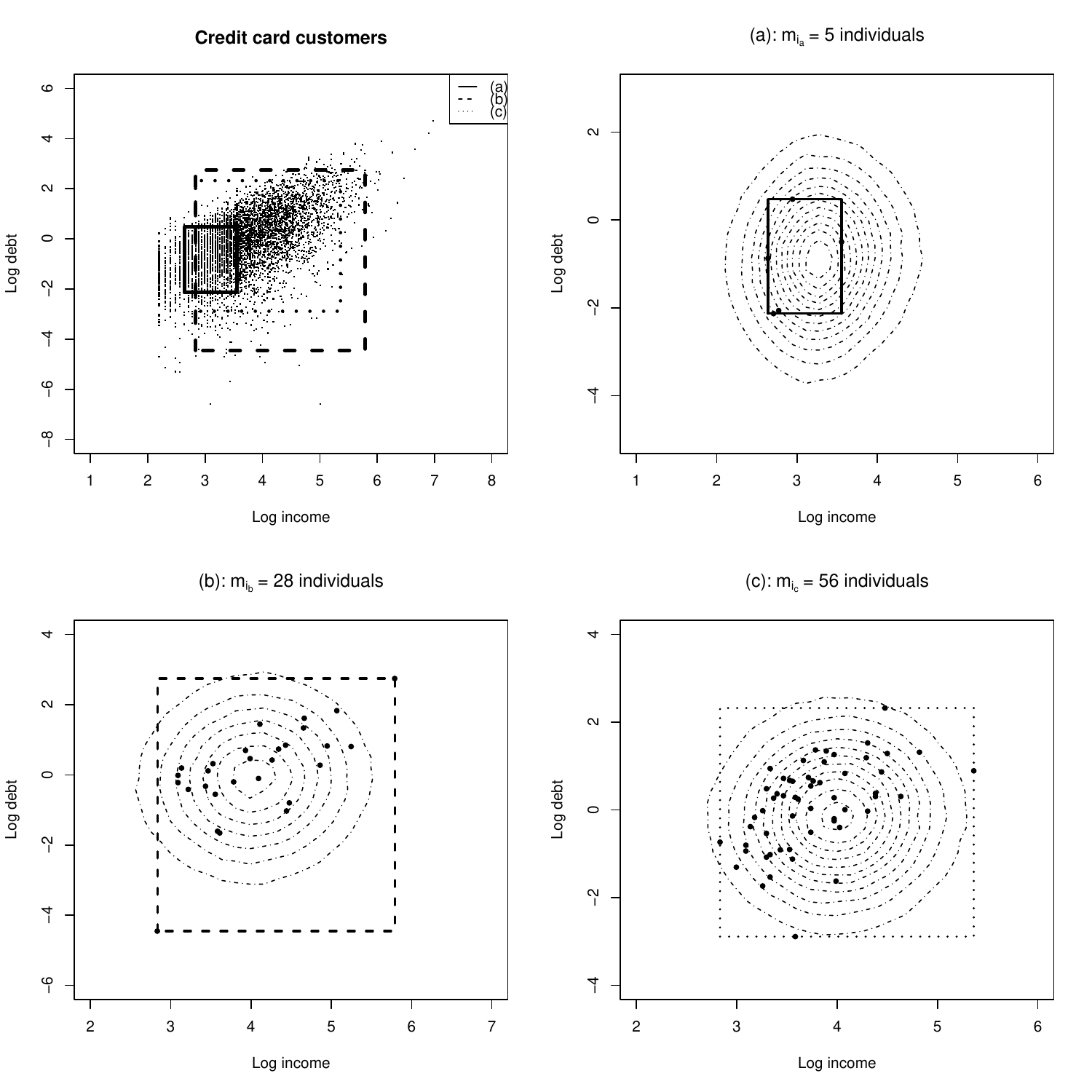}
\protect\caption{\label{fig:customer-data}
\small Log income and log credit card
debt in thousands of US\$ for $5,000$ customers. Panels illustrate 
rectangle-valued observations constructed from three groups of customers comprising (a) 
$m_{i_a}=5$, (b) $m_{i_b}=28$ and (c) $m_{i_c} = 56$ individuals.
The contours in the last three panels
indicate the predictive distributions of individuals for each
group conditional on the corresponding rectangle-valued observations, based on the generative model.}
\afterpage{\FloatBarrier}
\end{figure}

The data (available in the SPSS package \textit{customer.dbase}) comprise log income 
and log credit card debt in thousands of US\$ of 5,000 credit card customers.
In a previous analysis using descriptive models by \citeN{Brito2012}, 
these data were aggregated into random bivariate rectangles by  stratifying individuals 
according to gender, age category (18-24, 25-34, 35-49, 50-64, 65+ years old), 
level of education (did not complete high
school, high-school degree, some college, college degree, undergraduate
degree+), and designation of primary credit card (none, gold, platinum,
other). This leads to 192 non-empty groups, each producing a random rectangle 
$[x_{i1}]\times[x_{i2}]$ constructed by the intervals bounded by
the minimum and maximum observed values on log income and log credit card
debt. 

The data are illustrated in Figure~\ref{fig:customer-data}, along with the underlying 
data and constructed random rectangles for three of the 192 groups, containing (a) 
$m_{i_a}=5$, (b) $m_{i_b}=28$ and (c) $m_{i_c}$=56 individuals. The number of individuals in all 
groups varies greatly (from 5 to 56), and it is noticeable that the distribution of 
individuals within each group comes from a non-uniform distribution. As a result, the 
usual uniformity assumption of descriptive models for random rectangles is clearly 
inappropriate. The generative model is more suited to dealing with these heterogeneous 
rectangle-valued data containing complex intra-rectangle structures.

Given the clear non-uniformity within each group $i$, we assume that the underlying data 
are Gaussian with group-specific means and covariances. That is
\[
    (x_{i1},x_{i2})\sim N_2(\mu_i,\Sigma_i)
\]
for $i=1,\ldots,n=192$, where $\mu_i=(\mu_{i1},\mu_{i2})$ and 
$\Sigma_i=\mbox{diag}(\sigma^2_{i1},\sigma^2_{i2})$. Note that we choose to model log 
income and log credit card debt as uncorrelated, despite there being some visual evidence 
of positive correlation in the data underlying each random rectangle. It is worth briefly 
explaining this decision in detail. For a small number of latent data points $m_i$, it is 
possible for a single point to determine both upper (or lower) ranges of the random 
rectangle, and the probability of this occurring increases as the correlation of the 
underlying data increases. So in principle, there is some information about the 
correlation structure of the underlying data available through the associated random 
rectangle. However, for groups with larger $m_i$, the upper and lower ranges of the 
random rectangles are more likely to be determined by four individual data points, in 
which case it is not  then possible to discern the underlying correlation structure. 
Although we have several groups with small numbers of latent data points (e.g.~$m_{i_a}=5$), 
in principle allowing their correlation to be estimated, note that the same random 
rectangles will arise whether the latent data are positively or negatively correlated. That 
is, the correlation parameter is non-identifiable from the observed rectangle data.
As such, we proceed without attempting to estimate this parameter, despite information on 
the magnitude of the correlation being available in principle for some groups.

We model the group-specific (local) parameters as
\begin{eqnarray}
    (\mu_{i1},\mu_{i2}) &\sim& N_2(\theta_1,\theta_2,\lambda^2_1,\lambda^2_2,\rho_\mu),\nonumber \\
    \log\sigma^2_{ij} &\sim& N(\eta_j,\epsilon^2_j),\label{eq:customer_prior}
\end{eqnarray}
for $j=1,2$ and $i=1,\ldots,192$. The integration in the generative 
model~(\ref{eq:lik-g-mix-hyper}) is achieved using Gauss-Hermite quadratures with 
$20^4$ nodes to integrate over the four parameters.

\begin{table}[t]
\centering
\begin{tabular}{cc|ccccccccc}
& & $\theta_1$ & $\lambda_1^2$ & $\theta_2$ & $\lambda_2^2$ & $\rho_\mu$ 
& $\eta_1$ & $\epsilon^2_1$ & $\eta_2$ & $\epsilon^2_2$\\
\hline\multirow{3}{*}{Generative} & MLE & 3.76 & 0.13 & -0.36 & 0.21 & 0.90 & 
-1.20 & 0.48 & 0.41 & 0.09\\
& \multirow{2}{*}{95\% CI} & 3.70 & 0.10 & -0.44 & 0.13 & 0.83 & -1.31 & 0.35 & 0.34 & 0.04\\
& & 3.82 & 0.17 & -0.26 & 0.29 & 0.98 & -1.09 & 0.61 & 0.48 & 0.13\\
\hline\multirow{3}{*}{Descriptive} & MLE & 3.79 & 0.17 & -0.42 & 0.52 & 
0.57 & 0.02 & 0.20 & 0.82 & 0.09\\
& \multirow{2}{*}{95\% CI} & 3.74 & 0.13 & -0.53 & 0.41 & 0.47 & -0.04 & 0.16 & 0.78 & 0.07\\
& & 3.85 & 0.20 & -0.32 & 0.62 & 0.67 & 0.08 & 0.24 & 0.87 & 0.11\\ 
\end{tabular}
\protect\caption{\label{tab:customer-est}
\small Maximum likelihood estimates and 95\% asymptotic confidence intervals for the 
parameters of the generative and descriptive models for the credit card dataset.}
\end{table}

Maximum likelihood estimates and 95\% confidence intervals for each model parameter are 
illustrated in Table~\ref{tab:customer-est} for both generative and descriptive models.
Similar to the results for the
simulated examples, the point estimates of location ($\theta_{1}$ and $\theta_{2}$)
are broadly insensitive to the choice of model, however the estimated values
for many of other parameters differ between the two models. 
Most importantly, the estimated values of $\rho_\mu$ are considerably larger for the 
generative model ($\hat{\rho}_\mu=0.9040$) compared to the descriptive model 
($\hat{\rho}_\mu=0.5695$). While both of these indicate a positive relationship between 
income and credit card debt, which is evident in the underlying data in 
Figure~\ref{fig:customer-data}, there is a clear difference in the strength of that 
relationship. The descriptive model results in a weaker estimated value in the correlation
because it does not take the noisy data generating process into account.
While we suspect that the generative model may be the more accurate of the two given the 
data generating procedure used to construct the random rectangles, in terms of drawing inferential 
conclusions about the underlying data, it is critical that we are certain in this regard.

\begin{figure}[t]
\centering
\includegraphics[scale=0.46]{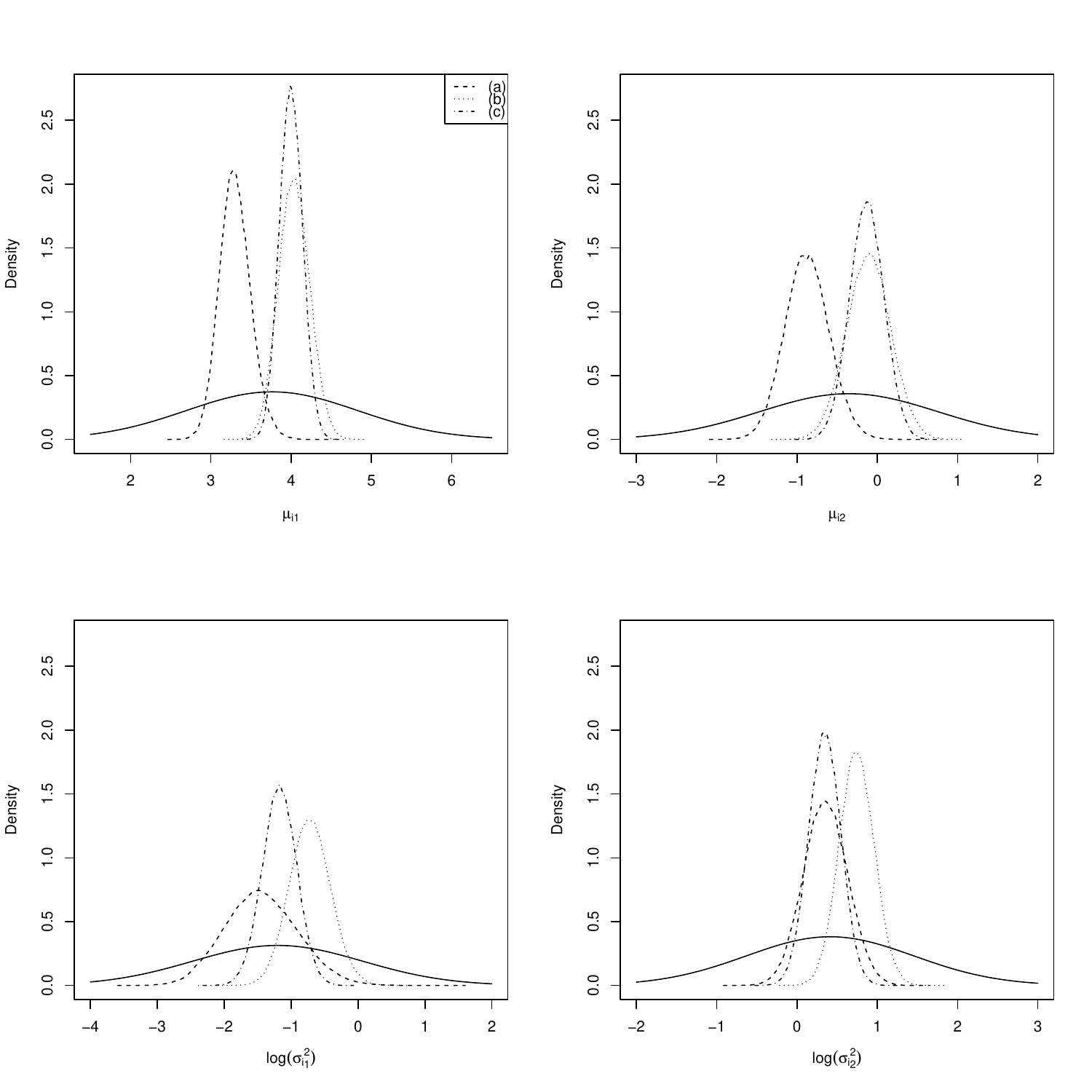}
\protect\caption{\label{fig:customer-posterior}
\small Estimated marginal posterior distributions of the local parameters 
$\mu_{i1}$, $\mu_{i2}$, $\sigma^2_{i1}$ and $\sigma^2_{i2}$ 
associated with the three groups (a)--(c) shown in Figure~\ref{fig:customer-data}. 
Solid lines correspond to the prior distributions for local parameters}
\afterpage{\FloatBarrier}
\end{figure}

\begin{figure}[t]
\centering
\includegraphics[scale=0.43]{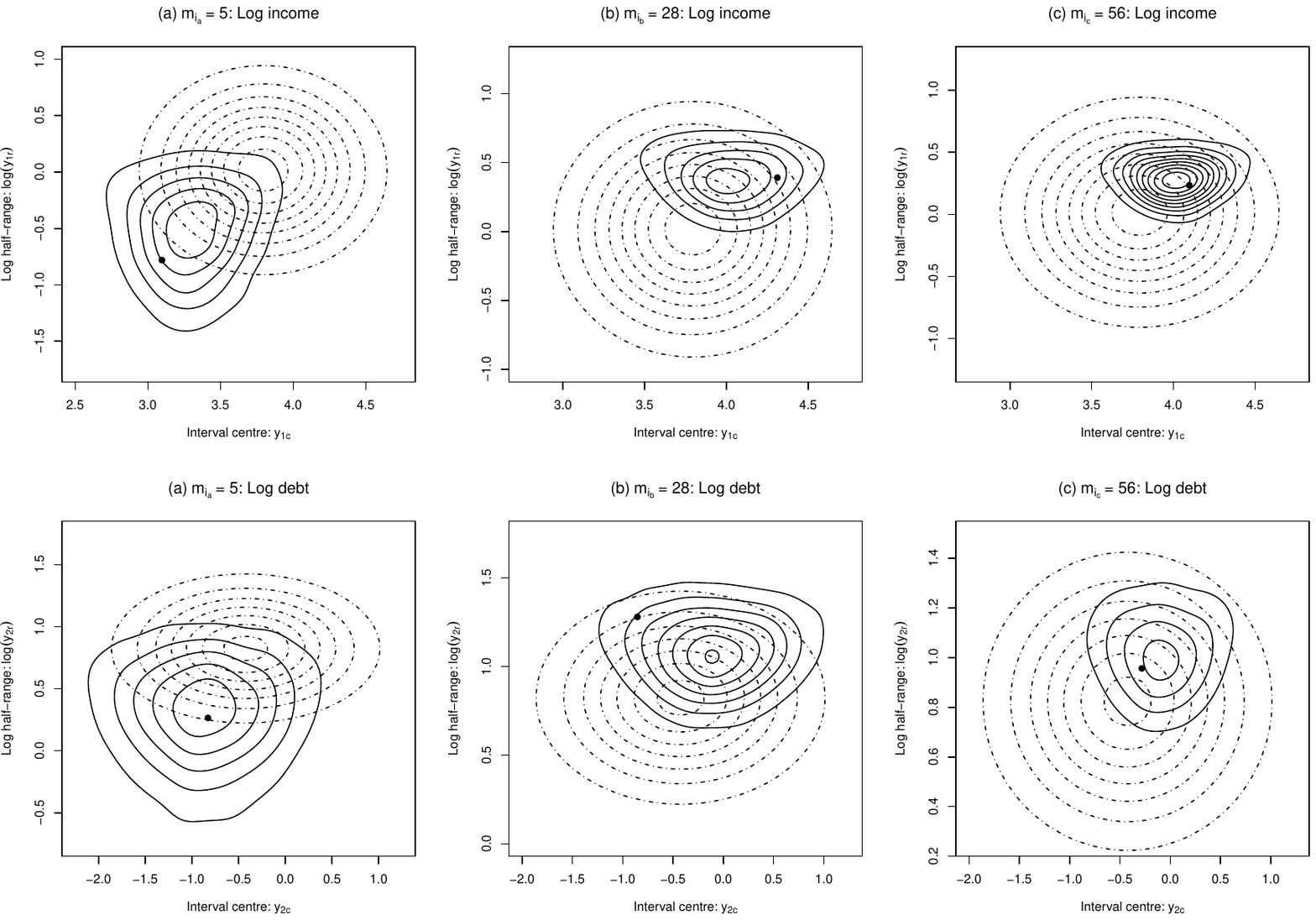}
\protect\caption{\label{fig:customer-modelcheck}\small Posterior predictive distribution of a random 
rectangle $[y_1]\times[y_2]$ for each of the groups (a)--(c) (left column to right) in 
Figure~\ref{fig:customer-data}. Columns illustrate the marginal random intervals of log income 
($[y_1]$, top row) and log debt ($[y_2]$, bottom row) with each interval $[y_j]=[y_{j1},y_{j2}]$ 
expressed in interval centre and half-range form $(y_{jc}, y_{jr})=((y_{j1}+y_{j2})/2, (y_{j2}-y_{j1})/2)$ 
for $j=1, 2$. Solid and dashed lines indicate predictive distributions of generative and descriptive 
models, respectively. The dot indicates the observed interval $[x_{i1}]\times[x_{i2}]$ used for model fitting.}
\afterpage{\FloatBarrier}
\end{figure}

For the generative models, 
the distribution of the local 
parameters $(\mu_{ij},\sigma_{ij}^2)$ for each rectangle-valued observation can be computed by 
empirical Bayesian methods (previously these parameters were integrated out for the optimisation 
in Table \ref{tab:customer-est}). The prior distribution for the local parameters is the global 
distribution~(\ref{eq:customer_prior}) with its parameter values given by the estimates
in Table~\ref{tab:customer-est}, and the likelihood function is the local density function 
of one observed rectangle.
The resulting marginal posterior distributions for the parameters of the observed rectangles (a)--(c) 
(Figure~\ref{fig:customer-data})
are shown in Figure~\ref{fig:customer-posterior}. Compared to the prior (solid line)
the parameters are well informed, even for rectangle (a) with $m_{i_a}=5$ observations, 
with the level of precision increasing with the number of individuals within each rectangle.

Goodness-of-fit for both descriptive and generative models can be evaluated through model 
predictive distributions of random rectangles, in addition to predictive distributions for 
individual data points for the generative model.
In the latter case, based on the posterior distributions of the local parameters in 
Figure~\ref{fig:customer-posterior}, the predictive distributions of individual data points 
within the random intervals (a)--(c), conditional on observing the associated random 
interval,  are shown in Figure~\ref{fig:customer-data}.
While the predictive distributions are marginally independent due to the model specification, 
their coverage describes the observed data well. For group (a) the predictive distribution 
covers a wider region than the observed rectangle, as this rectangle is constructed from only 5 
individuals. As the number of individuals increases in groups (b) and (c), the predictive 
regions more closely represent the region of the observed rectangle, indicating that the 
generative model has the ability to correctly account for the different numbers of individuals 
used to construct each rectangle. 
The predictive distribution for group (b) individuals  also indicates some robustness to the 
two outliers that completely define the observed rectangle. This occurs as the model correctly 
accounts for the fact that rectangle (b) is constructed from half the number of observations 
used to construct the rectangle of group (c), even though both rectangles are roughly the same 
size.

The predictive distributions of random rectangles  for groups (a)--(c) are illustrated in 
Figure~\ref{fig:customer-modelcheck} for both descriptive (dashed lines) and generative (solid lines) 
models. Shown are the bivariate predictive distributions of interval centre and log half-range, for both  
log income (top row) and log debt (bottom row). The dot indicates the observed interval. Under the
generative model, these distributions are obtained directly from the predictive distributions for 
individuals (Figure~\ref{fig:customer-data}).

In all cases, the predictive distributions of the generative model more accurately, and more precisely 
identify the location of the observed data. This is particularly the case in group (a) in which the 
descriptive model is clearly indicating a lack of model fit. The predicted interval for log debt in 
group (b) is not fully centred on the observed interval, as the model attempts to 
account for the unlikely (under the model) construction of the observed interval by outliers 
(Figure~\ref{fig:customer-data}). However, the observed data are still well predicted under the 
generative model.
The overall fit to the observed data is better under the generative model than the descriptive model, 
indicating that it more accurately describes the complexities of the observed data.

%%%%%%%%%%%%%%%%%%%%%%%%%%%%%%%%%%%%%%%%%%%%%%%%%
\subsection{Robustness to model mis-specification}
%%%%%%%%%%%%%%%%%%%%%%%%%%%%%%%%%%%%%%%%%%%%%%%%%
\label{sec:robustness}

Until now we have
focused on the setting where both the underlying model $f(x|\theta)$ and the data aggregation function $\varphi(\cdot)$ are known. When the true $f(x|\theta)$ is not known, this is the standard setting of statistical model mis-specification.
There are two possible mis-specification scenarios in which $\varphi(\cdot)$ may not be known.
Firstly, $\varphi(\cdot)$ may have been mis-reported, so that e.g.~different quantiles were used to  construct intervals from data than were modelled in $\varphi(\cdot)$. 
Secondly, $\varphi(\cdot)$ may simply be unknown, so that the task is to analyse data which has quantiles $\underline{X}$ and $\overline{X}$, but where it is not known what quantiles these are.
In this second scenario, at best the generative likelihood could be integrated over all possible $\varphi(\cdot)$ with respect to some prior measure.
 It is possible that with  informative prior information this could yield {\em some} viable inference, but this would likely be circumstantial and not ideal.

The following analysis aims to examine the effect of mis-specifying the fitted model and $\varphi(\cdot)$.
We consider data $x_{1:m}$, with $m=1000$, generated independently from either normal or uniform distributions, both with mean $\mu = 0$ and standard deviation $\sigma = 2$. 
To evaluate the effect of outliers, we create additional datasets which replace $5\%$ of each original dataset by observations drawn from the (normal or uniform) generating distribution with $\mu=0$ and $\sigma=5$.
For each dataset, observed intervals are constructed through the aggregation function $\varphi_i := \varphi_{i,m-i+1}(x_{1:m})= [x_{(i)},x_{(m-i+1)}]$, with $i=1$ and $i=250$ corresponding to constructing intervals based on sample minimum/maximum and the 1st/3rd quartiles.
For each of these interval datasets, we fit both normal and uniform models, and assess
the impact of knowing the aggregation function $\varphi(\cdot)$ by supposing the observed intervals are obtained from $\varphi_i$ with $i=1,50, 100, \ldots, 450$.

Figure~\ref{fig:robustness} shows boxplots of 500 replicate maximum likelihood estimates of $\mu$ (top row) and $\log(\sigma)$ (middle row), when the true underlying data distribution is normal, as a function of the aggregation function $\varphi_i$ used to fit the model. The true interval aggregation function is $\varphi_1$ ($i=1$; left two columns) and $\varphi_{250}$ ($i=250$; right two columns), and use of this is indicated by the shaded boxplots. In each panel
the horizontal line denotes the true parameter value and the rightmost boxplot shows the impact of using the true aggregation function with the outlier datasets.

The mean ($\mu$; top row)  is consistently well estimated, regardless of the model being fitted or the aggregation function. This is not surprising, as changing $\varphi_i$ affects the scale of the intervals and not their location.
However for $\log\sigma$ (middle row), when the model being fitted is correct (columns 1 and 3), using generative model aggregation functions that use narrower (wider) quantiles than actually used to construct the empirical interval, leads to larger (smaller) estimates of $\sigma$.
This observation also holds when fitting the uniform model, although the picture is distorted due to the model mis-specification (fitting a uniform model to normal data).
That is, when the model is correctly specified under the true data aggregation process, the maximum likelihood estimates are accurate.

A goodness-of-fit check between predicted and observed intervals would not reveal problems in any of the above analyses:
both models are in the location-and-scale family, and so each can describe all observed interval datasets well. However, differences can easily be seen by comparing to the original underlying data.
The bottom row of Figure~\ref{fig:robustness} denotes qq-plots of the fitted model ($y$-axis) against the original sample $x_{1:m}$ (in practice, this would be constructed from a sub-sample of the data when dealing with very large datasets).
In all cases, only when the model and aggregation function are correct does the qq-plot align on the $y=x$ axis. Deviation away from this indicates that  either model or $\varphi(\cdot)$, or both, are incorrect. As the data aggregation function will typically be known, this would usually suggest that it is the fitted model that needs further requirement. However, when the data aggregation function is mis-specified then it may be difficult to identify a fitted model that, in combination with the mis-specified $\varphi(\cdot)$, will fit the data well. A failure to improve on a model's goodness-of-fit when modifying the model, could therefore indicate that the data aggregation function is mis-specified.

In the presence of outliers in the original dataset (rightmost boxplots in each panel), as might be expected, constructing intervals that are robust to these (e.g.~using the 1st/3rd quartiles) produce more sensible results than less robust intervals (e.g.~using min/max). 
Qualitatively similar conclusions to the above can be drawn when the true data generating process is uniform rather than normal (see Supplementary Information).

\begin{landscape}
\begin{figure}
\vspace{-7mm}
\center
$
\begin{array}{cccc}
\:\:\:\:\:\:\textrm{ Normal model, } \varphi_1&
\:\:\:\:\:\:\textrm{ Uniform model, } \varphi_1&
\:\:\:\:\:\:\textrm{ Normal model, } \varphi_{250}&
\:\:\:\:\:\:\textrm{ Uniform model, } \varphi_{250}\\
\includegraphics[width=0.33\textwidth]{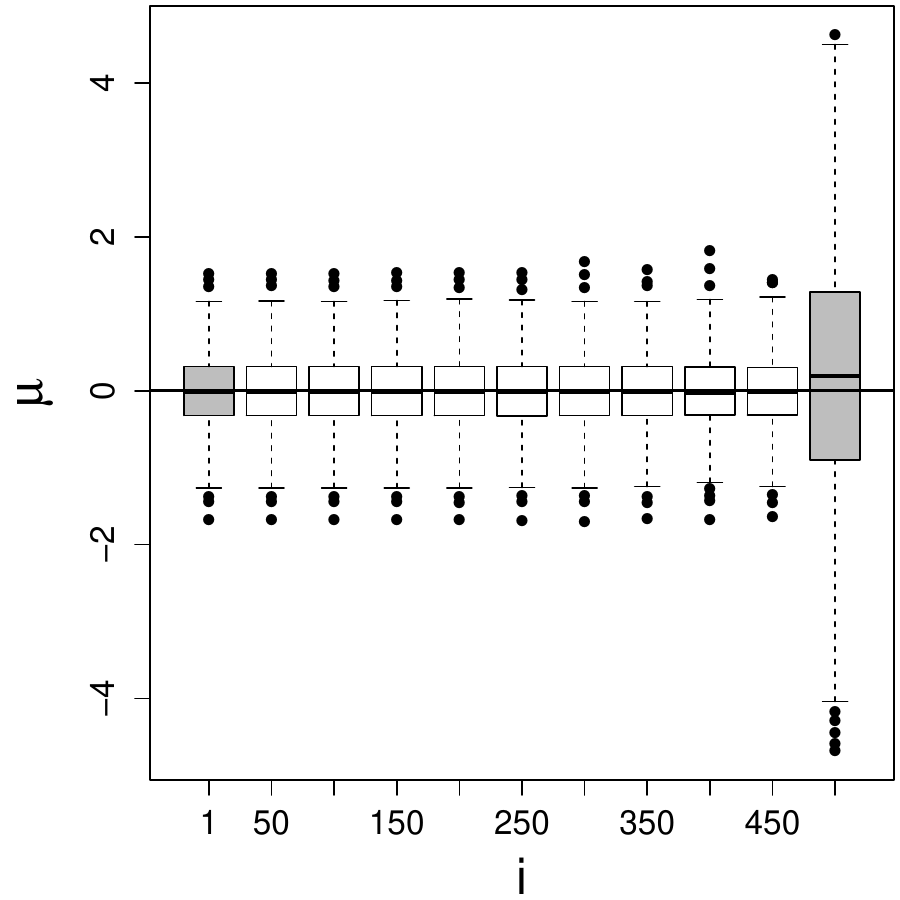} &
\includegraphics[width=0.33\textwidth]{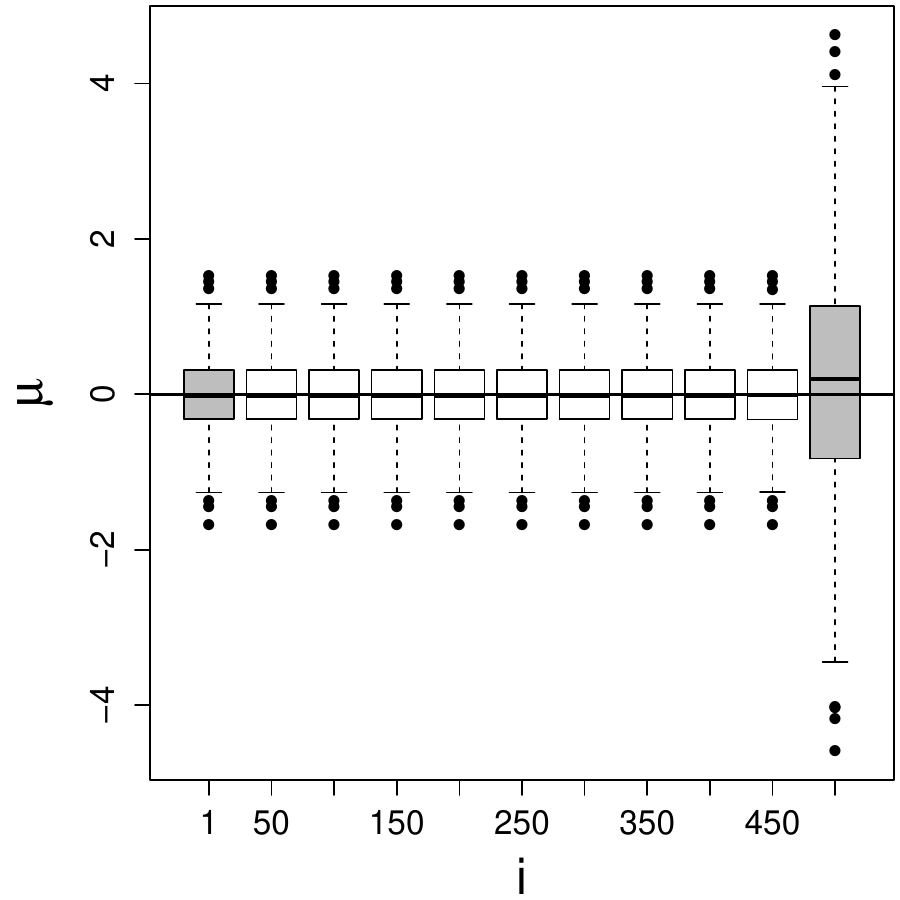} &
\includegraphics[width=0.33\textwidth]{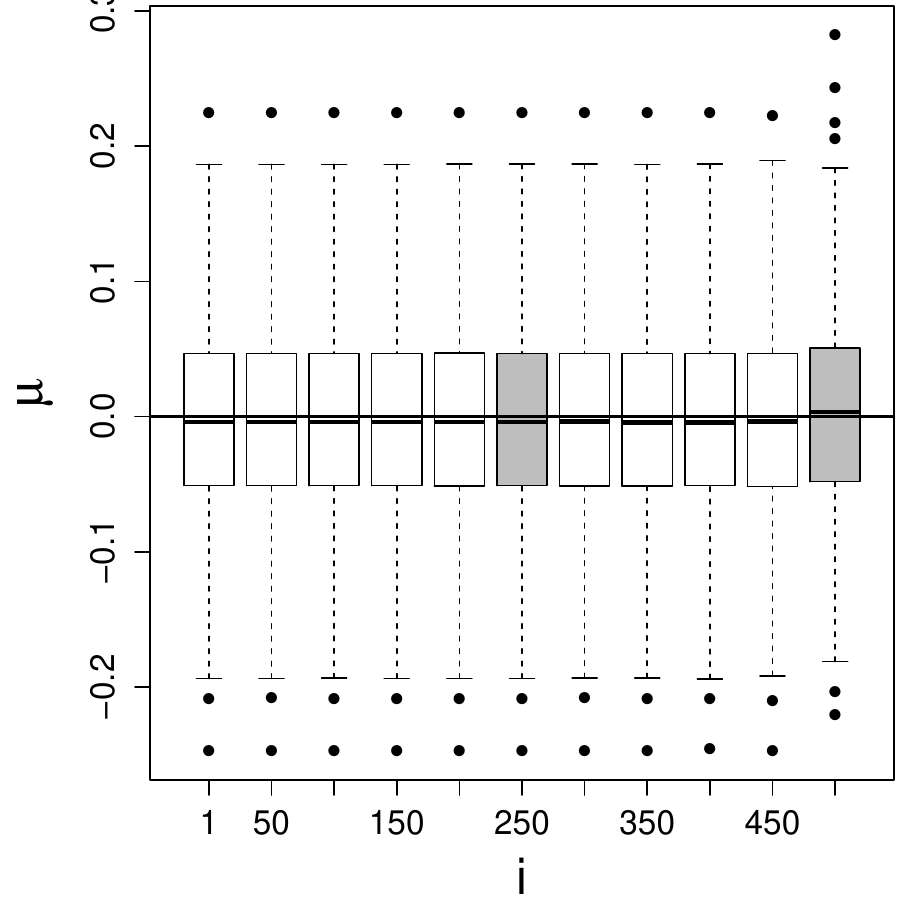} &
\includegraphics[width=0.33\textwidth]{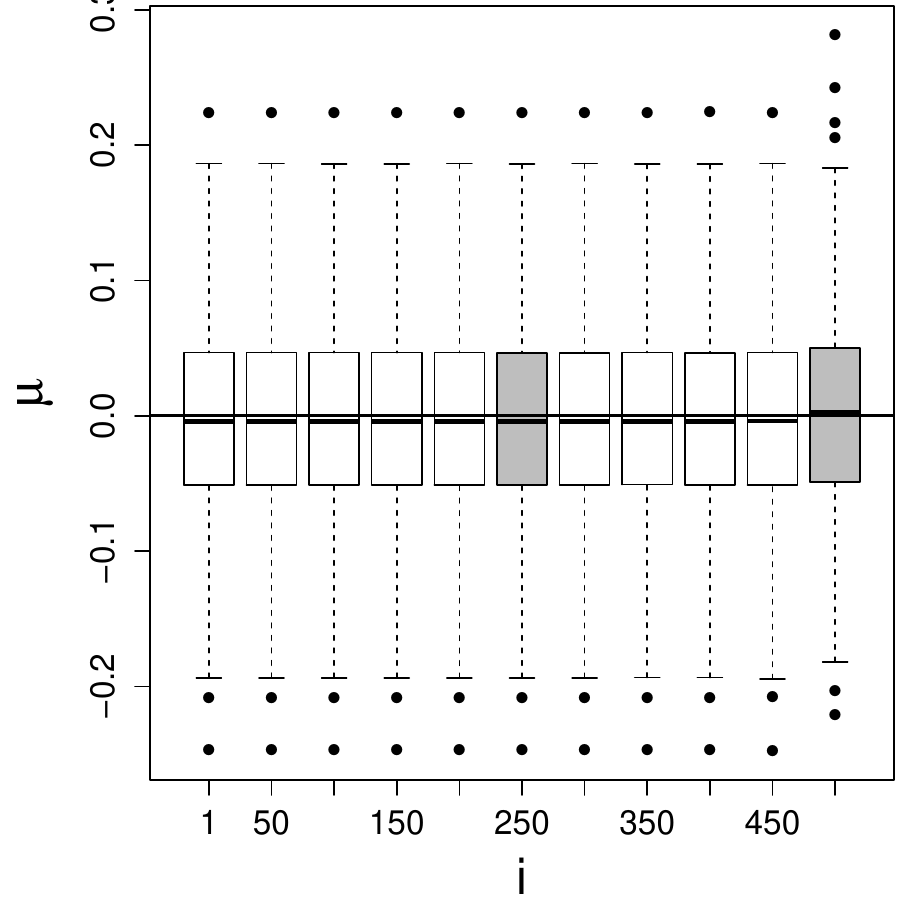} 
\\
\includegraphics[width=0.33\textwidth]{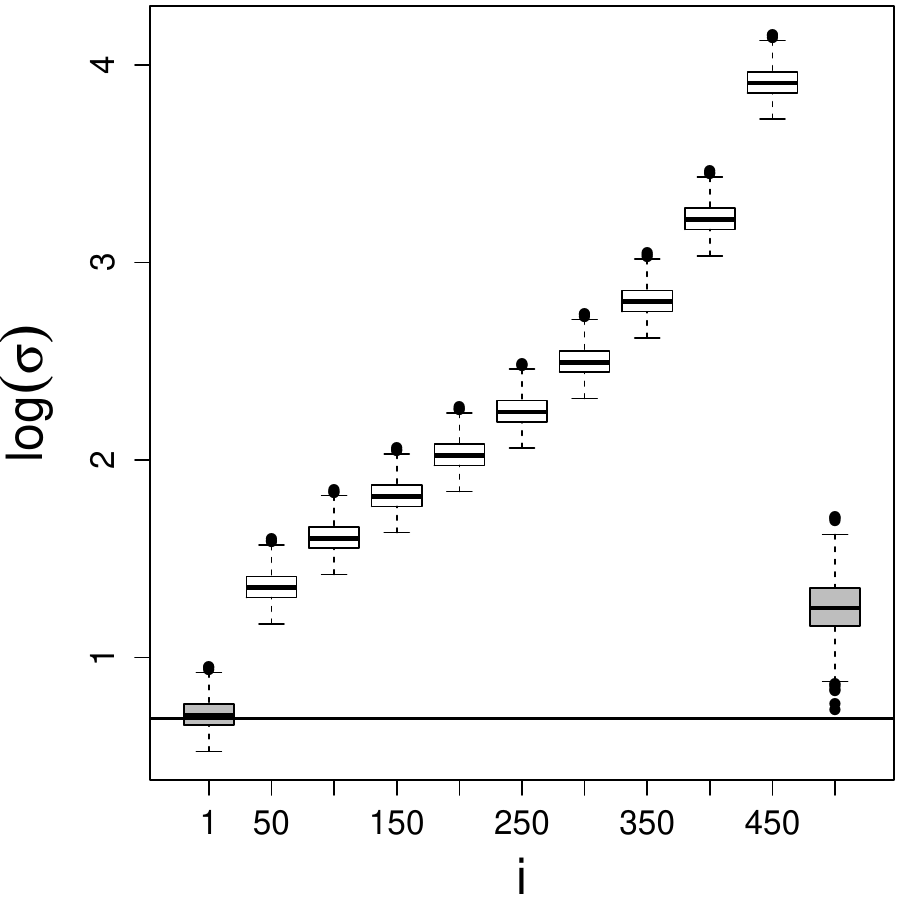} &
\includegraphics[width=0.33\textwidth]{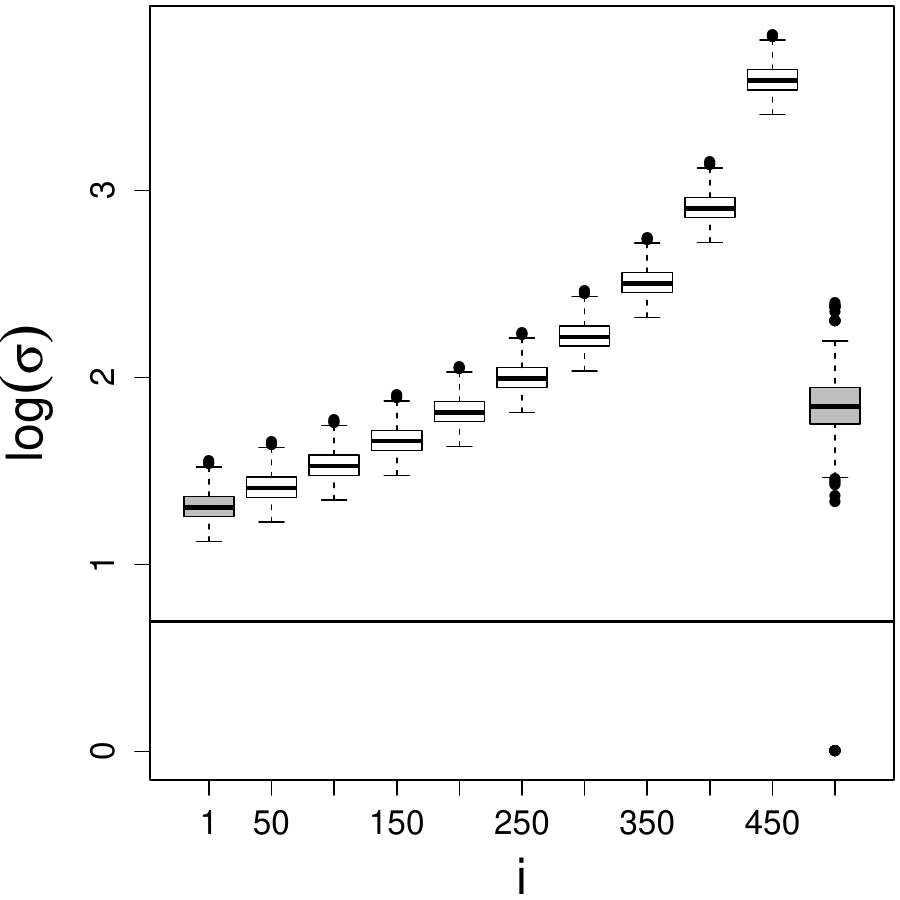} &
\includegraphics[width=0.33\textwidth]{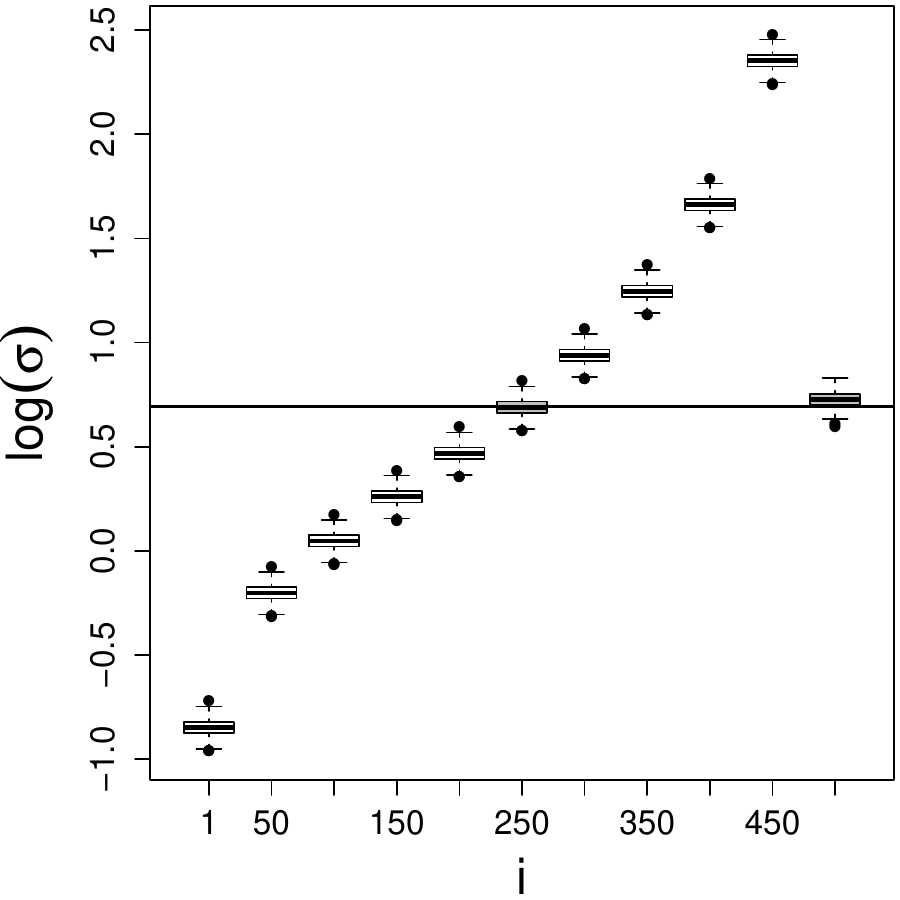} &
\includegraphics[width=0.33\textwidth]{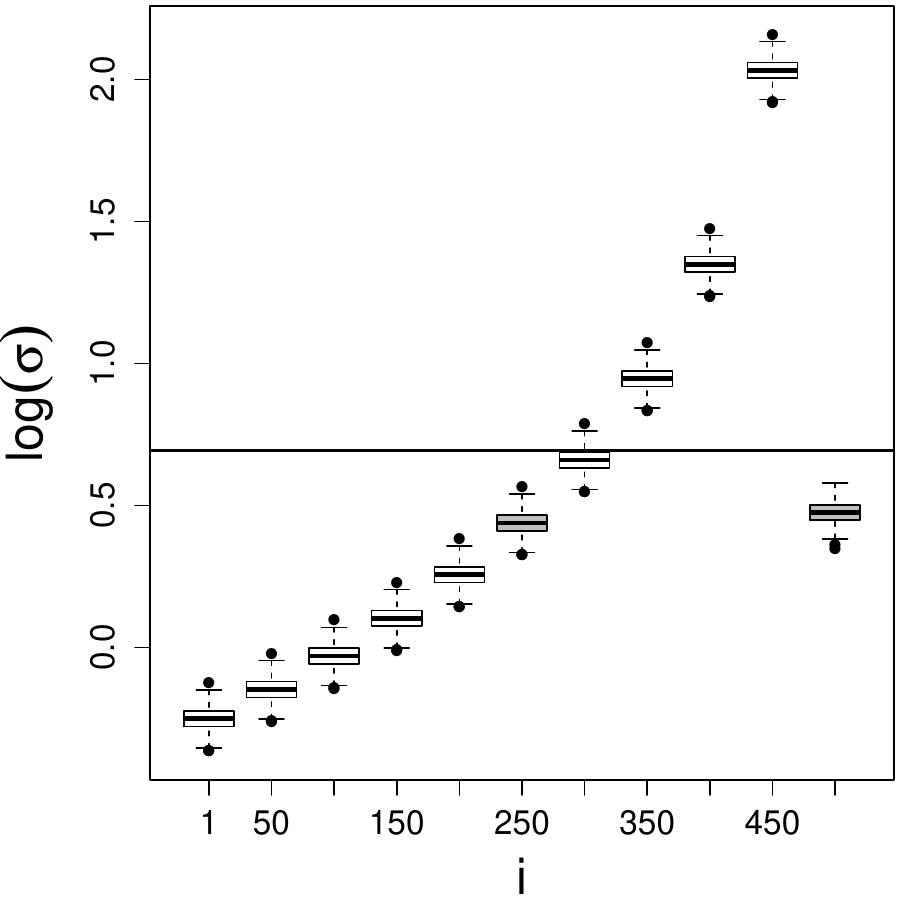} 
\\
\includegraphics[width=0.33\textwidth]{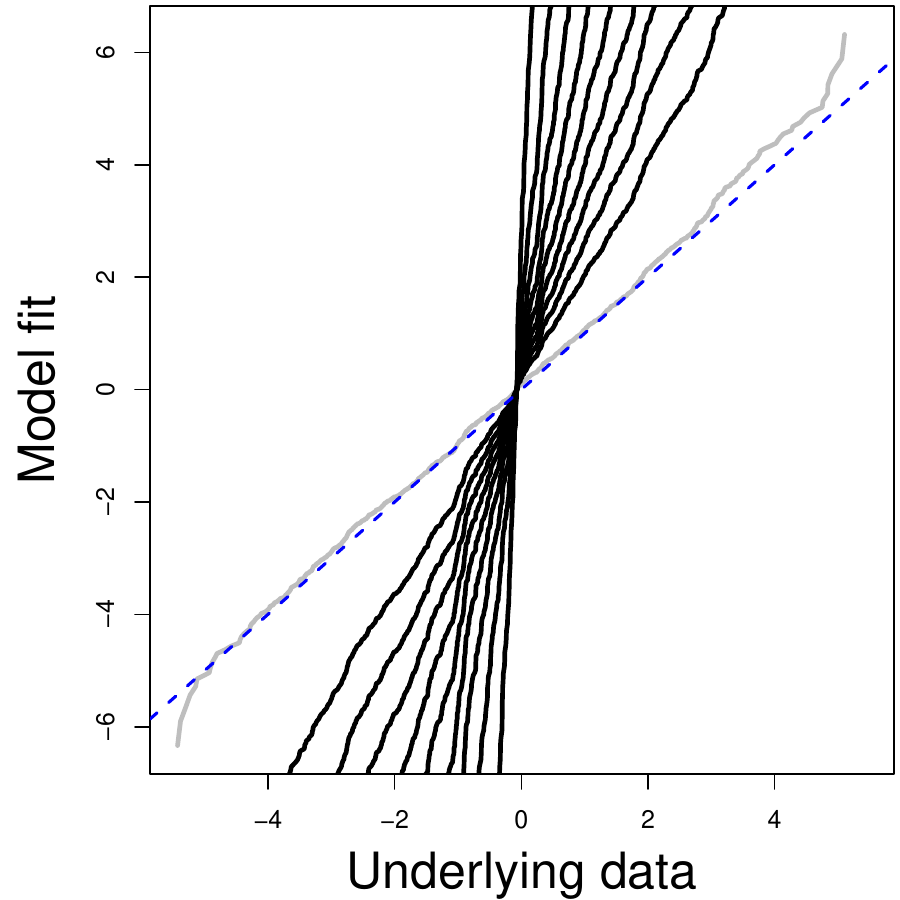} &
\includegraphics[width=0.33\textwidth]{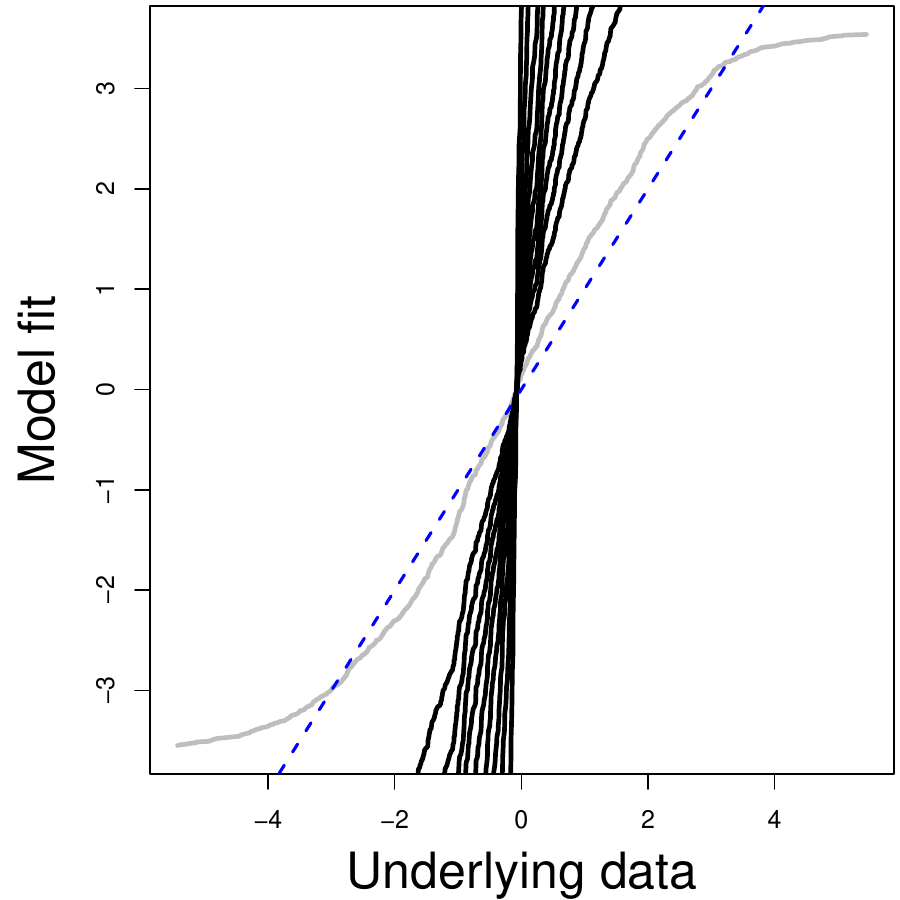} &
\includegraphics[width=0.33\textwidth]{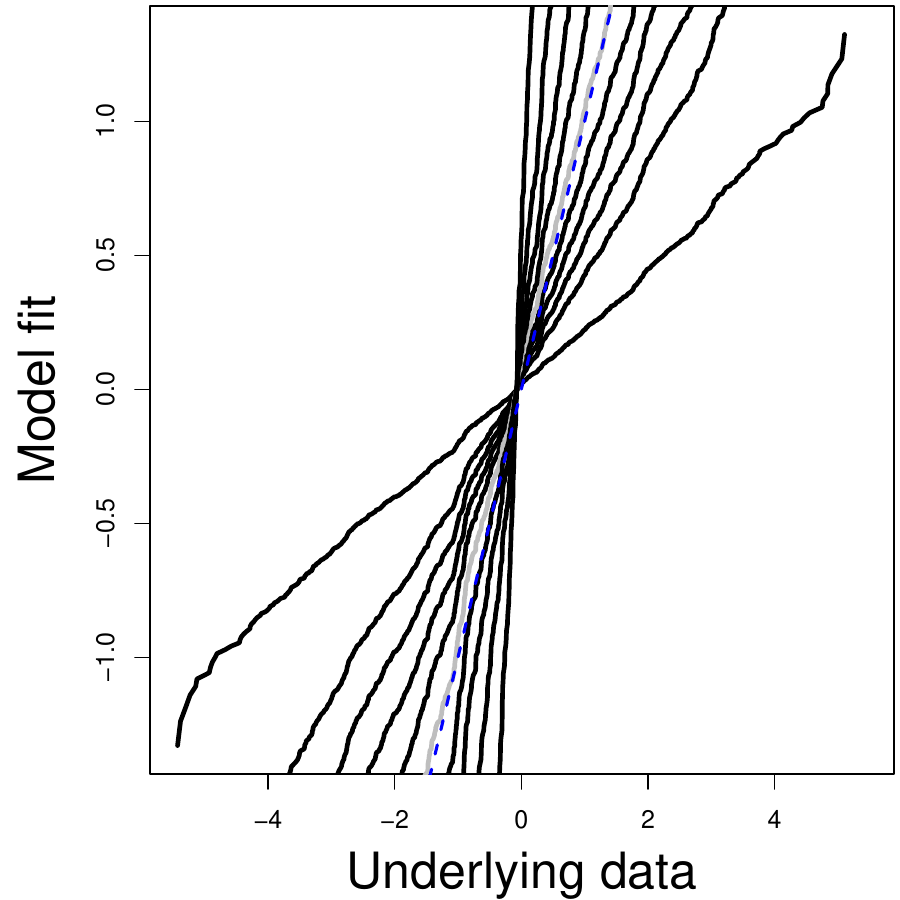} &
\includegraphics[width=0.33\textwidth]{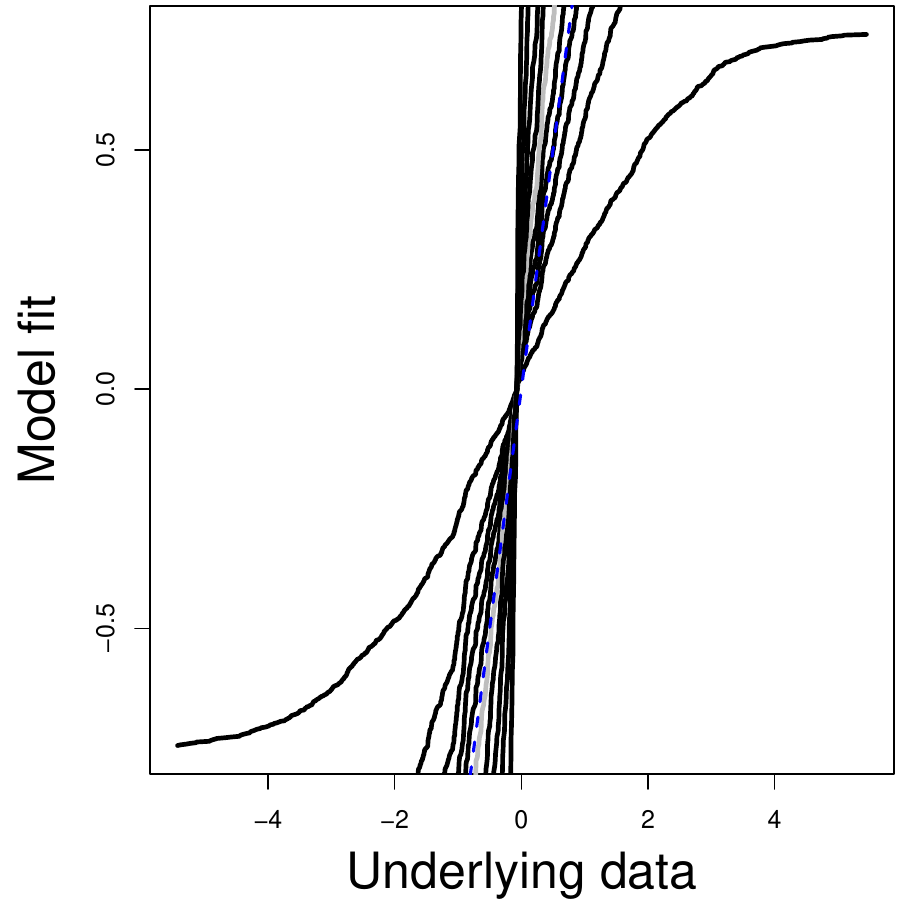} 
\end{array}
$
\caption{\small 
Boxplots of 500 replicate maximum likelihood estimates of $\mu$ and $\log\sigma$ under a $N(0,2^2)$ true data generating process with $m=1000$, and assuming data aggregation function $\phi_i$, $i=1,50, 100, \ldots, 450$. The true aggregation functions are $\phi_1$ (left two columns) and $\phi_{250}$ (right two  columns). The models fitted are the normal (columns 1 \& 3) and uniform (columns 2 \& 4) distributions. 
In each panel, the rightmost boxplot  indicates the outcome using the dataset with  $5\%$ outliers.  
Bottom row shows quantile-quantile curves of the fitted model ($y$-axis) versus the empirical underlying data quantiles ($x$-axis). 
Grey curves indicate use of the correct $\varphi(\cdot)$ function. Dashed line denotes $y=x$.
}
\label{fig:robustness}
\end{figure}
\end{landscape}

%%%%%%%%%%%%%%%%%%%%%%%%%%%%%%%%%%%%%%%%%%%%%%%%%
%%%%%%%%%%%%%%%%%%%%%%%%%%%%%%%%%%%%%%%%%%%%%%%%%
\section{Discussion\label{sec:discuss}}
%%%%%%%%%%%%%%%%%%%%%%%%%%%%%%%%%%%%%%%%%%%%%%%%%
%%%%%%%%%%%%%%%%%%%%%%%%%%%%%%%%%%%%%%%%%%%%%%%%%

Current techniques for modelling random intervals ($p$-\textcolor{red}{hyper-rectangle}s) are based 
on constructing models directly at the level of the interval-valued data 
(e.g.~\shortciteNP{Arroyo2010,Le-Rademacher2011,Brito2012}). These approaches are additionally 
based on the assumption that the unobserved individual data points from which the interval is 
constructed are uniformly distributed within the interval. As we have demonstrated in 
Section~\ref{sec:examples}, using these descriptive methods when the data are constructed 
from underlying individual data points, which is typical in practical applications, can 
result in misleading and biased parameter estimates and therefore unreliable inferences.

In this article we have established
the distribution theory for interval-valued
random variables which are constructed bottom-up from distributions of latent real-valued 
data and aggregation functions used to construct the random intervals. These generative 
models explicitly permit the fitting of standard statistical models for latent data 
points through likelihood-based techniques, while accounting for the manner in which the 
observed interval-valued data are constructed. This approach directly accounts for the 
non-uniformity of latent individual data points within intervals, and provides a natural way to 
handle the differing number of latent data points $m_i$ within each random interval, 
which is again typical in practice. 
The method as presented is fully parametric, although extending these ideas to the non-parametric framework would be of some interest (e.g.~\shortciteNP{jeon+ap15}).

By deriving a descriptive model as the limiting case of a generative model (i.e.~as 
$m_i\to\infty$ for each $i$), we have demonstrated that these descriptive models 
have an explicit underlying generative model interpretation. In turn this indicates why 
inferences from descriptive models may be potentially misleading in practice.

In order to evaluate the integrated generative likelihood function (\ref{eq:lik-g-mix-hyper}) 
for the unimodal distributions considered in Section~\ref{sec:examples}, we have used 
Gaussian quadrature methods. This technique will be less useful when integrating over 
more than 6 parameters \cite{Evans1995}, or when there are strong dependencies between 
local parameters. In these cases, approximate MLE's can be obtained using e.g.~Monte 
Carlo maximum likelihood estimation \cite{Geyer1992} or Monte Carlo expectation 
maximization techniques \cite{Wei1990}, or in the Bayesian framework, Gibbs sampling 
\cite{Geman1984} or pseudo-marginal and other likelihood-free Monte Carlo methods \shortcite{Andrieu2009,sisson+fb18}.

In order to construct the likelihood function (\ref{eq:lik-g-mix-hyper}) for 
$p$-\textcolor{red}{hyper-rectangle}s we assumed independence among all margins in local distributions to avoid 
the $2p$-th order mixed-differentiation of 
$F_{[\boldsymbol{X}]}(\underline{x}_{1},\overline{x}_{1},\ldots,\underline{x}_{p},\overline{x}_{p})$.
Although this differentiation may be achieved using symbolic computation software, the
resulting likelihood functions are complex even when $p=2$ 
(see the Supplementary Information), 
and the alternative of numerical differentiation would be highly computational.
However, this independence assumption does not hold if there is priori information on the 
dependence structure within each latent data point $\boldsymbol{x}$. As pointed out by 
\citeN{Billard2006}, this is often the case because the structure of symbolic data might
determine inherent dependencies such as logical, taxonomic and hierarchical dependencies, 
but not statistical dependencies. In the generative model, those dependencies as well as 
statistical dependencies 
can be addressed simultaneously through the local distribution function 
$f(\boldsymbol{x}|\boldsymbol{\theta})$. However, without the marginal independence 
assumption, inference for these models can be challenging.

While our examples have primarily focused on minimum and maximum based data aggregation 
functions $\varphi(x_{1:m})$, there is clear interest in parameter estimation and 
inference for more robust order-based functions $\varphi_{l,u}(x_{1:m})$, as the resulting intervals 
will be less sensitive to outliers, as demonstrated in the study in Section \ref{sec:robustness}. The procedures for constructing the associated 
likelihood functions are analogous to those presented here, and Theorem \ref{thm:g2d} 
provides their limiting descriptive model counterpart. An additional practical question 
for inference using order-based aggregation functions is which order-based statistics 
to use. As this choice will impact on the efficiency of the resulting inference, it is an 
open question to understand what method of random interval construction would be optimal 
for any given analysis (e.g.~\shortciteNP{beranger+ls18}).

Finally, we have derived an approximation $\hat{L}$ of the likelihood function of the underlying data, $L(x_{1:m}|\boldsymbol{\theta})$, based on constructing random intervals or $p$-hyper-rectangles through the data aggregation function $\varphi(\cdot)$, so that $\hat{L}(\varphi(x_{1:m})|\boldsymbol{\theta})\approx L(x_{1:m}|\boldsymbol{\theta})$. Clearly there can be some information loss when moving from $x_{1:m}$ to $\varphi(x_{1:m})$. Understanding the quality of this approximation is important both for quantifying inferential accuracy, as well as guiding the design of the aggregation function (where possible) to increase the performance of an analysis.
This is the focus of current research.

\section*{Acknowledgements}

XZ is supported by the China Scholarship Council.
BB and SAS are supported by the Australian Research Council through the Discovery Project scheme (FT170100079) and the Australian Centre of Excellence for Mathematical and Statistical Frontiers in Big Data, Big Models and New Insights (CE140100049).

%%%%%%%%%%%%%%%%%%%%%%%%%%%%%%%%%%%%%%%%%%%%%%%%%
%%%%%%%%%%%%%%%%%%%%%%%%%%%%%%%%%%%%%%%%%%%%%%%%%
\section*{Appendix}
%%%%%%%%%%%%%%%%%%%%%%%%%%%%%%%%%%%%%%%%%%%%%%%%%
%%%%%%%%%%%%%%%%%%%%%%%%%%%%%%%%%%%%%%%%%%%%%%%%%

Detailed proofs to all lemmas and theorems are provided in the Supplementary Information. For {\tt arXiv.org} this is provided below.

%%%%%%%%%%%%%%%%%%%%%%%%%%%%%%%%%%%%%%%%%%%%%%%%%
\section{Constructing a measurable space}
%%%%%%%%%%%%%%%%%%%%%%%%%%%%%%%%%%%%%%%%%%%%%%%%%
\label{sec:intvl-measure}

We denote $\Omega$ as a sample space equipped with a $\sigma$-algebra
$\mathscr{F}$ and a probability measure $P(\cdot)$. In order to construct a measurable space of $\mathbb{I}$,
we identify those subsets of $\mathbb{I}$, which are equivalent to 
particular subsets of $\mathbb{R}^m$. 
A subset of interest is $\{[x^{\prime}]\subseteq[x]\} = \{[x'] :[x^{\prime}]\subseteq[x]\}$, 
which corresponds to the
collection of all intervals that are a subinterval of or equal to
$[x]$. This subset is the image of the event 
$
    \{[X]\subseteq[x]\}=\{\omega\in\Omega:[X](\omega)\subseteq[x]\}
$
on $\mathbb{I}$.
The subset $\{[X]\subseteq[x]\}$ may also be written as 
$
    \{\varphi(X_{1:m}) \subseteq [x]\}=
    \{\omega \in \Omega: \varphi(X_{1:m}(\omega)) \subseteq [x]\},
$
of which the 
image on $\mathbb{R}^m$ is 
$\{\varphi(x'_{1:m})\subseteq [x]\} = \{x'_{1:m} : \varphi(x'_{1:m})\subseteq [x]\}$, 
i.e. the subset of 
$\mathbb{R}^{m}$ containing those $x'_{1:m}$ that can generate an interval which is a 
subinterval of or equal to $[x]$.
The two subsets $\{[x^\prime]\subseteq[x]\}$ and $\{\varphi(x'_{1:m})\subseteq [x]\}$ 
are equivalent as their preimages on $\Omega$ are identical. As a result,
given a probability measure on $\mathbb{R}^{m}$, the probability
of $\{\varphi(x'_{1:m})\subseteq [x]\}$, and hence of $\{[x^{\prime}]\subseteq[x]\}$, 
can be calculated if only if it is measurable. 
This implies that in a measurable space of $\mathbb{I}$, $\{[x^{\prime}] \subseteq [x]\}$ 
should be measurable.

We construct the metric topology on $\mathbb{I}$, denoted by $\mathscr{T}_{\mathbb{I}}$,
induced by the Hausdorff metric, which specifies the distance between
elements $[a]$ and $[b]$ as 
\[
    d_{H}\left([a],[b]\right)=
    \max\left\{ |\underline{a}-\underline{b}|,|\overline{a}-\overline{b}|\right\},
\]
where $|\cdot|$ denotes absolute value.
If we consider the mapping $h([x])=(\underline{x},\overline{x})$ from 
$\mathbb{I}$ to $\mathbb{R}^{2}$,
then we have 
$d_{2}\left((\underline{a},\overline{a}),(\underline{b},\overline{b})\right)=
d_{H}\left([a],[b]\right)$ 
for any $[a],[b]\in\mathbb{I}$, where $d_{2}(\cdot)$ is the square metric on 
$\mathbb{R}^{2}$. 
That is, $h$ is a distance preserving map, or isometry, and hence 
$(\mathbb{I},\mathcal{\mathscr{T}_{\mathbb{I}}})$
is isometrically embedded into the metric topological space on $\mathbb{R}^{2}$ induced 
by $d_{2}(\cdot)$, which is also known as the standard topology. The standard topology on 
$\mathbb{R}^2$ is generated by the open rectangles \cite{Munkres:2000}.
This implies that $\mathscr{T}_{\mathbb{I}}$ inherits 
properties of the standard topology on $\mathbb{R}^{2}$, such as completeness, local
compactness and separability. See Section~\ref{sec:intvl-topology} for details.

Let $\mathcal{F}=\{\{[x^\prime] \subseteq [x]\} \colon [x] \in \mathbb{I}\}$ 
be the collection of subsets of interest. We can now construct a
measurable space involving $\mathcal{F}$ from the topology $\mathscr{T}_{\mathbb{I}}$.
Let $\mathscr{B}_{\mathbb{I}}$ be the smallest
$\sigma$-algebra containing all open subsets 
$\mathscr{B}_{\mathbb{I}}=\sigma(\mathscr{T}_{\mathbb{I}})$, i.e. 
the Borel $\sigma$-algebra on $\mathbb{I}$. The topology $\mathscr{T}_{\mathbb{I}}$ 
is the collection of all open subsets of $\mathbb{I}$, and the Borel $\sigma$-algebra is the 
smallest $\sigma$-algebra containing all open subsets \cite{Munkres:2000}.
This Borel $\sigma$-algebra
contains $\mathcal{F}$, as all elements of $\mathcal{F}$ are closures of some elements of 
$\mathscr{T}_{\mathbb{I}}$ (Section~\ref{sec:intvl-topology}).
The following lemma provides a stronger
result that $\mathcal{F}$ is sufficient to construct $\mathscr{B}_{\mathbb{I}}$.
\begin{lem}
\label{thm:sigma-algebra}The Borel $\sigma$-algebra on $\mathbb{I}$
is the smallest $\sigma$-algebra generated by $\mathcal{F}$, i.e. 
$\mathscr{B}_{\mathbb{I}}=\sigma(\mathcal{F})$. 
\end{lem}
This property indicates that $\mathscr{B}_{\mathbb{I}}$ is rich enough
to ensure that all elements in $\mathcal{F}$ are measurable.
It also suggests that if we only define a proper non-negative function on $\mathcal{F}$,
we can extend it to a measure on $(\mathbb{I},\mathscr{B}_{\mathbb{I}})$.
In particular, if the induced measure is a probability measure, it
would then be the distribution function of $[X]$.

Based on the isometry $h([x])=(\underline{x},\overline{x})$
between $\mathbb{I}$ and $\mathbb{R}^{2}$,
we now construct a measure on $(\mathbb{I},\mathscr{B}_{\mathbb{I}})$,
representing the uniform measure on $\mathbb{I}$, which gives equal weight to all intervals. 
Let the Borel $\sigma$-algebra on $\mathbb{R}^{2}$ be $\mathscr{B}_{\mathbb{R}^{2}}$, 
and $\mu\colon\mathscr{B}_{\mathbb{R}^{2}}\mapsto[0,+\infty)$
be the Lebesgue measure on $(\mathbb{R}^{2},\mathscr{B}_{\mathbb{R}^{2}})$. 
Due to the isometry $h([x])=(\underline{x},\overline{x})$, 
we then have that $\mu_{\mathbb{I}}=\mu\circ h$ is 
the uniform measure on $\left(\mathbb{I},\mathscr{B}_{\mathbb{I}}\right)$.
Consequently, the uniform measure of every Borel subset of 
$\mathbb{I}$ can be calculated via $\mu(\cdot)$
and $h(\cdot)$. Specifically, for every element of $\mathcal{F}$, we have 
\[
    \mu_{\mathbb{I}}(\{[x^{\prime}] \subseteq [x]\}) = 
        \mu(h(\{[x^{\prime}] \subseteq [x]\})) = \frac{1}{2}(\overline{x}-\underline{x})^{2},
\]
as 
$
    h(\{[x^{\prime}] \subseteq [x]\} ) = 
        \{(\underline{x}^{\prime},\overline{x}^{\prime}) \colon 
        \underline{x} \leq \underline{x}^{\prime} \leq 
        \overline{x}^{\prime} \leq \overline{x}\} 
$
is the region of an isosceles right triangle on the real plane. 
From Lemma~\ref{thm:sigma-algebra},
the uniform measure of all Borel subsets 
$E\in\mathscr{B}_{\mathbb{I}}$ is also available.
\begin{lem}
\label{thm:nbhd}
Define the infinitesimal neighbourhood of $[x]$ as 
\[
    \mathrm{d}[x] = \{[x^{\prime}] \in \mathbb{I} \mid
    \underline{x} - \mathrm{d}\underline{x} < \underline{x}^{\prime} \leq \underline{x} \leq 
    \overline{x} \leq \overline{x}^{\prime} < \overline{x} + \mathrm{d}\overline{x}\},
\]
where $\mathrm{d}\underline{x},\mathrm{d}\overline{x} > 0$. Its uniform
measure is 
$\mu_{\mathbb{I}}(\mathrm{d}\left[x]\right)=
\mathrm{d}\underline{x}\times\mathrm{d}\overline{x}$.
\end{lem}
From the above we note that $\mu_{\mathbb{I}}(\cdot)$ is a non-atomic measure,
i.e. $\mu_{\mathbb{I}}(\{[x]\}) = 0$, where $\{[x]\}$ is a set containing a single 
interval $[x]$.
Further, there is a convenient way to compute the value of $\mu_{\mathbb{I}}(\cdot)$
for any Borel subsets via the Lebesgue integration on $\mathbb{R}^{2}$.
Namely, for any subsets $E\in\mathscr{B}_{\mathbb{I}}$
\[
\mu_{\mathbb{I}}(E)=
\int_{E}\mu_{\mathbb{I}}\left(\mathrm{d}\left[x\right]\right)=
\iint_{h\left(E\right)}\mathrm{d}\underline{x}\mathrm{d}\overline{x}.
\]

Accordingly, through such isometry, the measurable space of intervals 
$(\mathbb{I},\mathscr{B}_{\mathbb{I}})$
inherits the convenient structure and properties of the real plane. These
results permit the construction of distribution and density functions of random
intervals.

\section{Topology\label{sec:intvl-topology}}

The basis of the standard topology on $\mathbb{R}^{2}$ is the collection
of all open rectangles. Its subspace topology induced by $\{(x,y) \colon x\leq y\}$,
as shown in Figure~\ref{fig:basis}, has the basis of which each element
is the remaining part of a open rectangle on the top-left half plane.
Therefore, the collection of their counterparts on $\mathbb{I}$ via the
isometry, $h([x])=(\underline{x},\overline{x})$, is the basis of $\mathscr{T}_{\mathbb{I}}$. 

The open subset of $\mathbb{I}$ corresponding to the rectangle (a) in Figure~\ref{fig:basis} is 
\begin{eqnarray*}
B([a],[b]) & = 
& \{[x]\colon\underline{b}<\underline{x}<\underline{a}\leq\overline{a}<\overline{x}<\overline{b}\} .
\end{eqnarray*}
This is the collection of all intervals for which the lower bounds are
bounded between $\underline{a}$ and $\underline{b}$, while the upper bounds
are bounded between $\overline{a}$ and $\overline{b}$. The open subset of $\mathbb{I}$
corresponding to the triangle (b) is 
\[
W([c])=\{[x]\colon\underline{c}<\underline{x}\leq\overline{x}<\overline{c}\} .
\]
This is the collection of all intervals for which the lower bounds are
greater than $\underline{c}$, while the upper bounds are smaller than
$\overline{c}$. 
\begin{figure}[t]
\centering
\includegraphics[scale=0.6]{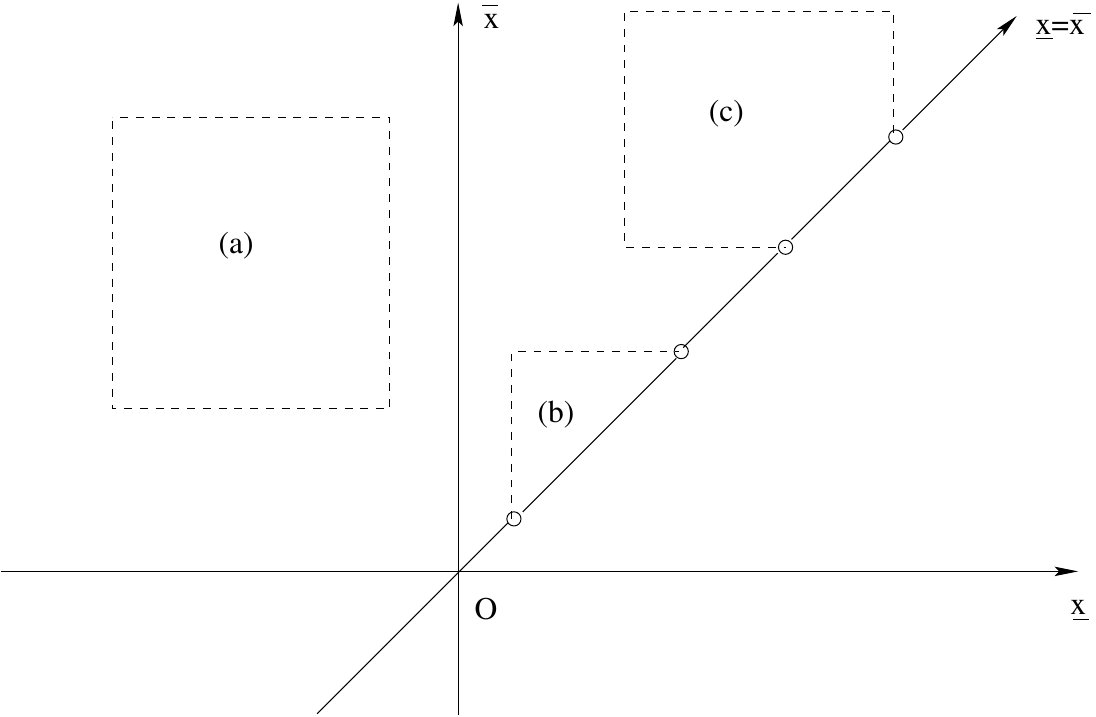}
\protect\caption{\label{fig:basis}\small$B([a],[b])$ and
$W\{ [a,b]\} $ are (a) and (b), respectively.
(a), (b) and (c) constitute the basis of $\mathscr{T}_{\mathbb{I}}$.}
\end{figure}
\begin{lem}
\label{lem:basis} 
Suppose that $\mathcal{E}$ is the collection of all $B([a],[b])$ and $W([c])$.
Then $\mathcal{E}$ is a basis for $\mathscr{T}_{\mathbb{I}}$. 
\end{lem}

\begin{lem}
\label{lem:w2b}
$B([a],[b]) = W([b]) \setminus
\left[\{[x]\subseteq[\underline{a},\overline{b}]\} \cup 
\{[x] \subseteq [\underline{b},\overline{a}]\}\right]$. 
\end{lem}

\begin{lem}
\label{lem:topology}
$\mathscr{T}_\mathbb{I}$ is the smallest topology
containing all $W([c])$ and $\comp{\{[x] \subseteq [c]\}}$. 
\end{lem}

\section{Hypercubes}
\label{sec:intvl-hypercubes}
Similarly, through the property of isometry, 
$h_p([\boldsymbol{x}]) = (\underline{x}_1,\overline{x}_1,\ldots,\underline{x}_p,\overline{x}_p)$, 
it can be shown
that a basis of the topology $\mathscr{T}_{\mathbb{I}^p}$ is the collection of the following
two classes of subsets:
\begin{eqnarray*}
B_p([\boldsymbol{a}],[\boldsymbol{b}]) & = & \{[\boldsymbol{x}] \colon 
\underline{b}_j < \underline{x}_j < \underline{a}_j \leq 
\overline{a}_j < \overline{x}_j < \overline{b}_j,\, j=1,\ldots,p\},\\
W_p([\boldsymbol{c}]) & = & \{[\boldsymbol{x}] \colon 
\underline{c}_j < \underline{x}_j \leq \overline{x}_j < \overline{c}_j,\, j=1,\ldots,p\}.
\end{eqnarray*}
The next lemma shows an analogous result of Lemma~\ref{lem:w2b}. 
\begin{lem}
\label{lem:w2b-hyper}
$B_p([\boldsymbol{a}],[\boldsymbol{b}]) = 
W_p([\boldsymbol{b}]) \setminus {\displaystyle\cup_{j=1}^p}
\left[\{[\boldsymbol{x}] \subseteq [\boldsymbol{a}_{j1}]\} \cup 
\{[\boldsymbol{x}] \subseteq [\boldsymbol{a}_{j2}]\}\right]$,
where 
\begin{eqnarray*}
[\boldsymbol{a}_{j1}] & = & \left([a_1],\ldots,[\underline{a}_j,\overline{b}_j],\ldots,[a_p]\right),\\
\left[\boldsymbol{a}_{j2}\right] & = 
& \left([a_1],\ldots,[\underline{b}_j,\overline{a}_j],\ldots,[a_p]\right).
\end{eqnarray*}
\end{lem}
Similar to the proof of Lemma~\ref{lem:w2b}.
Based on the above lemma, the hypercube's version of Lemma~\ref{thm:sigma-algebra}
can be proved in a similar way.

\section{Proofs}

\subsection{Proof of Lemma~\ref{thm:sigma-algebra}}

As an isometric embedding to the standard topology of the real plane, 
the topology $\mathscr{T}_\mathbb{I}$ is separable, and thus it has a countable basis.
We define rational intervals $[q]\in\mathbb{I}$ where
$\underline{q},\overline{q}$ are rational numbers. Then, the collection
of all rational intervals, $\mathbb{I}_\mathrm{Q}$, is dense in
$\mathbb{I}$. 

We first show that $\mathcal{E}_\mathrm{Q}$ is a countable basis of $\mathscr{T}_\mathbb{I}$.
Let $\mathcal{E}_\mathrm{Q}$ be the collection of all 
$B([q_1],[q_2])$ and $W([q])$.
As rational numbers are countable,  $\mathcal{E}_\mathrm{Q}$ is countable. 
It can be shown that $\mathcal{E}_\mathrm{Q}$
is a basis of a topology and its generated topology is $\mathscr{T}_\mathbb{I}$ in a similar 
way to Lemma~\ref{lem:basis}.
As a result, $\mathcal{E}_\mathrm{Q}$ is a countable basis of $\mathscr{T}_\mathbb{I}$.

Then we show that $\sigma(\mathcal{F})=\sigma(\mathcal{E}_\mathrm{Q})$.
For any $\{[x^\prime] \subseteq [x]\} \in \mathcal{F}$, 
$\{[x^\prime] \subseteq [x]\} =  \comp{W([x])}$ and $W([x]) \in \mathscr{T}_\mathbb{I}$ can be 
generated by set operations over countable elements from $\mathcal{E}_\mathrm{Q}$, 
as $\mathcal{E}_\mathrm{Q}$ is a countable basis of $\mathscr{T}_\mathbb{I}$. 
So, $\sigma(\mathcal{F}) \subseteq \sigma(\mathcal{E}_\mathrm{Q})$. 
On the other hand, for any $W([q]) \in \mathcal{E}_\mathrm{Q}$, we have  
$W([q]) = \cup_{n=k}^\infty \{[q^\prime] \subseteq [\underline{q}+1/n,\overline{q}-1/n]\}$,
where $\overline{q}-\underline{q} \geq 2/k$, and for any 
$B([q_1],[q_2]) \in \mathcal{E}_\mathrm{Q}$, we have 
$B([q_1],[q_2]) = W([q_2]) \setminus
\left[\{[q] \subseteq [\underline{q_1},\overline{q_2}]\} \cup 
\{[q] \subseteq [\underline{q_2},\overline{q_1}]\}\right]$ (Lemma~\ref{lem:w2b}).
So, $\sigma(\mathcal{E}_\mathrm{Q}) \subseteq \sigma(\mathcal{F})$.

That is, $\sigma(\mathcal{F})=\sigma(\mathcal{E}_\mathrm{Q})=\sigma(\mathscr{T}_\mathbb{I})$.

\subsection{Proof of Lemma~\ref{thm:nbhd}}

We let 
\begin{equation}
B^\star([a],[b])=\{[x] \colon 
\underline{b }< \underline{x} \leq \underline{a} \leq \overline{a} \leq \overline{x} < \overline{b}\}.
\label{eq:thm-nbhd-b}
\end{equation}
In a way analogous to Lemma~\ref{lem:w2b}, we have 
\begin{equation}
B^\star([a],[b])=
W([b]) \setminus \left[W([\underline{a},\overline{b}]) \cup W([\underline{b},\overline{a}])\right].
\label{eq:thm-nbhd-w2b}
\end{equation}
By the continuity of the measure, 
\begin{align*}
\mu_\mathbb{I}(W([x])) & = \mu_\mathbb{I}(\overset{\infty}{\cup}\{[x]^\prime \subseteq 
[\underline{x}+1/n,\overline{x}-1/n]\})\\
& = \lim_{n\to\infty}\mu_\mathbb{I}(\{[x]^\prime \subseteq 
[\underline{x}+1/n,\overline{x}-1/n]\})\\
& = \lim_{n\to\infty}\frac{1}{2}(\overline{x}-\underline{x}-2/n)^{2}
=\frac{1}{2}(\overline{x}-\underline{x})^{2}.
\end{align*}
Note that $W([\underline{a},\overline{b}])\cap W([\underline{b},\overline{a}])=W([a])$.
We have
\begin{align*}
\mu_\mathbb{I}(B^\star([a],[b]))
& = \mu_\mathbb{I}(W([b]))-\mu_\mathbb{I}(W([\underline{a},\overline{b}]))
-\mu_\mathbb{I}(W([\underline{b},\overline{a}]))+\mu_\mathbb{I}(W([a]))\\
& = (\underline{a}-\underline{b})(\overline{b}-\overline{a}).
\end{align*}
Therefore, 
$\mu_\mathbb{I}(\mathrm{d}[x]) = 
\mu_\mathbb{I}\left(B^\star([x],[\underline{x} - \mathrm{d}\underline{x},
\overline{x} + \mathrm{d}\overline{x}])\right)
= \mathrm{d}\underline{x} \times \mathrm{d}\overline{x}$.

\subsection{Proof of Lemma~\ref{lem:basis}}

We first show that $\mathcal{E}$ is a basis for a topology. 
Note that for any $[x] \in \mathbb{I}$, there exists at least one
$E \in \mathcal{E}$ s.t. $[x] \in E$. Then, we show in the following that 
for any $E_1,E_2 \in \mathcal{E}$, 
if $[x] \in E_1 \cap E_2$, then there exists $E_3 \in \mathcal{E}$ s.t. $[x] \in E_3$ and
$E_3 \subset E_1 \cap E_2$. Note that $\vee$ and $\wedge$ take the maximum and the minimum 
of two operands, respectively.
\begin{enumerate}
\item[i)] Consider $[x] \in B([a],[b]) \cap B([a^\prime],[b^\prime]) \neq \varnothing$.
Then $\underline{b} \vee \underline{b}^\prime < \underline{a} \wedge \underline{a}^\prime$
and $\overline{a} \vee \overline{a}^\prime < \overline{b} \wedge \overline{b}^\prime$.
% Otherwise, $B([a],[b])\cap B([a^\prime],[b^\prime])=\varnothing$. 
From $[x] \in B([a],[b])$, we have that
$\underline{b} < \underline{x} < \underline{a} \leq \overline{a} < \overline{x} < \overline{b}$.
From $[x] \in B([a^\prime],[b^\prime])$, we have that
$\underline{b}^\prime < \underline{x} < \underline{a}^\prime
\leq\overline{a}^\prime < \overline{x} < \overline{b}^\prime$.
Therefore, $\underline{b} \vee \underline{b}^\prime < \underline{x} 
< \underline{a} \wedge \underline{a}^\prime$
and $\overline{a} \vee \overline{a}^\prime < \overline{x} < \overline{b} \wedge \overline{b}^\prime$.
There exists $[a^{\prime\prime}],[b^{\prime\prime}] \in \mathbb{I}$ s.t. 
$\underline{b} \vee \underline{b}^\prime < \underline{b}^{\prime\prime} < \underline{x} < 
\underline{a}^{\prime\prime} < \underline{a} \wedge \underline{a}^\prime$ and 
$\overline{a} \vee \overline{a}^\prime < \overline{a}^{\prime\prime} < \overline{x} < 
\overline{b}^{\prime\prime} < \overline{b} \wedge \overline{b}^\prime$.
That is $[x] \in B([a^{\prime\prime}],[b^{\prime\prime}])$
and $B([a^{\prime\prime}],[b^{\prime\prime}])\subset B([a],[b]) \cap B([a^\prime],[b^\prime])$.

\item[ii)] Consider $[x] \in W([c_1]) \cap W([c_2]) \neq \varnothing$.
Then $\underline{c}_1 \vee \underline{c}_2 < \overline{c}_1 \wedge \overline{c}_2$.
From $[x] \in W([c_1])$, we have that 
$\underline{c}_1 < \underline{x} \leq \overline{x} < \overline{c}_1$.
From $[x] \in W([c_2])$, we have that
$\underline{c}_2 < \underline{x} \leq \overline{x} < \overline{c}_2$. 
Therefore, $\underline{c}_1 \vee \underline{c}_2 < \underline{x} 
\leq\overline{x} < \overline{c}_1 \wedge \overline{c}_2$.
There exists $[c] \in \mathbb{I}$ s.t. 
$\underline{c}_1 \vee \underline{c}_2 < \underline{c} < \underline{x} \leq
\overline{x} < \overline{c} < \overline{c}_1 \wedge \overline{c}_2$.
That is $[x] \in W([c])$ and $W([c]) \subset W([c_1]) \cap W([c_2])$.

\item[iii)] Consider $[x] \in B([a],[b]) \cap W([c]) \neq \varnothing$.
Then $\underline{c} < \underline{a}$ and $\overline{c} > \overline{a}$.
From $[x] \in B([a],[b])$, we have that 
$\underline{b} < \underline{x} < \underline{a} \leq \overline{a} < \overline{x} < \overline{b}$.
From $[x] \in W([c])$, we have that 
$\underline{c} < \underline{x} \leq \overline{x} < \overline{c}$.
Therefore, $\underline{c} \vee \underline{b} < \underline{x} < \underline{a} \leq 
\overline{a} < \overline{x} < \overline{c} \wedge \overline{b}$.
There exists $[a^\prime],[b^\prime] \in \mathbb{I}$
s.t. $\underline{c} \vee \underline{b} < \underline{b}^\prime < \underline{x} < 
\underline{a}^\prime < \underline{a}$
and $\overline{a} < \overline{a}^\prime < \overline{x} < \overline{b}^\prime < 
\overline{c} \wedge \overline{b}$.
That is $[x] \in B([a^\prime],[b^\prime])$
and $B([a^\prime],[b^\prime]) \subset B([a],[b]) \cap W([c])$.
\end{enumerate}
That is, $\mathcal{E}$ is a basis for a topology. Next, we show
$\mathcal{E}$ is a basis for $\mathscr{T}_\mathbb{I}$.
Figure~\ref{fig:basis} shows that the basis of $\mathscr{T}_\mathbb{I}$ 
consists of three types of subsets. As $B([a],[b])$
is an (a)-type subset and $W\{[c]\}$ is a (b)-type
subset, the topology generated by
$\mathcal{E}$ is coarser than $\mathscr{T}_\mathbb{I}$. 
On the other hand, for any $[x]$ in a (c)-type subset, we can find at
least one (a)-type subset or (b)-type subset that contains that $[x]$
and subsets of that (c)-type subset. Therefore, the topology generated
by $\mathcal{E}$ is finer than $\mathscr{T}_\mathbb{I}$. 
In conclusion, the topology generated by $\mathcal{E}$ is $\mathscr{T}_\mathbb{I}$.

\subsection{Proof of Lemma~\ref{lem:w2b}}

For any $[x] \in B([a],[b])$, i.e.
$\underline{b}<\underline{x}<\underline{a}\leq\overline{a}<\overline{x}<\overline{b}$,
we have $[x] \in W([b])$. 
Also $[x] \nsubseteq [\underline{a},\overline{b}]$ and 
$[x] \nsubseteq [\underline{b},\overline{a}]$, i.e
$[x] \notin \{[x] \subseteq [\underline{a},\overline{b}]\} 
\cup \{[x] \subseteq [\underline{b},\overline{a}]\}$.
Therefore, $B([a],[b]) \subseteq W([b]) \setminus 
\left[\{[x] \subseteq [\underline{a},\overline{b}]\} \cup 
\{[x] \subseteq [\underline{b},\overline{a}]\}\right]$.

On the other hand, for any $[x] \in W([b]) \setminus 
\left[\{[x] \subseteq [\underline{a},\overline{b}]\} \cup 
\{[x] \subseteq [\underline{b},\overline{a}]\}\right]$, 
we have $[x] \in W([b])$, i.e. $\underline{b} < \underline{x} \leq \overline{x} < \overline{b}$.
Also $[x] \nsubseteq [\underline{a},\overline{b}]$
and $[x] \nsubseteq [\underline{b},\overline{a}]$, i.e. $\underline{x}<\underline{a}$
and $\overline{x}>\overline{a}$.
Hence $\underline{b} < \underline{x} < \underline{a}$ and $\overline{a} < \overline{x} < \overline{b}$,
i.e. $[x] \in B([a],[b])$.
Therefore, $W([b]) \setminus \left[\{[x] \subseteq [\underline{a},\overline{b}]\} \cup 
\{[x] \subseteq [\underline{b},\overline{a}]\}\right] \subseteq B([a],[b])$.

In conclusion, 
$B([a],[b]) = W([b]) \setminus 
\left[\{[x]\subseteq[\underline{a},\overline{b}]\} \cup 
\{[x]\subseteq[\underline{b},\overline{a}]\}\right]$.

\subsection{Proof of Lemma~\ref{lem:topology}}

$\{[x] \subseteq[c] \} $ is a closed subset, as it is the 
closure of $W\{[c]\}$. 
Accordingly its complement $\comp{\{[x] \subseteq [c]\}}$
is open, and thus $\comp{\{[x] \subseteq [c]\} } \in \mathscr{T}_{\mathbb{I}}$. 
From Lemma~\ref{lem:w2b}, 
$B([a],[b]) = W([b]) \cap \comp{\{[x] \subseteq [\underline{a},\overline{b}]\}} \cap 
\comp{\{[x] \subseteq [\underline{b},\overline{a}]\}}$.
So every element in $\mathcal{E}$ can be generated by set operations over finite elements
of $W\{[c]\}$ and $\comp{\{[x] \subseteq [c]\}}$.
As $\mathcal{E}$ is a basis of $\mathscr{T}_\mathbb{I}$,
every element in $\mathscr{T}_{\mathbb{I}}$ can be generated by set operations over finite elements
of $W\{[c]\}$ and $\comp{\{[x] \subseteq [c]\}}$.
Therefore, $\mathscr{T}_\mathbb{I}$ is the smallest topology containing
$W([c])$ and $\comp{\{[x] \subseteq [c]\}}$.

\subsection{Proof of Theorem~1} %\ref{thm:pdf-intvl}}

For any function $f_{[X]}(\underline{x},\overline{x})$
satisfying the conditions in the theorem,
we can construct its containment distribution function 
$F_{[X]}(\underline{x},\overline{x})$ as
\[
F_{[X]}(\underline{x},\overline{x}) = 
\int_{\underline{x}}^{\overline{x}}\int_{\underline{x}}^{b}\!f_{[X]}(a,b)\,\mathrm{d}a\mathrm{d}b
\:\text{ or }\:
F_{[X]}(\underline{x},\overline{x}) = 
\int_{\underline{x}}^{\overline{x}}\int_{a}^{\overline{x}}\!f_{[X]}(a,b)\,\mathrm{d}b\mathrm{d}a.
\]
It is easy to check that $F_{[X]}(\underline{x},\overline{x})$
satisfies the conditions in Definition 1. %\ref{def:cdf-intvl}.

\subsection{Proof of Theorem~2} %\ref{thm:pdf-derivative}}

Let $C_{[X]}([x]) = F_{[X]}(\underline{x},\overline{x})$ be the containment functional.
From Theorem~3 %\ref{thm:cdf-intvl} 
and its proof, it determines 
a unique probability measure $P_{[X]} \colon \mathscr{B}_\mathbb{I} \mapsto [0,1]$ 
on the space of intervals s.t. $P_{[X]}([x]) = F_{[X]}(\underline{x},\overline{x})$.
As $\mathrm{d}[x]=B_\star([x],[\underline{x}-\mathrm{d}\underline{x},\overline{x}+
\mathrm{d}\overline{x}])$,
from (\ref{eq:thm-nbhd-b}) and (\ref{eq:thm-nbhd-w2b}),
\[
B_\star([x],[\underline{x}-\mathrm{d}\underline{x},\overline{x}+\mathrm{d}\overline{x}]) = 
W([\underline{x}-\mathrm{d}\underline{x},\overline{x}+\mathrm{d}\overline{x}]) \setminus 
\left[W([\underline{x},\overline{x}+\mathrm{d}\overline{x}]) \cup 
W([\underline{x}-\mathrm{d}\underline{x},\overline{x}])\right].
\]
Therefore, 
\begin{multline*}
P_{[X]}(\mathrm{d}[x])=
P_{[X]}(W([\underline{x}-\mathrm{d}\underline{x},\overline{x}+\mathrm{d}\overline{x}])) - 
P_{[X]}(W([\underline{x},\overline{x}+\mathrm{d}\overline{x}])) - \\
P_{[X]}(W([\underline{x}-\mathrm{d}\underline{x},\overline{x}])) + 
P_{[X]}(W([\underline{x},\overline{x}])).
\end{multline*}
By the continuity of the measure and 
$W([x]) = \cup_{n=k}^\infty\{[x]^\prime \subseteq [\underline{x} + 
\frac{1}{n},\overline{x} - \frac{1}{n}]\}$,
\[
P_{[X]}(W([x]))=\lim_{\underline{x}^\prime\to\underline{x}+}
\lim_{\overline{x}^\prime\to\overline{x}-}
P_{[X]}([X] \subseteq [x]^{\prime})=
\lim_{\underline{x}^\prime\to\underline{x}+}
\lim_{\overline{x}^\prime\to\overline{x}-}
F_{[X]}(\underline{x}^\prime,\overline{x}^\prime).
\]
As $F_{[X]}$ is twice differentiable (thus continuous),
$P_{[X]}(W([x])) = F_{[X]}(\underline{x},\overline{x})$. Therefore,
\[
P_{[X]}(\mathrm{d}[x])=F_{[X]}(\underline{x}-\mathrm{d}\underline{x},
\overline{x}+\mathrm{d}\overline{x}) - 
F_{[X]}(\underline{x},\overline{x}+\mathrm{d}\overline{x}) - 
F_{[X]}(\underline{x}-\mathrm{d}\underline{x},\overline{x}) + 
F_{[X]}(\underline{x},\overline{x}).
\]
Substituting second order Taylor expansions for the first three
terms in the above equation, we obtain
\[
P_{[X]}(\mathrm{d}[x]) = -\frac{\partial^{2}}{\partial\underline{x}\partial\overline{x}}
F_{[X]}(\underline{x},\overline{x})\mathrm{d}\underline{x}\mathrm{d}\overline{x} + 
o(\mathrm{d}\underline{x}\mathrm{d}\overline{x}).
\]
Note that 
$\mu_\mathbb{I}(\mathrm{d}[x])=\mathrm{d}\underline{x}\mathrm{d}\overline{x}$ 
(Theorem~\ref{thm:nbhd}), and so
\[
P_{[X]}(\mathrm{d}[x]) = -\frac{\partial^{2}}{\partial\underline{x}\partial\overline{x}}
F_{[X]}(\underline{x},\overline{x})\mu_\mathbb{I}(\mathrm{d}[x]) + 
o(\mu_\mathbb{I}(\mathrm{d}[x])).
\]
In addition, $P_{[X]}(\mathrm{d}[x]) = 0$
when $\mu_\mathbb{I}(\mathrm{d}[x]) = 0$, i.e. 
$P_{[X]}(\cdot)$ is absolute continuous w.r.t. $\mu_\mathbb{I}(\cdot)$. 
Therefore the the Radon-Nikodym derivative exists and
\[
\frac{P_{[X]}(\mathrm{d}[x])}{\mu_\mathbb{I}(\mathrm{d}[x])} = 
-\frac{\partial^2}{\partial\underline{x}\partial\overline{x}}F_{[X]}(\underline{x},\overline{x}).
\]

\subsection{Proof of Theorem~3} %\ref{thm:cdf-intvl}}

As $\mathscr{B}_\mathbb{I}=\sigma(\mathcal{F})$ (Lemma~\ref{thm:sigma-algebra}), 
any $E \in \mathscr{B}_{\mathbb{I}}$ can be generated by set operations over at most
countable elements from $\mathcal{F}$. 
So, it's probability measure $P([X] \in E)$ will be available if 
$P([X] \subseteq [x])$ is known for any $[x]$.
Therefore, the uniqueness has been proved.

Next, we prove the existence of a probability measure 
$P_{[X]} \colon \mathscr{B}_\mathbb{I} \mapsto [0,1]$ on the space of intervals s.t. 
$P_{[X]}(\{[x^\prime] \subseteq [x]\}) = C_{[X]}([x])$. 
Let $\mathcal{G}$ be the collection of all 
$B^\prime([x],[y]) = \{[x^\prime] \colon \underline{y} \leq \underline{x}^\prime 
< \underline{x} \leq \overline{x} < \overline{x}^\prime \leq \overline{y}\}$.
Similar to Lemma~\ref{lem:w2b}, we have 
$B^\prime([x],[y]) = \{[x]^\prime \subseteq [y]\} \setminus 
\left[\{[x]^\prime \subseteq [\underline{x},\overline{y}]\} 
\cup \{[x]^\prime \subseteq [\underline{y},\overline{x}]\}\right]$.
Then, define $\mathcal{H} = \mathcal{F} \cup \mathcal{G} \cup \{\varnothing,\mathbb{I}\}$, and
extend $C_{[X]}(\cdot)$ to a function $P_C(\cdot)$ on $\mathcal{H}$
s.t. $P_C(\varnothing) = 0$, $P_C(\mathbb{I}) = 1$, 
$P_C(\{[x^\prime] \subseteq [x]\} ) = C_{[X]}([x])$ and 
\[
P_C(B^\prime([x],[y])) = C_{[X]}([y]) - C_{[X]}([\underline{x},\overline{y}]) - 
C_{[X]}([\underline{y},\overline{x}]) + C_{[X]}([x]) \geq 0,
\]
by condition \emph{iii)} of 
the definition of $C_{[X]}(\cdot)$ in Section~2.2. %\ref{sec:cdf}.
That is $P_{C}(\cdot)$ is non-negative. 

In addition as $\mathbb{I}$ is locally compact, for any $A \subset \mathbb{I}$ and $\delta>0$,
there exists $E_1,\ldots,E_N \in \mathcal{H}$ 
with all $\mu_\mathbb{I}(E_i)\leq\delta$, such that $A\subset\cup_{i=1}^{N}E_{i}$.
Therefore, we can use Carath\'{e}odory construction \cite{Durrett:2010} to define a metric
outer measure. Let $P_{[X]}^\star(A)=\lim_{\delta\to0}P_\delta(A)$, where
\[
P_\delta(A) = \inf\left\{\sum_{i=1}^\infty P_C(E_i) \colon 
E_{i} \in \mathcal{H},\,\text{diam}(E_i) \leq \delta,\,\cup_{i=1}^\infty E_i \supseteq A\right\},
\]
where $\text{diam}(E_i)$ is the diameter of $E_i$.
So, $P_{[X]}^\star(\cdot)$ is a metric outer measure, and thus
the Borel subsets on $\mathbb{I}$ are measurable w.r.t. $P_{[X]}^\star(\cdot)$.
That is, there exists a probability
measure $P_{[X]} \colon \mathscr{B}_\mathbb{I} \mapsto [0,1]$,
such that $P_{[X]}(E)=P_{[X]}^\star(E)$ for any $E \in \mathscr{B}_\mathbb{I}$.

Finally, we can check that $P_{[X]}(\{[x^\prime] \subseteq [x]\})=C_{[X]}([x])$. 
For any $n=1,2,\ldots$, there exits 
$\delta_n\to0$ as $n\to\infty$, s.t. 
$P_{\delta_n}(\{[x^\prime] \subseteq [x]\}) \leq 
C_{[X]}([\underline{x}-\frac{1}{n},\overline{x}+\frac{1}{n}])$.
Also $P_\delta(\{[x^\prime] \subseteq [x]\}) \geq C_{[X]}([x])$ by definition for any $\delta > 0$. 
Therefore
\[
C_{[X]}([x]) \leq \lim_{n\to\infty}P_{\delta_n}(\{[x^\prime] \subseteq [x]\}) \leq 
\lim_{n\to\infty}C_{[X]}([\underline{x}-1/n,\overline{x}+1/n]).
\]
By condition \emph{ii)} of 
the definition of $C_{[X]}(\cdot)$ in Section~2.2, %\ref{sec:cdf}, 
\[
\lim_{n\to\infty}C_{[X]}([\underline{x}-1/n,\overline{x}+1/n])
=C_{[X]}(\cap_{n=1}^\infty[\underline{x}-1/n,\overline{x}+1/n])
=C_{[X]}([x]).
\]
Therefore, 
\[
P_{[X]}(\{[x^\prime] \subseteq [x]\} ) = 
\lim_{n\to\infty}P_{\delta_n}(\{[x^\prime] \subseteq [x]\}) = C_{[X]}([x]).
\]

As a result, given a random interval $[X] \colon \Omega \mapsto \mathbb{I}$, 
we obtain a probability measure $P \colon \sigma([X]) \mapsto [0,1]$, s.t.
$P([X] \subseteq [x]) = P_{[X]}(\{[x^\prime] \subseteq [x]\}) = C_{[X]}([x])$.

\subsection{Proof of Theorem~4} %\ref{thm:g2d-unif}}

Let $c=\frac{a+b}{2}\in(-\infty,+\infty)$ and $r=\frac{b-a}{2} \geq 0$. 
We can rewrite 
\[
f_{[X]}(\underline{x},\overline{x}|m)=\iint_{\{a \leq \underline{x}, b \geq \overline{x}\}}\!
m(m-1)\frac{(\overline{x}-\underline{x})^{m-2}}{(b-a)^{m}}\pi(a,b)\,\mathrm{d}a\mathrm{d}b.
\]
as
\[
f_{[X]}(\underline{x},\overline{x}|m) = 
2^{-m}m(m-1)(\overline{x} - \underline{x})^{m-2}\iint_A\!r^{-m}g(c,r)\,\mathrm{d}c\mathrm{d}r,
\]
where $g(c,r)=2\pi(c-r,c+r)$ is the density function of $(c,r)$ and 
$A=\{(c,r) \colon \overline{x} - r \leq c \leq \underline{x} + r,r 
\geq \frac{\overline{x}-\underline{x}}{2}\}$.
As $\pi(\cdot)$ is bounded continuously, $\int_{-\infty}^{+\infty}\!g(c,r)\,\mathrm{d}c < \infty$.
Let $g(r) = \int_{-\infty}^{+\infty}\!g(c,r)\,\mathrm{d}c$, $B_0=\{r \colon g(r) = 0\}$
and $B_1=\{r \colon g(r) \neq 0\}$.  
When $g(r) \neq 0$, we have $g(c | r) = \frac{g(c,r)}{g(r)}$.
The above integration can be decomposed into the following two
cases. In the case that $g(r) \neq 0$, we replace $g(c,r)$
with $g(r)g(c|r)$ and integrate out $c$,
\[
\iint_{A \cap B_1}\!r^{-m}g(c,r)\,\mathrm{d}c\mathrm{d}r=
\int_{\frac{\overline{x}-\underline{x}}{2}}^{\infty}\!r^{-m}g(r)
\left\{\int_{\overline{x}-r}^{\underline{x}+r}\!g(c|r)\,\mathrm{d}c\right\}\,\mathrm{d}r.
\]
In the case that $g(r)=0$, we have $g(c,r) = 0$ and
$
\iint_{A \cap B_0}\!r^{-m}g(c,r)\,\mathrm{d}c\mathrm{d}r=0.
$

Then, writing $z=(m-1)(\log r - \log\frac{\overline{x}-\underline{x}}{2})$, we have
\begin{multline*}
f_{[X]}(\underline{x},\overline{x}|m)=\frac{1}{2}m(\overline{x}-\underline{x})^{-1}\times\\
\int_0^\infty\!\mathrm{e}^{-z}g\left(\frac{\overline{x}-\underline{x}}{2}\mathrm{e}^{(m-1)^{-1}z}\right)
\left\{\int_{\overline{x}-\frac{\overline{x}-\underline{x}}{2}\mathrm{e}^{(m-1)^{-1}z}}
^{\underline{x}+\frac{\overline{x}-\underline{x}}{2}\mathrm{e}^{(m-1)^{-1}z}}\!
g\left(c|\frac{\overline{x}-\underline{x}}{2}\mathrm{e}^{(m-1)^{-1}z}\right)
\,\mathrm{d}c\right\}\,\mathrm{d}z.
\end{multline*}
As $\pi(\cdot)$ is bounded continuously, $g(c,r) = 2\pi(c-r,c+r)$ is bounded continuously.
Due to the mean value theorem, the above term can be simplified as
\[
f_{[X]}(\underline{x},\overline{x}|m) = \frac{1}{2}
\int_0^\infty\! m\left\{\mathrm{e}^{(m-1)^{-1}z} - 1\right\} 
\mathrm{e}^{-z}g(\xi,\frac{\overline{x}-\underline{x}}{2}\mathrm{e}^{(m-1)^{-1}z})\,\mathrm{d}z,
\]
where $\underline{x}+\frac{\overline{x}-\underline{x}}{2}\mathrm{e}^{(m-1)^{-1}z} \leq 
\xi \leq \overline{x}-\frac{\overline{x}-\underline{x}}{2}\mathrm{e}^{(m-1)^{-1}z}$. 
Let $M(\xi) = \sup_{z \geq 0}g(\xi,\frac{\overline{x}-\underline{x}}{2}\mathrm{e}^{(m-1)^{-1}z})$.
$M(\xi)$ is bounded as $g(c,r)$ is bounded.
When $m\geq3$, we have  
\[
f_{[X]}(\underline{x},\overline{x}|m) \leq \frac{M(\xi)}{2}
\int_0^\infty\!m\left\{\mathrm{e}^{(m-1)^{-1}z}-1\right\}\mathrm{e}^{-z}\,\mathrm{d}z = 
\frac{m}{2(m-2)}M(\xi) \leq \frac{3}{2}M(\xi).
\]
Therefore, $f_{[X]}(\underline{x},\overline{x}|m)$ is bounded when $m \to \infty$, and thus
\begin{align*}
\lim_{m\to\infty}f_{[X]}(\underline{x},\overline{x}|m) 
& = \frac{1}{2}\int_0^\infty\!\lim_{m\to\infty}m\left\{\mathrm{e}^{(m-1)^{-1}z}-1\right\}\mathrm{e}^{-z}
g\left(\xi,\frac{\overline{x}-\underline{x}}{2}\mathrm{e}^{(m-1)^{-1}z}\right)\,\mathrm{d}z\\
& =  \frac{1}{2}g\left(\frac{\underline{x}+\overline{x}}{2},\frac{\overline{x}-\underline{x}}{2}\right)
 = \pi(\underline{x},\overline{x}).
\end{align*}

\subsection{Proof of Theorem~5} %\ref{thm:g2d}}
\label{appthm7}

Let $f_{\mu,\tau}=f(\cdot|\mu,\tau)$, and denote $F_{\mu,\tau}=F(\cdot|\mu,\tau)$ 
and $Q_{\mu,\tau}=Q(\cdot;\mu,\tau)$ as its cumulative distribution function 
and quantile function, respectively. As $f_{\mu,\tau}$ is 
positive and continuous in the neighbourhoods of $Q_{\mu,\tau}(\underline{p})$ and 
$Q_{\mu,\tau}(\overline{p})$ with $\underline{p} > 0$ and $\overline{p} < 1$,
the joint density function of
\[
\begin{cases}
(m+1)^{\frac{1}{2}}f_{\mu,\tau}(Q_{\mu,\tau}(\underline{p}))
(\underline{X}-Q_{\mu,\tau}(\underline{p}))\\
(m+1)^{\frac{1}{2}}f_{\mu,\tau}(Q_{\mu,\tau}(\overline{p}))(\overline{X}-Q_{\mu,\tau}(\overline{p}))
\end{cases}
\]
converges pointwise to a bivariate normal density function, with zero mean and covariance matrix
\[
\Sigma = \begin{pmatrix}
\underline{p}(1-\underline{p}) & \underline{p}(1-\overline{p})\\
\underline{p}(1-\overline{p}) & \overline{p}(1-\overline{p})
\end{pmatrix}
\]
when $m\to\infty$ \cite{Reiss1989}. Thus when $m$ is large, 
the density function of the $\mathrm{i.i.d.}$ generative model 
\begin{multline*}
f_{[X]}^\star(\underline{x},\overline{x} | \theta,m,l,u)=
\frac{m!}{(l-1)!(u-l-1)!(m-u)!}[F(\underline{x}|\theta)]^{l-1}\\
\times[F(\overline{x}|\theta)-F(\underline{x}|\theta)]^{u-l-1}
[1-F(\overline{x}|\theta)]^{m-u}f(\underline{x}|\theta)f(\overline{x}|\theta),
\end{multline*}
is asymptotically equivalent to
\begin{equation*}
\frac{m+1}{2\pi|\Sigma|^\frac{1}{2}}
f_{\mu,\tau}(Q_{\mu,\tau}(\underline{p}))f_{\mu,\tau}(Q_{\mu,\tau}(\overline{p})) 
\exp\{-(m+1)T(\underline{x},\overline{x};\mu,\tau)\},
\end{equation*}
where
\begin{align*}
T(\underline{x},\overline{x};\mu,\tau) = 
& \frac{1}{2}(\underline{t}(\underline{x};\mu,\tau),\overline{t}(\overline{x};\mu,\tau)) 
\Sigma^{-1}(\underline{t}(\underline{x};\mu,\tau),\overline{t}(\overline{x};\mu,\tau))^\intercal,\\ 
\underline{t}(\underline{x};\mu,\tau) = 
& f_{\mu,\tau}(Q_{\mu,\tau}(\underline{p}))(\underline{x}-Q_{\mu,\tau}(\underline{p})),\\
\overline{t}(\overline{x};\mu,\tau) = 
& f_{\mu,\tau}(Q_{\mu,\tau}(\overline{p}))(\overline{x}-Q_{\mu,\tau}(\overline{p})).
\end{align*}
That is, the density function of the hierarchical generative model 
\[
f_{[X]}(\underline{x},\overline{x} | \alpha,m)=
\int\!f_{[X]}^\star(\underline{x},\overline{x} | \theta,m)
\pi(\theta|\alpha)\,\mathrm{d}\theta,
\]
is asymptotically equivalent to 
\begin{equation}
\frac{m+1}{2\pi|\Sigma|^\frac{1}{2}} \times H(\underline{x},\overline{x};\underline{p},\overline{p},m)
\label{eq:cdf-g-mix-order-approx}
\end{equation}
where
\begin{multline*}
H(\underline{x},\overline{x};\underline{p},\overline{p},m) = \\
\iint\!f_{\mu,\tau}(Q_{\mu,\tau}(\underline{p}))
f_{\mu,\tau}(Q_{\mu,\tau}(\overline{p}))\pi(\mu,\tau)
\exp\{-(m+1)T(\underline{x},\overline{x};\mu,\tau)\}\,\mathrm{d}\mu\mathrm{d}\tau.
\end{multline*}
Note that $\Sigma$ is positive definite, and so $T(\underline{x},\overline{x};\mu,\tau)\geq0$.
Also  $T(\underline{x},\overline{x};\underline{p},\overline{p},m)$ reaches
its minimum 0, when $Q_{\mu,\tau}(\underline{p}) = \underline{x}$ and 
$Q_{\mu,\tau}(\overline{p}) = \overline{x}$.
As $f_{\mu,\tau}$ is interval-identifiable, the system of equations, 
$Q_{\mu,\tau}(\underline{p}) = \underline{x}$ and 
$Q_{\mu,\tau}(\overline{p}) = \overline{x}$, 
has a unique solution, and thus $T(\underline{x},\overline{x};\underline{p},\overline{p},m)$ 
is unimodal.

As $\mu_\star = \mu(\underline{x},\overline{x};\underline{p},\overline{p})$
and $\tau_\star = \tau(\underline{x},\overline{x};\underline{p},\overline{p})$
are the solution of 
$Q_{\mu,\tau}(\underline{p}) = \underline{x}$ and 
$Q_{\mu,\tau}(\overline{p}) = \overline{x}$,
given conditions \emph{i)} and \emph{iii)} in the theorem, 
a Laplace approximation can be applied to  $H(\underline{x},\overline{x};\underline{p},\overline{p},m)$
at the point $(\mu_\star,\tau_\star)$, giving
\begin{multline}
H(\underline{x},\overline{x};\underline{p},\overline{p},m) \approx\\
2\pi(m+1)^{-1}|\nabla^{2}T(\underline{x},\overline{x};\mu_\star,\tau_\star)|^{-\frac{1}{2}}
f_{\mu_\star,\tau_\star}(\underline{x})f_{\mu_\star,\tau_\star}(\overline{x})\pi(\mu_\star,\tau_\star).
\label{eq:laplace-approx}
\end{multline}

We let $T = T(\underline{x},\overline{x};\mu,\tau)$, 
$\underline{t} = \underline{t}(\underline{x};\mu,\tau)$, 
$\overline{t} = \overline{t}(\overline{x};\mu,\tau)$
and $\Sigma^{-1} = \begin{pmatrix} a_{11} & a_{12}\\ a_{12} & a_{22} \end{pmatrix}$, 
so we have $T=\frac{1}{2}(a_{11}\underline{t}^2+2a_{12}\underline{t}\overline{t}+a_{22}\overline{t}^2)$.
The first order partial derivatives of $T$ are
\begin{eqnarray*}
\frac{\partial T}{\partial\mu} & = & a_{11}\underline{t}\frac{\partial\underline{t}}{\partial\mu} + 
a_{12}\overline{t}\frac{\partial\underline{t}}{\partial\mu} + 
a_{12}\underline{t}\frac{\partial\overline{t}}{\partial\mu} + 
a_{22}\overline{t}\frac{\partial\overline{t}}{\partial\mu},\\
\frac{\partial T}{\partial\tau} & = & a_{11}\underline{t}\frac{\partial \underline{t}}{\partial\tau} + 
a_{12}\overline{t}\frac{\partial\underline{t}}{\partial\tau} + 
a_{12}\underline{t}\frac{\partial\overline{t}}{\partial\tau} + 
a_{22}\overline{t}\frac{\partial\overline{t}}{\partial\tau}.
\end{eqnarray*}
Let $T^\star$, $\underline{t}^\star$ and $\overline{t}^\star$ denote the corresponding
functions and their derivatives taking values at $(\mu_\star,\tau_\star)$.
As $\underline{t}^{\star}=\overline{t}^{\star}=0$, the second order partial derivatives
at $(\mu_\star,\tau_\star)$ are
\begin{eqnarray*}
\frac{\partial^2T^\star}{\partial\mu^2} & = 
& a_{11}\left(\frac{\partial\underline{t}^\star}{\partial\mu}\right)^2 + 
2a_{12}\frac{\partial\overline{t}^\star}{\partial\mu}\frac{\partial\underline{t}^\star}{\partial\mu} + 
a_{22}\left(\frac{\partial\overline{t}^\star}{\partial\mu}\right)^2,\\
\frac{\partial^2T^\star}{\partial\tau^{2}} & = 
& a_{11}\left(\frac{\partial \underline{t}^\star}{\partial\tau}\right)^2 + 
2a_{12}\frac{\partial\overline{t}^\star}{\partial\tau}
\frac{\partial\underline{t}^\star}{\partial\tau} + 
a_{22}\left(\frac{\partial\overline{t}^\star}{\partial\tau}\right)^2,\\
\frac{\partial^{2}T^{\star}}{\partial\mu\partial r} & = 
& a_{11}\frac{\partial\underline{t}^\star}{\partial\mu}
\frac{\partial\underline{t}^\star}{\partial\tau} + 
a_{12}\frac{\partial\underline{t}^\star}{\partial\mu}\frac{\partial\overline{t}^\star}{\partial\tau} + 
a_{12}\frac{\partial\underline{t}^\star}{\partial\tau}\frac{\partial\overline{t}^\star}{\partial\mu} + 
a_{22}\frac{\partial\overline{t}^\star}{\partial\mu}\frac{\partial\overline{t}^\star}{\partial\tau}.
\end{eqnarray*}
Therefore, $\nabla^{2}T$ at $(\mu_\star,\tau_\star)$ is
\[
\nabla^{2}T^\star = \begin{pmatrix} 
\frac{\partial\underline{t}^\star}{\partial\mu} & \frac{\partial\underline{t}^\star}{\partial\tau}\\
\frac{\partial\overline{t}^\star}{\partial\mu} & \frac{\partial\overline{t}^\star}{\partial\tau}
\end{pmatrix}^{\intercal}
\Sigma^{-1}
\begin{pmatrix}
\frac{\partial\underline{t}^\star}{\partial\mu} & \frac{\partial\underline{t}^\star}{\partial\tau}\\
\frac{\partial\overline{t}^\star}{\partial\mu} & \frac{\partial\overline{t}^\star}{\partial\tau}
\end{pmatrix},
\]
and its determinant is 
$|\nabla^{2}g| = |\Sigma|^{-1}\left(
\frac{\partial\underline{t}^\star}{\partial\mu}\frac{\partial\overline{t}^\star}{\partial\tau} - 
\frac{\partial\underline{t}^\star}{\partial\tau}
\frac{\partial\overline{t}^\star}{\partial\mu}\right)^2$.

The derivatives of $\underline{t}$ and $\overline{t}$ at $(\mu^\star,\tau^\star)$ are
\begin{eqnarray*}
\frac{\partial\underline{t}^\star}{\partial\mu} & = 
& -f_{\mu^\star,\tau^\star}(\underline{x}) \times 
\frac\partial{\partial\mu}Q_{\mu^\star,\tau^\star}(\underline{p}),\\
\frac{\partial\underline{t}^\star}{\partial\tau} & = 
& -f_{\mu^\star,\tau^\star}(\underline{x}) \times 
\frac{\partial}{\partial\tau}Q_{\mu^\star,\tau^\star}(\underline{p}),\\
\frac{\partial\overline{t}^\star}{\partial\mu} & = 
& -f_{\mu^\star,\tau^\star}(\overline{x}) \times 
\frac{\partial}{\partial\mu}Q_{\mu^\star,\tau^\star}(\overline{p}),\\
\frac{\partial\overline{t}^\star}{\partial\tau} & = 
& -f_{\mu^\star,\tau^\star}(\overline{x}) \times 
\frac{\partial}{\partial\tau}Q_{\mu^\star,\tau^\star}(\overline{p}).
\end{eqnarray*}
and thus
\begin{equation}
|\nabla^{2}T^\star| = 
|\Sigma|^{-1}f_{\mu^\star,\tau^\star}(\underline{x})^2f_{\mu^\star,\tau^\star}(\overline{x})^2
\left|J(\mu^\star,\tau^\star;\underline{p},\overline{p})\right|^2,
\label{eq:laplace-det}
\end{equation}
where 
\[
J(\mu^\star,\tau^\star;\underline{p},\overline{p}) = 
\begin{pmatrix}
\frac{\partial}{\partial\mu}Q_{\mu^\star,\tau^\star}(\underline{p}) 
& \frac{\partial}{\partial\tau}Q_{\mu^\star,\tau^\star}(\underline{p})\\
\frac{\partial}{\partial\mu}Q_{\mu^\star,\tau^\star}(\overline{p}) 
& \frac{\partial}{\partial\tau}Q_{\mu^\star,\tau^\star}(\overline{p})
\end{pmatrix}.
\]

From (\ref{eq:laplace-approx}) and (\ref{eq:laplace-det}), 
we obtain that the density function of the hierarchical generative model 
(\ref{eq:cdf-g-mix-order-approx})
converges pointwise to 
$\pi(\mu^\star,\tau^\star)\left|J(\mu^\star,\tau^\star;\underline{p},\overline{p})\right|^{-1}$.

\subsection{\label{app:8+9}Proof of Theorems~6 and 7} %\ref{thm:pdf-derivative-hyper}~and~\ref{thm:pdf-hyper}}

Similar to the proof of Theorems 1 and 2. %~\ref{thm:pdf-derivative}~and~\ref{thm:pdf-intvl}.
Use Lemma\ \ref{lem:w2b-hyper} and Taylor expansions.

\subsection{\label{sub:like-bi}Likelihood function of two dimension i.i.d. generative model}

Let $[\boldsymbol{X}] = [X_1] \times [X_2]$
be the random rectangle generated from $m$ $\mathrm{i.i.d.}$ bivariate latent
data points from $f(x_1,x_2|\boldsymbol{\theta})$
with the data aggregation function taking the minimum and maximum values at each margin. 
Let $F(x_1,x_2|\boldsymbol{\theta})$
be the distribution function of $f(x_1,x_2|\boldsymbol{\theta})$.
The distribution function of $[X_1]\times[X_2]$ is
\begin{multline*}
F_{[\boldsymbol{X}]}(\underline{x}_1,\overline{x}_1,
\underline{x}_2,\overline{x}_2|\boldsymbol{\theta}) = \\
\left[F(\overline{x}_1,\overline{x}_2|\boldsymbol{\theta}) - 
F(\underline{x}_1,\overline{x}_2|\boldsymbol{\theta}) - 
F(\overline{x}_1,\underline{x}_2|\boldsymbol{\theta}) + 
F(\underline{x}_1,\underline{x}_2|\boldsymbol{\theta})\right]^m.
\end{multline*}
This is the probability that all $m$ latent data points fall within the
rectangle $[x_1]\times[x_2]$. From Theorem~7, %\ref{thm:pdf-derivative-hyper}, 
the likelihood function
is the fourth order mixed derivative as shown below
\begin{align*}
 & f_{[\boldsymbol{X}]}(\underline{x}_1,\overline{x}_1,
 \underline{x}_2,\overline{x}_2|\boldsymbol{\theta}) = m(m-1)(m-2)(m-3)\times\\
& \left\{F(\overline{x}_1,\overline{x}_2|\boldsymbol{\theta}) - 
F(\underline{x}_1,\overline{x}_2|\boldsymbol{\theta}) - 
F(\overline{x}_1,\underline{x}_2|\boldsymbol{\theta}) + 
F(\underline{x}_1,\underline{x}_2|\boldsymbol{\theta})\right\}^{m-4} \times\\ 
& \int_{\underline{x}_1}^{\overline{x}_1}\!f(y_1,\underline{x}_2|\boldsymbol{\theta})\,\mathrm{d}y_1
\int_{\underline{x}_1}^{\overline{x}_1}\!f(y_2,\overline{x}_2|\boldsymbol{\theta})\,\mathrm{d}y_2
\int_{\underline{x}_2}^{\overline{x}_2}\!f(\underline{x}_1,y_3|\boldsymbol{\theta})\,\mathrm{d}y_3
\int_{\underline{x}_2}^{\overline{x}_2}\!f(\overline{x}_1,y_4|\boldsymbol{\theta})\,\mathrm{d}y_4 + \\ 
& m(m-1)(m-2)\left\{F(\overline{x}_1,\overline{x}_2|\boldsymbol{\theta}) - 
F(\underline{x}_1,\overline{x}_2|\boldsymbol{\theta}) - 
F(\overline{x}_1,\underline{x}_2|\boldsymbol{\theta}) + 
F(\underline{x}_1,\underline{x}_2|\boldsymbol{\theta})\right\}^{m-3} \times \\
& \bigg\{f(\underline{x}_1,\underline{x}_2|\boldsymbol{\theta})
\int_{\underline{x}_1}^{\overline{x}_1}\!f(y_2,\overline{x}_2|\boldsymbol{\theta})\,\mathrm{d}y_2
\int_{\underline{x}_2}^{\overline{x}_2}\!f(\overline{x}_1,y_4|\boldsymbol{\theta})\,\mathrm{d}y_4 + \\
& f(\underline{x}_1,\overline{x}_2|\boldsymbol{\theta})
\int_{\underline{x}_1}^{\overline{x}_1}\!f(y_1,\underline{x}_2|\boldsymbol{\theta})\,\mathrm{d}y_1
\int_{\underline{x}_2}^{\overline{x}_2}\!f(\overline{x}_1,y_4|\boldsymbol{\theta})\,\mathrm{d}y_4 + \\
& f(\overline{x}_1,\underline{x}_2|\boldsymbol{\theta})
\int_{\underline{x}_1}^{\overline{x}_1}\!f(y_2,\overline{x}_2|\boldsymbol{\theta})\,\mathrm{d}y_2
\int_{\underline{x}_2}^{\overline{x}_2}\!f(\underline{x}_1,y_3|\boldsymbol{\theta})\,\mathrm{d}y_3 + \\
& f(\overline{x}_1,\overline{x}_2|\boldsymbol{\theta})
\int_{\underline{x}_1}^{\overline{x}_1}\!f(y_1,\underline{x}_2|\boldsymbol{\theta})\,\mathrm{d}y_1
\int_{\underline{x}_2}^{\overline{x}_2}\!
f(\underline{x}_1,y_3|\boldsymbol{\theta})\,\mathrm{d}y_3\bigg\} + \\
& m(m-1)\left\{F(\overline{x}_1,\overline{x}_2|\boldsymbol{\theta}) - 
F(\underline{x}_1,\overline{x}_2|\boldsymbol{\theta}) - 
F(\overline{x}_1,\underline{x}_2|\boldsymbol{\theta}) + 
F(\underline{x}_1,\underline{x}_2|\boldsymbol{\theta})\right\}^{m-2} \times \\
& \bigg\{f(\underline{x}_1,\underline{x}_2|\boldsymbol{\theta})
f(\overline{x}_1,\overline{x}_2|\boldsymbol{\theta}) + 
f(\underline{x}_1,\overline{x}_2|\boldsymbol{\theta})
f(\overline{x}_1,\underline{x}_2|\boldsymbol{\theta})\bigg\}.
\end{align*}
Although it is rather complex, in fact it has a similar intuitive interpretation
to (7). %(\ref{eq:lik-g-mm-iid}). 
The first term denotes the case that $m-4$
points fall within $[x_1]\times[x_2]$ while the
remaining four points are $(y_1,\underline{x}_2)$,
$(y_2,\overline{x}_2)$, $(\underline{x}_1,y_3)$
and $(\overline{x}_1,y_4)$, where $\underline{x}_1\leq y_1,\,y_2\leq\overline{x}_1$
and $\underline{x}_2\leq y_3,\,y_4\leq\overline{x}_2$. The
second term represents the case that $m-3$ points fall within
$[x_1]\times[x_2]$ while the remaining
three points determine the boundary of the rectangle. The last terms is
the case the boundary is formed by only two points. 

\subsection{\label{sub:robustness} Additional plots from simulation study}
\begin{landscape}
\begin{figure}
\center
$
\begin{array}{cccc}
\includegraphics[width=0.33\textwidth]{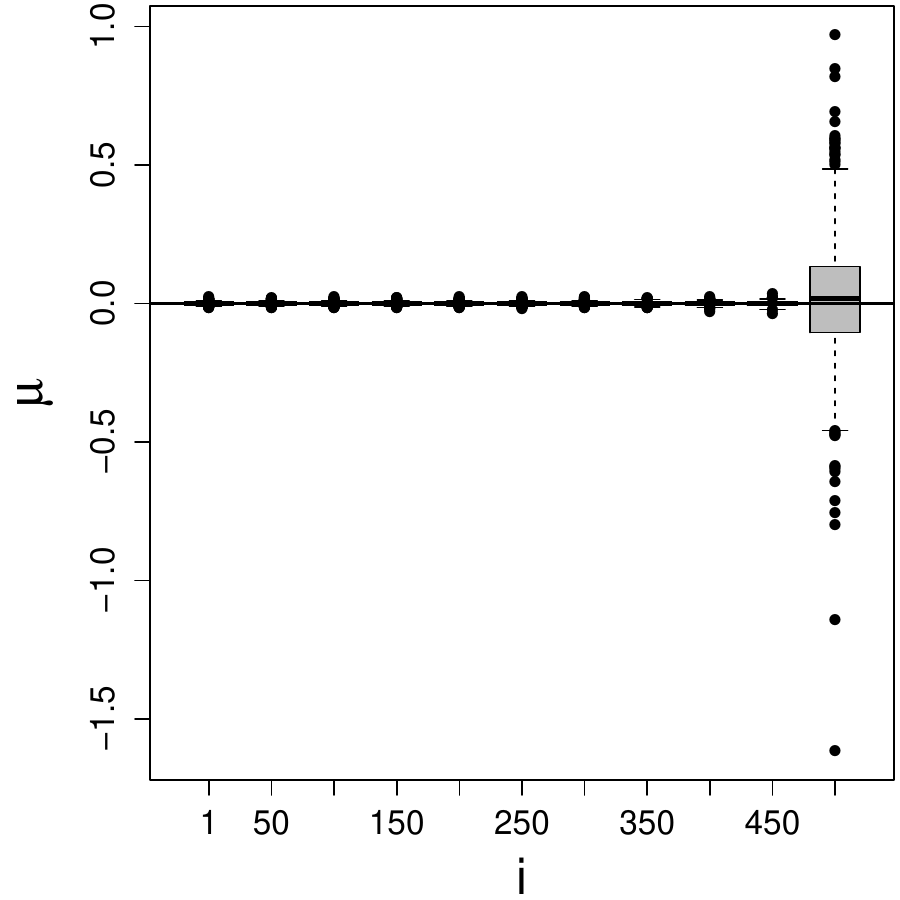} &
\includegraphics[width=0.33\textwidth]{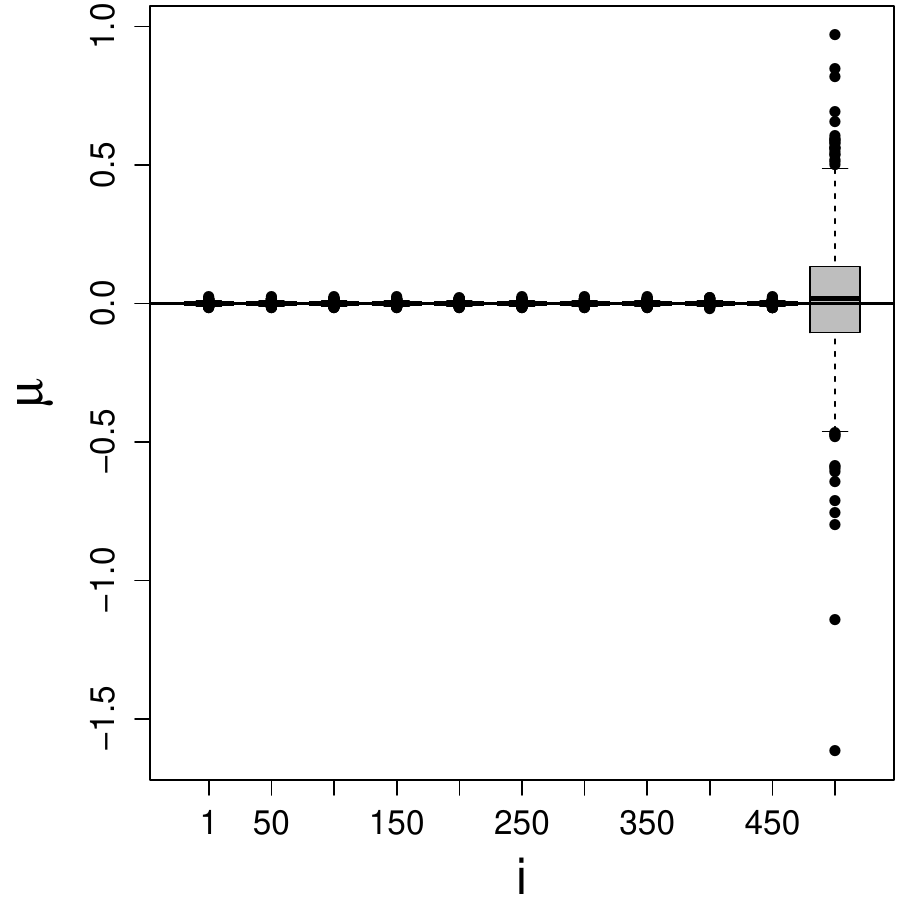} &
\includegraphics[width=0.33\textwidth]{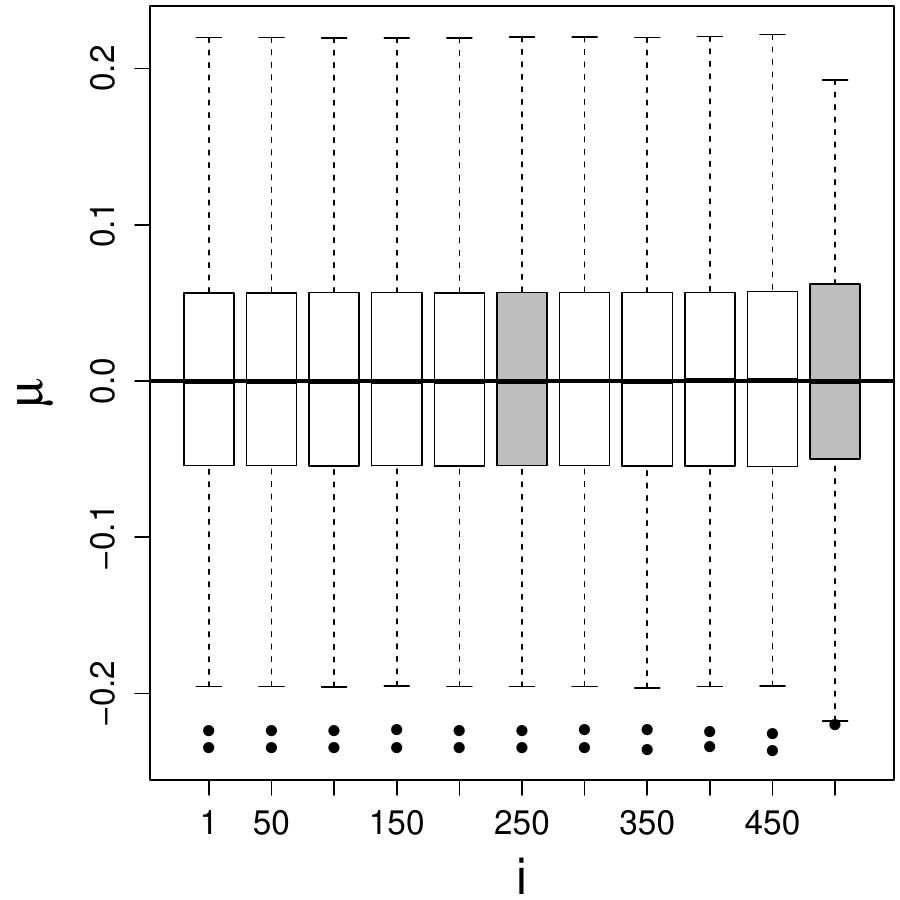} &
\includegraphics[width=0.33\textwidth]{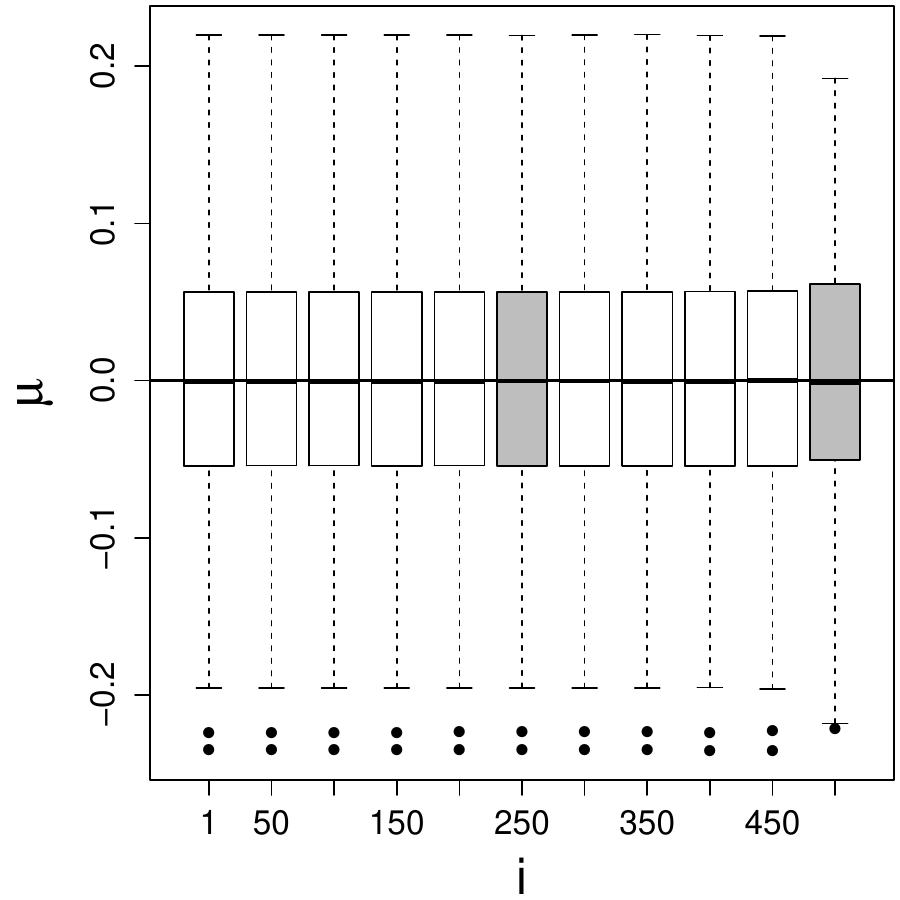} 
\\
\includegraphics[width=0.33\textwidth]{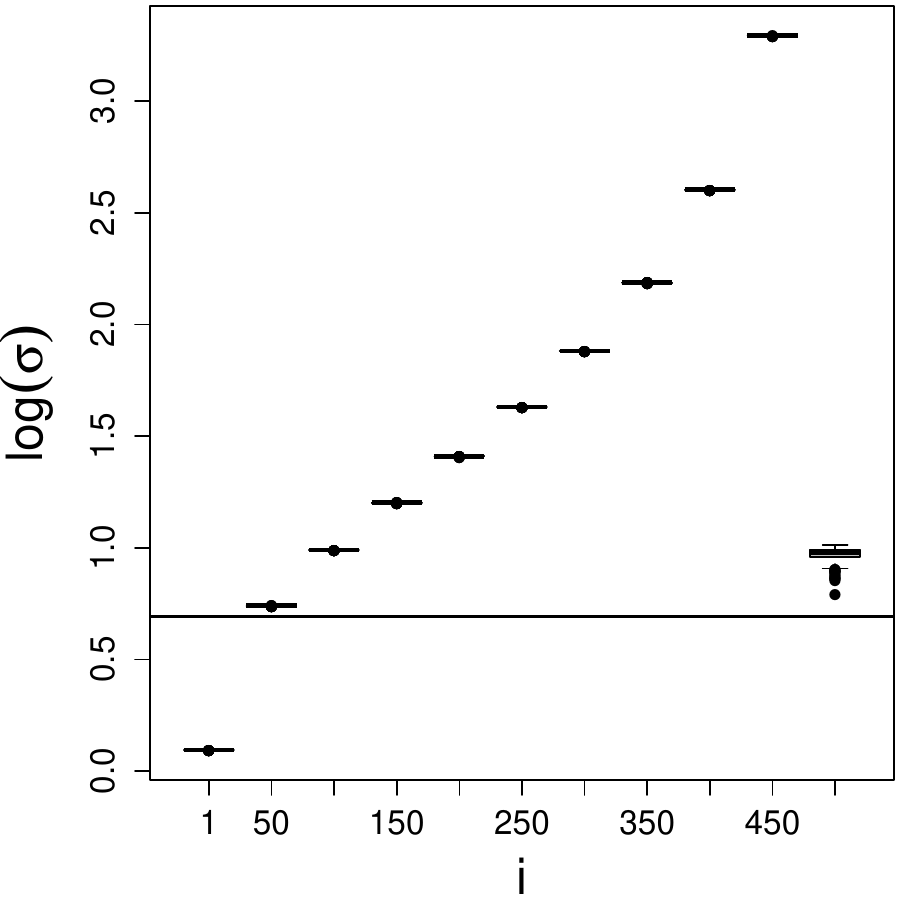} &
\includegraphics[width=0.33\textwidth]{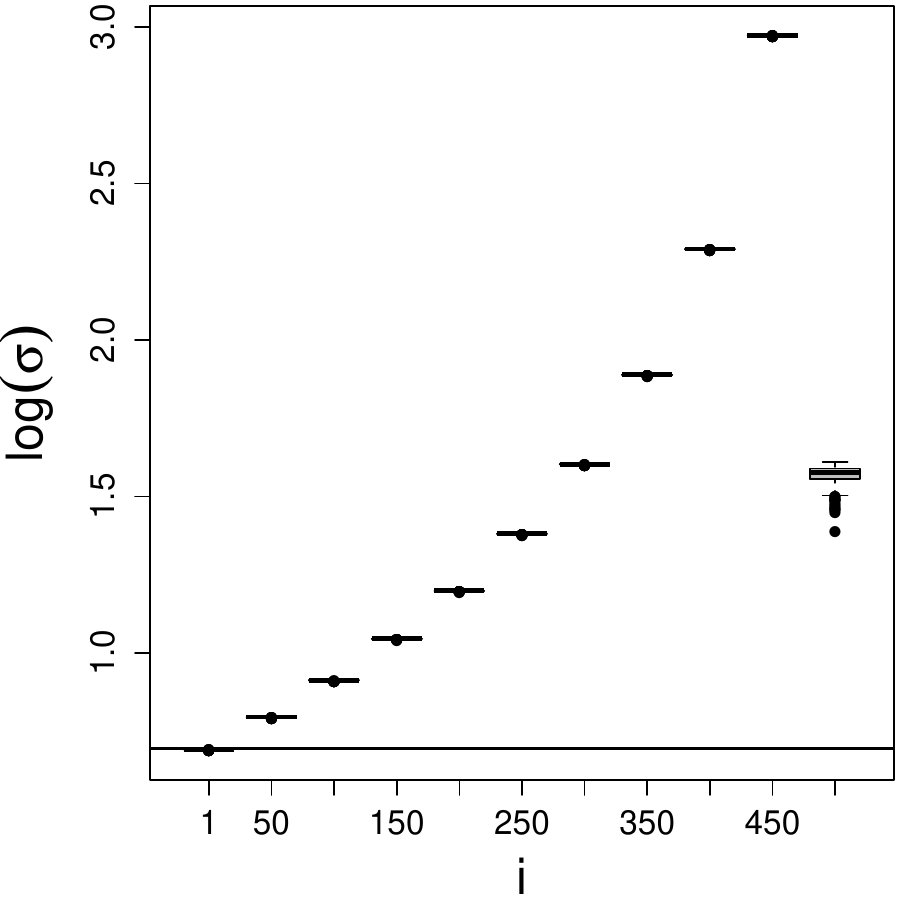} &
\includegraphics[width=0.33\textwidth]{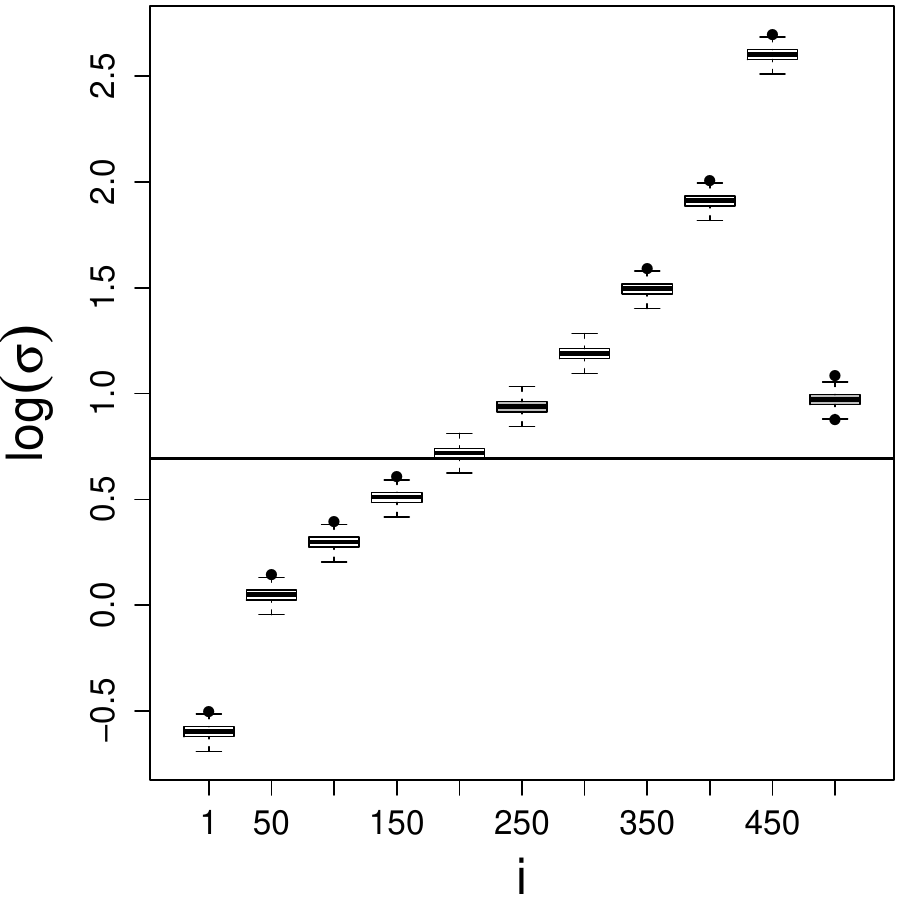} &
\includegraphics[width=0.33\textwidth]{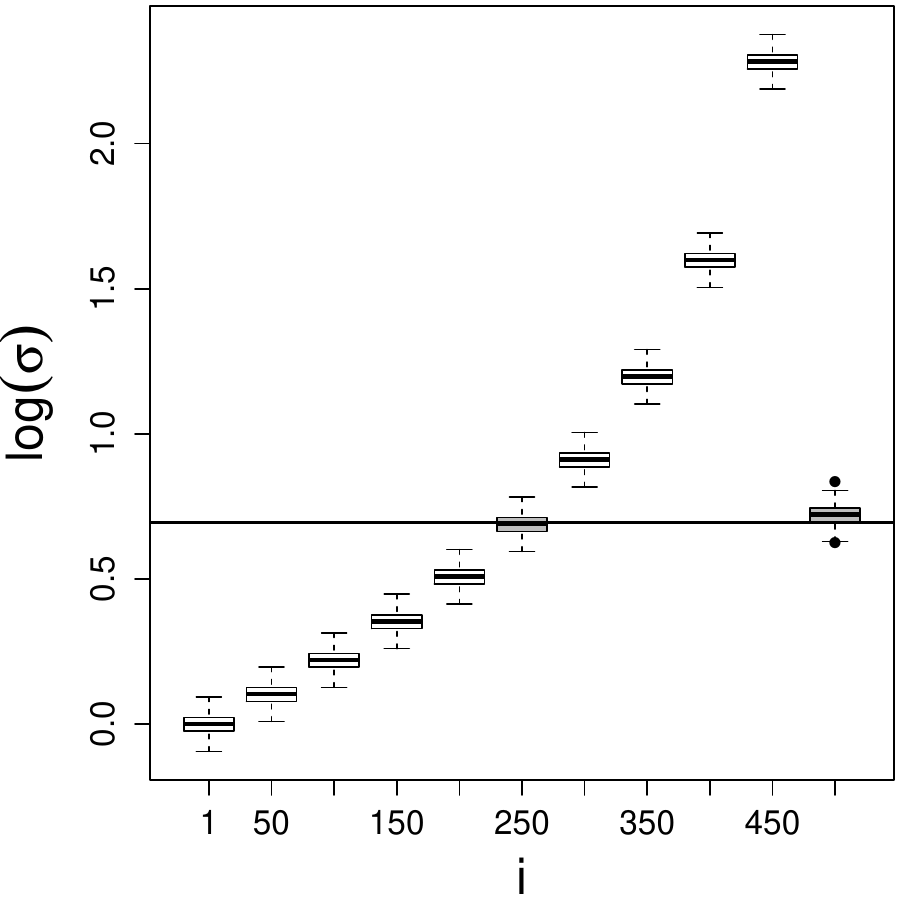} 
\\
\includegraphics[width=0.33\textwidth]{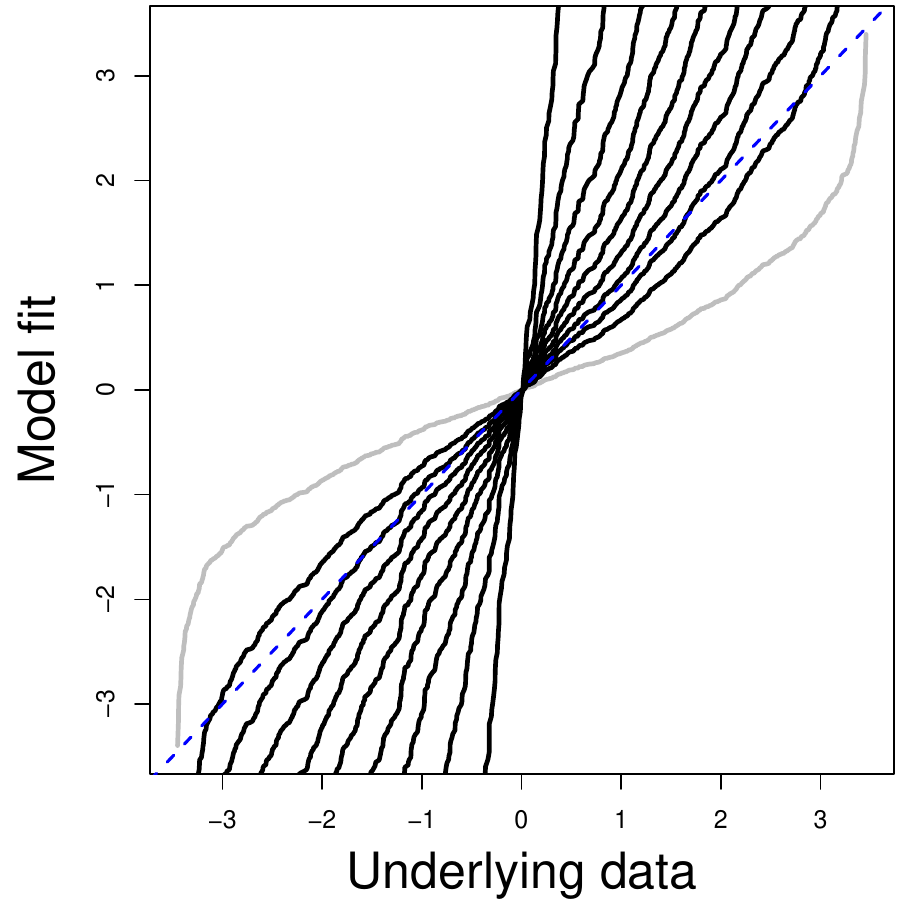} &
\includegraphics[width=0.33\textwidth]{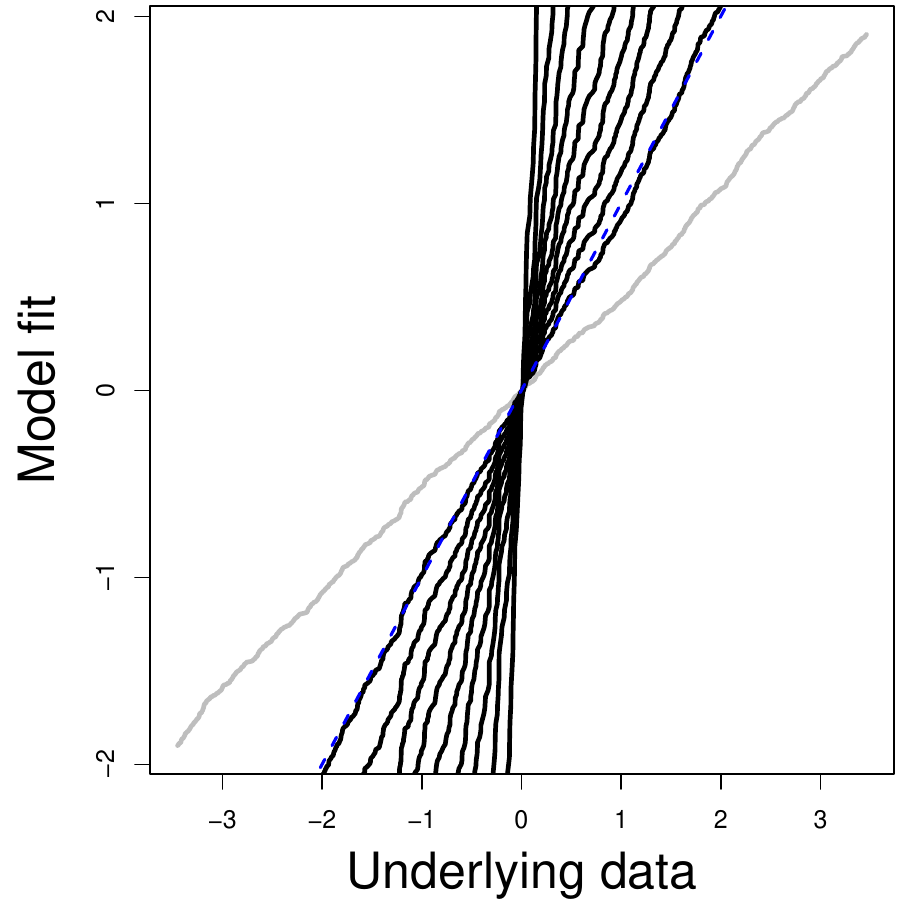} &
\includegraphics[width=0.33\textwidth]{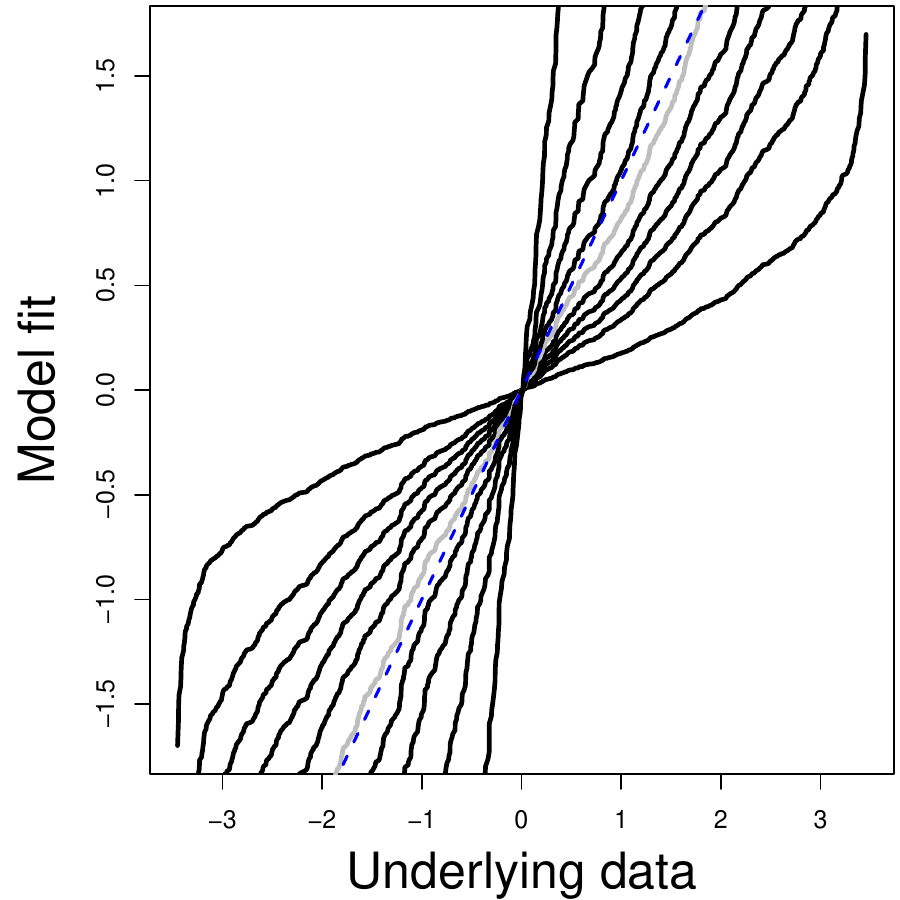} &
\includegraphics[width=0.33\textwidth]{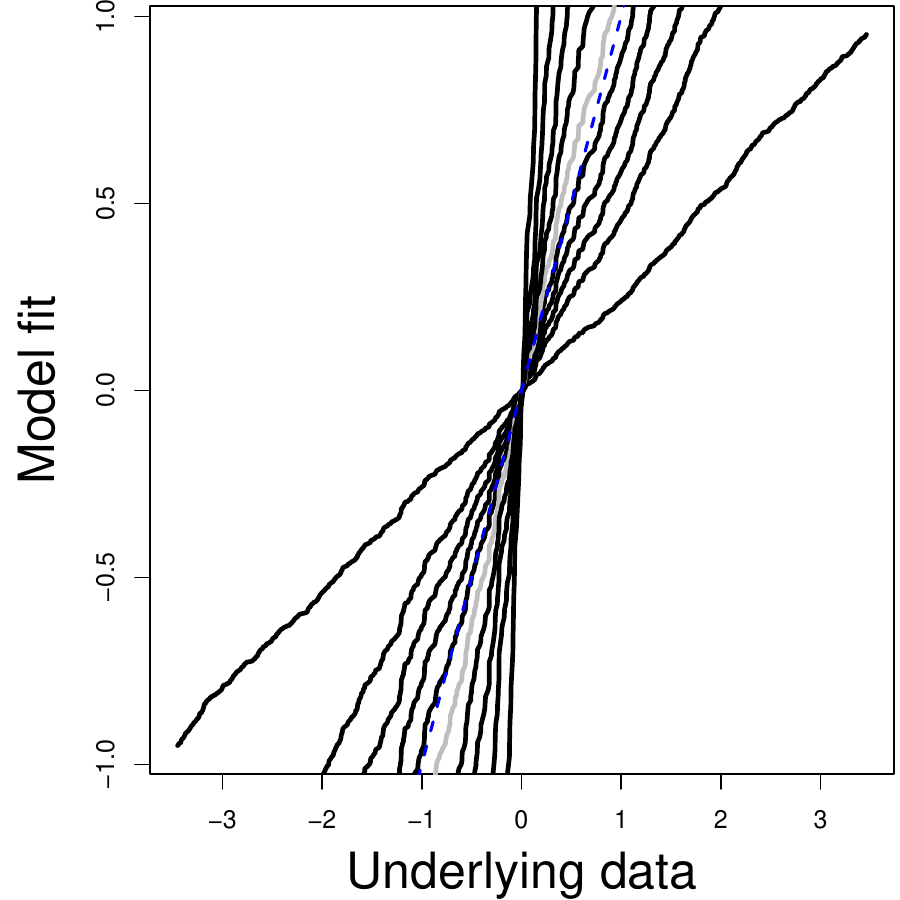} 
\end{array}
$
\caption{\small . 
{\bf As for Figure 6 in the Main Text, except that that the true data generating process is uniform.}
Boxplots of 500 replicate maximum likelihood estimates of $\mu$ and $\log\sigma$ under a uniform distribution with mean $\mu=0$ and standard deviation $\sigma=2$ as the true data generating process with $m=1000$, and assuming data aggregation function $\phi_i$, $i=1,50, 100, \ldots, 450$. The true aggregation functions are $\phi_1$ (left two columns) and $\phi_{250}$ (right two  columns). The models fitted are the normal (columns 1 \& 3) and uniform (columns 2 \& 4) distributions. 
In each panel, the rightmost boxplot  indicates the outcome using the dataset with  $5\%$ outliers.  
Bottom row shows quantile-quantile curves of the fitted model ($y$-axis) versus the empirical underlying data quantiles ($x$-axis). 
Grey curves indicate use of the correct $\varphi(\cdot)$ function. Dashed line denotes $y=x$.
}
\label{fig:robustness}
\end{figure}
\end{landscape}

%\section*{References}

\bibliography{intvllik}

\end{document}